\documentclass{amsart}

\usepackage{ifthen}
\newif\ifdraftmode
\draftmodetrue

\usepackage{times,enumerate}
  
\usepackage{latexsym,amssymb,amscd,amsthm,amsmath,verbatim}
\usepackage[all]{xy}
\usepackage{setspace}
\usepackage{tikz}
\usepackage{hyperref}
\usepackage{relsize,nccmath} 

\newcommand\interior[1]{{(#1)_0}}

\usepackage{xcolor}
\usepackage{mathrsfs} 

\usepackage[normalem]{ulem}
\renewcommand\sout{\bgroup\markoverwith
{\textcolor{red}{\rule[0.7ex]{3pt}{1.4pt}}}\ULon}

\newcommand{\maA}{\mathcal A}
\newcommand{\maB}{\mathcal B}
\newcommand{\maC}{\mathcal C}
\newcommand{\maD}{\mathcal D}
\newcommand{\maE}{\mathcal E}
\newcommand{\maF}{\mathcal F}

\newcommand{\maH}{\mathcal H}

\newcommand{\maL}{\mathcal L}

\newcommand{\maP}{\mathcal P}

\newcommand{\maS}{\mathcal S}

\newcommand{\sS}{\mathscr{S}}

\newcommand\bl[1]{ \{ #1 \} }
\newcommand{\psbmanifold}{p-sub\-mani\-fold}
\newcommand{\psbmanifolds}{p-sub\-mani\-folds}
\newcommand{\bsbmanifold}{b-sub\-mani\-fold}
\newcommand{\bsbmanifolds}{b-sub\-mani\-folds}
\newcommand{\dsbmanifold}{d-sub\-mani\-fold}
\newcommand{\dsbmanifolds}{d-sub\-mani\-folds}
\newcommand{\wibsbmanifold}{wib-sub\-mani\-fold}
\newcommand{\wibsbmanifolds}{wib-sub\-mani\-folds}
\newcommand{\Melrb}{b}

\newcommand{\Msbmanifold}{\nmagenta{submanifold in the sense of Definition~\ref{def.submanifold-gen}}}   
\newcommand{\Msbmanifolds}{\nmagenta{submanifolds in the sense of Definition~\ref{def.submanifold-gen}}}  

\newcommand{\wsbmanifold}{weak submanifold}
\newcommand{\wsbmanifolds}{weak submanifolds}
\newcommand\XGV{X_{\mathrm{GV}}}

\definecolor{lightgreen}{rgb}{.9,1,.9}
\newcommand\lightgreen[1]{\textcolor{lightgreen}{#1}}
\definecolor{darkgreen}{rgb}{.2,.6,.2}
\definecolor{darkblue}{rgb}{.3,.3,1}

\newcommand\blue[1]{\textcolor{blue}{#1}}

\newcommand\red[1]{\textcolor{red}{#1}}

\newcommand\green[1]{\textcolor{green}{#1}}

\newcommand\nmagenta[1]{{#1}}

\ifdraftmode
\newcommand\bacomment[2][B2All]{\red{#1: #2}}
\else
\newcommand\bacomment[2][B]{}
\renewcommand\sout[1]{}
\fi

\newcommand\hidden[1]{}

\newcommand\sgn{\mathop{\mathrm{sgn}}}
\newcommand\medcup{\mathsmaller{\bigcup}}
\newcommand\unionF{\mathchoice{\medcup\,\maF}{\medmath{\textstyle \bigcup}\,\maF}{\medmath{\scriptstyle\bigcup}\maF}{\medmath{\scriptscriptstyle\bigcup}\maF}}
\newcommand\unionP{\mathchoice{\medcup\,\maP}{\medmath{\textstyle \bigcup}\,\maP}{\medmath{\scriptstyle\bigcup}\maP}{\medmath{\scriptscriptstyle\bigcup}\maP}}

\DeclareMathOperator{\Spec}{Spec}

\newcommand\ovra\overrightarrow
\newcommand\can{\mathop{\mathrm{can}}} 
\newcommand\GL{\mathop{\mathrm{GL}}} 
\newcommand\ev{\mathop{\mathrm{ev}}} 

\newcommand\arxiv[1]{\href{https://www.arxiv.org/abs/#1}{arXiv:~#1}}
\newcommand\Arxiv[1]{\href{https://www.arxiv.org/abs/#1}{ArXiv:~#1}}

\newtheorem{theorem}{Theorem}[section]
\newtheorem{lemma}[theorem]{Lemma}
\newtheorem{proposition}[theorem]{Proposition}
\newtheorem{corollary}[theorem]{Corollary}

\theoremstyle{definition} 
\newtheorem{definition}[theorem]{Definition}

\theoremstyle{definition} 
\newtheorem{remark}[theorem]{Remark}
\newtheorem{remarks}[theorem]{Remarks}
\newtheorem{example}[theorem]{Example}
\newtheorem{examples}[theorem]{Examples}

\newcommand{\dist}{\operatorname{dist}}
\newcommand{\id}{\operatorname{\mathrm{id}}}
\newcommand{\depth}{\operatorname{depth}}

\newcommand{\pa}{\partial}

\renewcommand{\SS}{\mathbb{S}}
\newcommand{\CC}{\mathbb C}
\newcommand{\FF}{\mathbb F}

\newcommand{\RR}{\mathbb R}

\newcommand{\ZZ}{\mathbb Z}

\newcommand{\CI}{{\mathcal C}^{\infty}}

\newcommand\<{\langle}
\renewcommand\>{\rangle} 
\newcommand\ede{\ :=\ }
\newcommand\seq{\ =\ }

\def\mfkA{\mathfrak A}

\newcommand{\oX}{\overline{X}}
\newcommand{\oXY}{\overline{X/Y}}

\newcommand\mEF[1]{\maE_{\maF}(#1)}

\let\setminus\smallsetminus

\newcommand\eg{\emph{e.\thinspace g.,\ }\ignorespaces}
\newcommand\ie{\emph{i.\thinspace e.,\ }\ignorespaces}
\newcommand\iet{\emph{i.\thinspace e.,~}\ignorespaces}

\let\oldtocsection=\tocsection
\let\oldtocsubsection=\tocsubsection
\let\oldtocsubsubsection=\tocsubsubsection
\renewcommand{\tocsection}[2]{\hspace{0em}\oldtocsection{#1}{#2}}
\renewcommand{\tocsubsection}[2]{\hspace{2em}\oldtocsubsection{#1}{#2}}
\renewcommand{\tocsubsubsection}[2]{\hspace{3em}\oldtocsubsubsection{#1}{#2}}

\newcommand\eff{\mathrm{eff}}

\author[B. Ammann]{Bernd Ammann} \address{B. A., Fakult\"at f\"ur
  Mathematik, Universit\"at Regensburg, 93040 Regensburg, Germany}
\email{bernd.ammann@mathematik.uni-regensburg.de}

\author[J. Mougel]{J\'{e}r\'{e}my Mougel} \address{J. M.,
  Mathematisches Institut Georg-August-Universit\"at G\"ottingen,
  37083 G\"ottingen, Germany} \email{jeremy.mougel@uni-goettingen.de}

\author[V. Nistor]{Victor Nistor}\address{V. N., Universit\'{e} de
  Lorraine, CNRS, IECL, F-57000 Metz, France
and Inst. Math. Romanian Acad.  PO BOX 1-764, 014700 Bucharest
Romania} \email{victor.nistor@univ-lorraine.fr}

\thanks{B.A. has been partially supported by SPP 2026 (Geometry at
  infinity) and the SFB 1085 (Higher Invariants), both funded by the
  DFG (German Science Foundation). J.M. and V.N. have been partially
  supported by ANR-14-CE25-0012-01 (SINGSTAR) funded by ANR (French
  Science Foundation).\\
}

\begin{document}

\title[$N$-body spaces]{A comparisons of the Georgescu and Vasy spaces
  associated to the $N$-body problems and applications}

\begin{abstract}
We provide new insight into the analysis of $N$-body problems by studying a 
compactification $M_N$ of $\mathbb{R}^{3N}$ that is compatible with the analytic 
properties of the $N$-body Hamiltonian~$H_N$. We show that our compactification 
coincides with a compactification introduced by Vasy using blow-ups in order to 
study the scattering theory of $N$-body Hamiltonians and with a compactification 
introduced by Georgescu using $C^*$-algebras. In particular, the compactifications
introduced by Georgescu and by Vasy coincide (up to a homeomorphism
that is the identity on $\mathbb{R}^{3N}$). Our result has applications to the spectral theory 
of $N$-body problems and to some 
related approximation properties. For instance, results about the essential spectrum, 
the resolvents, and the scattering matrices of~$H_N$ (when they exist)
may be related to the behavior 
near $M_N\setminus \mathbb{R}^{3N}$ (i.e.\ ``at infinity'') of their distribution kernels, 
which can be efficiently studied using our methods. The compactification~$M_N$ 
is compatible with the action of the permutation group $S_N$, which allows to implement 
bosonic and fermionic (anti-)symmetry relations. We also indicate how our results lead 
to a regularity result for the eigenfunctions of $H_N$.
\end{abstract}

\maketitle

\setcounter{page}{1}

\tableofcontents

\section{Introduction}

\subsection{A general introduction and motivation for our work}

The quantum behavior of an atomic system is often investigated
via its associated Hamiltonian. A good model for $N$ non-relativistic
particles interacting with each other by Coulomb type forces is given by the 
Hamiltonian
\begin{equation}\label{eq.def.HN}
   (H_N u)(x) \ede \Big (  - \sum_{j=1}^N \frac{1}{2m_j} \Delta_{x_j} 
   + \sum_{1\leq j<k\leq N} \frac{b_{jk}}{|x_j - x_k|} \Big ) u(x)\,,
\end{equation}
where $x_j\in \RR^3$ describes the position of the $j$-th particle, $x = (x_1, x_2, 
\ldots, x_N) \in \RR^{3N}$, the operator $\Delta_{x_j}$ is the 
Laplacian with respect to $x_j$, $m_j\in \RR_+$, and $b_{jk} \in \RR$.
See, for instance, \cite{GeorgescuBookNew, DerezinskiAnnals}. 
As usual, by moving to the center of mass coordinates and effective operators, 
to an atom with $N-1$ electrons (corresponding to $N$ ``bodies'')
will correspond an operator $H_{N-1}^{\eff}$ acting on functions
defined on $\RR^{3(N-1)}$.

The way the mathematical properties of Hamiltonians are reflected in
the properties of the physical system was explained in many works,
including \cite{GeorgescuBookNew, DerezinskiAnnals, DerGer1, HunzikerSigal, 
LiebSeiringerbook, ReedSimon4}. In 
particular, the mathematical study of the operator $H_N$ (and of its simplified 
version $H_{N-1}^{\eff}$) is a very vast domain
of study in quantum mechanics and in mathematics. We will not be able to 
do justice to all the people who have contributed to the field, but let us nevertheless
mention some works that are among the closest to the methods of this paper, namely
the monographs of Amrein, Boutet de Monvel, and Georgescu
\cite{GeorgescuBookNew}, Derezi\'{n}ski and G\'{e}rard \cite{DerGer2} 
and Teschl \cite{Teschlbook}, as well as the research papers
\cite{Breteaux, DaGe04, DerezinskiAnnals, DerezinskiFaupin, FultonMacPh-Compact, 
Georgescu2018, GeIf06, MelroseEucSpace}.
More specific references even closer related to our work can be found below.

The mathematical study of $H_N$ and $H_{N-1}^{\eff}$ is quite challenging, especially for $N>2$.
The simplest case is that of hydrogen type atoms, which corresponds to $N=2$ and 
$m_1\gg m_2$. Then
\begin{equation}\label{eq.def.H1eff}
   H^{\eff}_1 u(x) \ede \Big (  - \frac{1}{2\mu} \Delta 
   + \frac{b}{|x|} \Big ) u(x)\,,\quad \mu=\frac{m_1m_2}{m_1+m_2}.
\end{equation}
In order to understand the mathematical properties of this operator, one usually writes  
the Hamiltonian $H_1^{\eff}$ in spherical coordinates 
\begin{equation*}
    (r, y) \in (0, \infty) \times \SS^{2}\,,\quad r=|x|,\quad x=ry\,,
\end{equation*} 
where $\SS^{n-1}$ denotes the unit sphere in $\RR^n$, as usual.
The use of spherical coordinates has led, for instance, to the determination of the spectrum 
of $H_1^{\eff}$ and to explicit formulas for its eigenfunctions (see \cite{DerGer2, 
DerezinskiExact,Teschlbook,WeylBook} for details and historical references),
which is the basis for the orbital model in (quantum) chemistry. 
For $N>2$ explicit calculations seem to be impossible. Nevertheless,  
one can still try to find ``more convenient coordinates'' in which to perform
our calculations than the usual, euclidean coordinates.
In this vein, one of the main results of this paper is to provide
convenient coordinates that generalize the polar coordinates and which are helpful
to study both particle interactions at infinity and the regularity of eigenfunctions 
for $N > 2$ particles.

More precisely, our ``more convenient coordinates'' (in the case of
$N$-particles and the effective Hamiltonian) patch together to yield a compact (smooth) 
manifold with corners $M_{N-1}$, whose interior 
$M_{N-1} \smallsetminus \pa M_{N-1}$ is $\RR^{3(N-1)}$. Thus $M_{N-1}$ is a {\em compactification} of 
$\RR^{3(N-1)}$ in the usual mathematical sense. 
In turn, the construction of such a compactification 
$M_{N-1}$ may possibly yield more convenient coordinate systems via its natural coordinate 
charts. For the hydrogen atom, $N = 2$ and $M_1$ is the \emph{radial compactification} 
$M_1 \ede \overline{\RR^3} = \RR^{3} \cup \SS_{\RR^3}$
(see Subsection~\ref{ssec.sph.comp} for the definition of the radial compactification).
In the earlier literature,  two compactifications of $\RR^{3(N-1)}$ for $N > 1$ have played an especially important role in the study of the $N$-body problem, a role which will be explained now.

To do that, it will be convenient to place ourselves in a slightly
more general setting in which we consider a finite semilattice $\maF$ of linear subspaces a finite
dimensional, real vector space $X$. In this setting, a first such compactification is Georgescu's compactification, which was obtained as the primitive ideal spectrum $\Spec(\mEF X)$
of a certain commutative $C^*$-algebra  \cite{BMBG} $\mEF X$ \cite{Georgescu2018,GeIf06,GN}. 
Another compactification, $[X: \SS_{\maF}$] was constructed by Vasy using iterated blow-ups 
\cite{VasyAsympt}. See also \cite{Kottke-Lin}. The constructions of these compactifications 
and the notation will be explained shortly in Subsection~\ref{subsec.vasy-blowup}.
One of the {\em main results} of this paper is to show that both Georgescu's compactification 
and Vasy's compactification are naturally homeomorphic with the one introduced in \cite{MNP}
(see Equation  \eqref{eq.def.maES}). This common compactification will be called
the ``Georgescu-Vasy'' space and will be denoted $X_{GV}$. In case $\maF$ is the semilattice
generated by the collision planes of the $N$-body problem for the effective Hamiltonian,
then the Georgescu-Vasy compactification of $\RR^{3(N-1)}$ will be the 
space $M_{N-1}$ that provides the more convenient coordinates were looking for.
The identification of Georgescu's and Vasy's compactifications 
for the $N$-body problem with the same space~$M_{N-1}$ 
will allow us to obtain further properties for the
corresponding Hamiltonian since, as we will explain below, each of the 
constructions of $M_{N-1}$ mentioned above has its own advantages. Of course, the 
role of the compactifications in the work of 
these authors -- as well as in ours -- is to use the the properties of the space $M_{N-1}$ to obtain 
a better insight into the properties of the Hamiltonian $H_{N-1}^{\eff}$ (or of related Hamiltonians). 
In this spirit, we present some applications of our results in Section~\ref{sec7}. These applications are
also summarized at the end of this introduction.

\subsection{Our setting and our construction of the compactification $M_N$}
Our results will be mostly for operators that are somewhat different from the 
$N$-body Hamiltonians $H_N$ or $H_{N-1}^{\eff}$ of Equations~\eqref{eq.def.HN}
and \eqref{eq.def.H1eff}.
The operators that we study are, in fact, in many regards, more general than the 
$N$-body Hamiltonians and, in any case, they retain most of the main features of the 
operators $H_N$ and $H_{N-1}^{\eff}$ that are relevant 
to this paper. To describe the class of operators that we will study, let us thus 
first explain the following two customary modifications of the operators $H_N$, 
following, for instance,~\cite{GeorgescuBookNew, DerezinskiAnnals, Georgescu2018} 
and the references therein.

The first modification is to ``smooth out'' 
the singularity in the potential. To see why this is reasonable, let us mention that, 
from the point of view of  Partial Differential Equations, there are two main
issues that distinguish $H_N$ and $H_{N-1}^{\eff}$ from the customary differential operators
studied in the introductory courses, namely: 
\begin{itemize}
  \item the behavior at infinity of the potential and 
  \item the singularities in the potential. 
\end{itemize}
Somewhat counterintuitively (from a pure mathematical point of
view) is that the singularities of the potential are less important
than the behavior at infinity, at least in what the spectral theory is concerned. 
In fact, it is known that many results concerning the essential spectrum of operators 
with potentials with Coulomb singularities can be obtained from the results for
the analogous operators, but with
smooth potentials (see, for example, \cite{GN} and the references
therein; incidentally, in addition to the Hardy inequality,
the argument there involves norm closures and
elementary $C^*$-algebra results). Thus, except in the very last subsection, 
in this paper, we will ``smooth out''
the singularities and hence look instead at a class of operators 
containing operators of the form
\begin{equation}\label{eq.BO}
   H_N^\prime  \ede  D + \sum_{1\leq j\leq N} v_j(x_j) +
   \sum_{1\leq j<k\leq N} v_{jk}(x_j - x_k) \,,
\end{equation}
where $D$ is a strongly elliptic differential
operator with constant coefficients, and $v_j$ and $v_{jk}$ are \emph{smooth functions} 
on $\RR^3$ with uniform radial limits at infinity.
(We suppressed from the notation the function~$u$ on which $H_N^\prime$ acts.
Also, our results remain valid for certain operators $D$ with
suitable, non-constant coefficients.)
In this paper, by a function with ``uniform radial limits at infinity'' we
mean a function that extends to a smooth function on the radial compactification of
that space, $\RR^3$ in this case. 
The choice of functions with uniform radial limits at infinity is 
motivated by the choice of ``more convenient coordinates'' in the case
of the hydrogen type atom and also because it leads to a less singular compactification
of $\RR^{3N}$. We note, however, that in the last subsection, the singularities
of the potential will play a central role in our regularity
estimates of Equation \eqref{eq.reg.est}.

Our second modification of the operators $H_N$ (already alluded to above)
will be to extend our setting from $\RR^{3N}$ to an arbitrary real, 
finite dimensional vector space $X$ and to allow for more general collision
planes. (More general than the collision planes $\{ x_j - x_k = 0 \}$ for $H_N$). 
More precisely, we will allow our collision planes to belong to a suitable finite 
set $\maF$ of linear subspaces of~$X$, as above and
as in \cite{GeorgescuBookNew, DerezinskiAnnals, 
DerGer2, GeIf06}, for instance. As in those works, which serve as a motivation for our 
approach, it will be convenient to assume that $\maF$ is stable under intersection. Recall 
that a family $\maS$ of subsets of~$M$ is a \emph{semilattice} (with respect to the inclusion) 
if, for all $P_1, P_2 \in \maS$, we have 
$P_1 \cap P_2 \in \maS$.  It will also be convenient -- and that will not 
decrease the generality -- to consider semilattices $\maF$ with $\{0\} \in \maF$.
Our operator $H_N^\prime$ will then be replaced 
with a more general operator of the form 
\begin{equation}\label{eq.def.HvY}
   H  \ede  D + \sum_{Y \in \maF} v_Y\,,
\end{equation}
where $v_Y$ is a smooth function on $X/Y$ with uniform radial limits at infinity
(more precisely $v_Y \in \maC(\oXY)$, where $X/Y$ is the quotient
vector space and $\oXY$ is its radial compactification). This completes our sequence of 
modifications of $H_N$ and provides us with the concepts needed to introduce our definition 
of compactification space $\XGV$ (the Georgescu-Vasy space) as follows.

Let $\maF$ be a finite semilattice of linear subspaces of~$X$ with  
$\{0\} \in \maF$, as above, and let
\begin{equation}\label{eq.def.delta}
    \delta_{\maF} : X \to \prod_{Y \in \maF} \overline{X/Y}
\end{equation}
be the diagonal map obtained from all the projections $X \to X/Y$. We define
the {\em Georgescu-Vasy space $\XGV$} as the closure 
\begin{equation}\label{eq.def.XGV}
    \XGV \ede \overline{\delta_{\maF}(X)}
\end{equation}
of $\delta_{\maF}(X)$ in $\prod_{Y \in \maF} \overline{X/Y}$. 
Since each $\overline{X/Y}$ is compact, $\XGV$ 
is also compact. Note that, our assumption that $\{0\} \in \maF$
implies, in particular, that the factor $\overline{X/Y}$ corresponding to
$Y = \{0\}$ is $\overline{X/\{0\}} = \overline{X}$, which contains~$X$ as
a dense, open subset, and hence the map~$\delta_{\maF}$ is injective on~$X$.
If $\maF$ is the semilattice corresponding to the $N$-body problem or
the effective $N$-body problem, then $X_{GV}$
will yield the desired space $M_k$ (for suitable $k$).

\subsection{Georgescu's and Vasy's compactifications}\label{subsec.vasy-blowup}
Let us recall first the definition of Geor\-ges\-cu's compactification 
\cite{GeorgescuBookNew, Georgescu2018, GeIf06} in the form
considered in \cite{GN}. To this end, let us consider the norm 
closed algebra ($C^*$-algebra)
\begin{equation}\label{eq.def.maES}
  \mEF X \ede \langle \maC(\overline{X/Y}) \rangle
\end{equation}
generated by all the spaces $\maC(\overline{X/Y})$ in $L^\infty(X)$,
with $Y \in \maF$. (Here $\maC(Z)$ denotes the space of continuous functions 
$Z\to \CC$, as usual.) The spectrum $\Spec(\mEF X)$ of this algebra (the set of its
characters) is a compact space containing naturally $X$. \emph{Georgescu's
compactification} is $\Spec(\mEF X)$. It was proved in \cite{MNP} using
elementary $C^*$-algebra arguments that $\Spec(\mEF X)$ is naturally 
homeomorphic to the spaces $\XGV := \overline{\delta_\maF(X))}$ considered 
above. A variant of the algebra $\mEF X$ defined above is obtained
by considering the one-point compactifications $(X/Y)^{+}$ \cite{GeorgescuBookNew, BMBGeorgescu, 
Georgescu2018}, see Remark \ref{rem.Georgescu}. 
Also, in \cite{MougelPrudhon}, the definition of
the algebra $\mEF X$
was generalized so that the only requirement on $\maF$ is $\{0\} \in \maF$. 
In particular, the family $\maF$ can be infinite and not necessarily stable by intersection
(but the sum defining the potential must be convergent).

Vasy's compactification is obtained using blow-ups
of manifolds with corners. Let $M$ be a manifold with corners (we recall the 
definition of manifolds with corners and related constructions in Section~\ref{sec.mfds.corners}). Recall 
that a {\em \psbmanifold} $P \subset M$ is
a submanifold of~$M$ that has a tubular neighborhood:
$P \subset U_P \subset M$ that is locally of a product form (see Definition~\ref{def.psubmanifold} for details). If $P$ is a closed \psbmanifold{} -- where 
``closed'' means closed as a subset -- then the \emph{blow-up $[M: P]$ of~$M$ with respect 
to $P$} is defined by replacing~$P$ with the set $\SS (N_+^MP)$ of \emph{interior} directions 
in the normal bundle $N^MP$ of~$P$ in~$M$ 
(see \cite{ACN, JoyceCorners, MelroseBook, 
nistorDesing}, or Definiton~\ref{def.blow-up}). As in those papers, for
a suitable $k$-tuple $(P_j)_{j=1}^k = (P_1, P_2, \ldots, P_k)$, we can define the iterated
blow-up $[M: (P_j)_{j=1}^k]$ by blowing up $M$ first with respect to $P_1$, then
with respect to (the lift of) $P_2$ to $[M: P_1]$, and then continuing in
this way (see Definition~\ref{def.iterated.bu}).
We apply these results to the study of the $N$-body problem in the
following way.
Let~$\oX$ denote be the spherical compactification of a
finite-dimensional vector space $X$, with boundary at infinity the sphere $\SS_X
:= \oX \smallsetminus X$ (see Subsection~\ref{ssec.sph.comp} for the detailed definitions). 
To our finite semilattice $\maF$ of linear subspaces of~$X$ containing the zero 
subspace we associate the semilattice 
\begin{equation}\label{eq.def.SSmaF}
    \SS_{\maF} \ede\{ \SS_Y = \SS_X \cap \overline{Y} \, \mid \ Y
\in \maF \}
\,.
\end{equation}
(Note that our assumption that $\{ 0 \} \in \maF$ implies that $\emptyset = 
\SS_{ \{0 \} } \in \SS_{\maF}$. Our approach, in fact, works also
for semilattices that do not contain the zero subspace, but, in any case,
we can reduce to this case by including the zero subspace in $\maF$.)
We arrange the elements of~$\SS_{\maF}$ according an ``admissible order'': 
$\SS_{\maF} = (P_0 = \emptyset, P_1, P_2, \ldots )$ (roughly, in the ascending inclusion
order, see Definition~\ref{def.admissible}). Then \emph{Vasy's compactification} is
obtained by iteratively blowing up $\oX$ with respect to $\SS_{\maF}$
as explained earlier to obtain  $[\oX: \SS_{\maF}]$ \cite{VasyReg, VasySurv}.
(This construction is discussed
in detail in the main body of the paper, Definition~\ref{def.iterated.bu}.)
See also Kottke's paper \cite{Kottke-Lin}, where this compactification was also 
recently studied.

Yet a fourth compactification of~$X$, still diffeomorphic to $\XGV$, is
the \emph{graph blow-up} $\bl{\oX: \SS_{\maF}}$. The \emph{graph blow-up} $\bl{M: \maP}$ of
$M$ with respect to the family $\maP$ is the closure
  $\bl{M: \maP} \ede 
  \overline{\delta(M \smallsetminus \medcup\limits_{P \in \maP} P)}  \ \subset \ 
  \prod_{P \in \maP} [M: P]$
where $\delta$ is the diagonal embedding (see Definition
\ref{def.unres.blowup}). One of the main results of this paper, Theorem \ref{thm.main1},
 establishes a diffeomorphism $[M: (P_j)_{j=1}^k] \simeq 
\bl{M: (P_j)_{j=1}^k}$ between the iterated blow-up and the graph blow-up
introduced in Definition \ref{def.unres.blowup}, provided that $(P_j)_{j=1}^k$
is an admissible ordered clean semilattice.
This compactification arises as an intermediate step in our 
proof of the equivalence (homeomorphism) of Georgescu's
construction (the space $\Spec(\mEF X)$) and Vasy's construction (the
space $[\oX : \SS_{\maF}]$) that relies, in order, on the following sequence of homeomorphisms
\begin{equation}\label{eq.plan.proof}
  [\oX : \SS_{\maF}] \, \simeq \, \bl{\oX: \SS_{\maF}}
   \, \simeq \, \XGV \ede \overline{\delta_\maF(X)}  
   \, \simeq \, \Spec(\mEF X) \,. 
\end{equation}
The smooth structures on the last three spaces
come from the first one. We stress that, while the last homeomorphism
(already proved in \cite{MNP}) has an easy proof, the other two, proved in
this paper, are rather difficult results.

\subsection{Contents of the paper and applications}
Let us now present the contents of the paper, including some applications
of our results (mainly of the homeomorphisms of Equation \eqref{eq.plan.proof}).
Section~\ref{sec.mfds.corners} contains
background material on manifolds with corners. In particular, we
devote quite a bit of effort to introduce and compare several classes of 
submanifolds of manifolds with corners. Section~\ref{sec.blowup.mfd.corners} 
recalls the definition of the blow-up of a manifold with corners with respect 
to a closed \psbmanifold{} and establishes a few properties of this
blow-up. In Section~\ref{sec.graph.blowup} we study the iterated and
the graph blow-ups. In particular, we prove that, for clean
semilattices of closed \psbmanifolds, the graph blow-up can also be
obtained as an iterated blow-up, which is one of the main technical
results of this paper. In Section \ref{sec6}, we use the identification
of the graph blow-up with the iterated blow-up to show that Georgescu's
compactification $\Spec(\mEF X)$ and Vasy's compactification $[\oX : \SS_{\maF}]$, 
are homeomorphic to the space $\XGV$ introduced above (see Equation
\eqref{eq.plan.proof}).

An important application of our results is the existence of various smooth group 
actions on blow-ups, in general, and on $\XGV$, in particular. For instance, 
we obtain an action of~$X$ on $\XGV$ by translations. Moreover, when $\maF$ corresponds 
to the $N$-body problem, we also obtain and action of the symmetric group $S_N$
and of the orthogonal group $\GL(3,\RR)$ on $M_N = \XGV$, see  Remark \ref{rem.Nbody}. 
In addition to group-action applications, in the 
last section, we outline \emph{four other applications} 
of our results, which we resume next.
\begin{enumerate}[\kern6mm (A)]
\item The first application is on the relation between the action of the 
symmetric group $S_N$ on $M_N$ and Pauli exclusion principle. Finding a good 
compactification of $\RR^{3N}$ that behaves well with respect to the action of
$S_N$ is a problem posed by Melrose and Singer from \cite{MelroseSinger}, which 
was solved in \cite{Kottke-Lin} using differential geometry and in 
\cite{MougelPrudhon} using $C^*$-algebras, see Subsection~\ref{subsec.N-body}.

\item The second application is to investigate
the relation between Vasy's pseudodifferential calculus and Georgescu’s
algebras. It combines the results of \cite{aln2} with the existence of
the smooth structure on $\XGV$ provided by our results
and with the smooth action of~$X$ on $\XGV$ to define
a pseudo-differential calculus $\Psi_{c}(\XGV)$ consisting of properly 
supported operators. This calculus is clearly smaller than the Vasy calculus, 
and a preliminary discussion of the relation of the two calculi is contained
in Subsection \ref{ssec.psdo}. We thus also define a completion $\Psi_{NB}(\XGV)$ 
of $\Psi_{c}(\XGV)$ that can be proved to be spectrally invariant (\ie it is 
stable for holomorphic functional calculus) and thus leads to a description of 
the distribution kernels of the resolvents of $H_N$, Proposition \ref{prop.sp.inv}. 
The more precise relation between Vasy's calculus and ours is, certainly,
worth further exploring, see Subsection~\ref{ssec.psdo}.

\item Our third application is to establish some connections
between our results and the HVZ theorem, see Subsection~\ref{ssec.HVZ}

\item Finally, the last application is a regularity result for bound states
for Hamiltonians with {\em inverse square singularities}. For instance, let 
$u \in L^2(X \smallsetminus \cup \maF)$ be  an eigenfunction of $H_N$, namely 
$H_N u = \lambda u$ on $X \smallsetminus \cup \maF$, $\lambda \in \RR$, and let 
$\rho(x) \ede \min\big \{ \dist(x, \unionF), 1 \big \}$, where $\dist(x, \unionF)$ is
the distance in the usual euclidean metric from $x$ to $\unionF$. Then, 
for any multi-indices $\alpha$, we have
\begin{equation}\label{eq.reg.est}
     \rho^{|\alpha|} \pa^\alpha u \in L^2(\RR^{3N})\,.
\end{equation}
The same result holds for $H_N^{\eff}$ in place of $H_N$ and for
many other operators. See Theorem~\ref{theorem.main.reg} in Subsection~\ref{subsec.regres} for a more
general statement. Full details can be found in \cite{AMN2}. 
\end{enumerate}

The study of inverse-square potentials is relevant since they 
appear in relativistic physics. See Subsection \ref{ssec.reg} for 
references and more on the physical motivation for inverse square potentials.

Two appendices include some related topological results on proper maps 
and on submanifolds of manifolds with corners. The reader can thus see that this 
paper relies essentially on geometry, necessarily so since Vasy's construction is
geometric.

\subsection*{Acknowledgements} 
We thank Vladimir Georgescu for useful discussions. We also thank anonymous
referees and the handling editor for carefully reading our paper and for
useful suggestions.

\section{Manifolds with corners and their submanifolds}\label{sec.mfds.corners}

We begin with some background material, mostly about manifolds with
corners. This section contains few new results, but the presentation
is original.

\subsection{Manifolds with corners}
We now introduce manifolds with corners and their smooth structure. We
also set up some important notation to be used throughout the paper.
The terminology used for manifold with corners is not uniform.
Nevertheless, good overviews of the concept of a manifold with
corners can be found in \cite{JoyceCorners, kottke-gblow,
  KottkeMelrose, Melrose92, nistorDesing}, to which we refer for the
concepts not defined here and for further references. In this paper,
we will mostly use the terminology introduced by Melrose and his
coauthors, which predates most of the other ones.

\subsubsection{Notation and conventions}\label{ssec.notation}
For any finite-dimensional real vector space $Z$, let $\SS_Z$ denote
the \emph{set of vector directions} in $Z$, that is, the set of
(non-constant) open half-lines $\RR_{+} v$, with $0 \neq v \in Z$ and
$\RR_{+}:=(0,\infty)$. We will also use the standard notation
$\SS^{n-1} := \SS_{\RR^n}$, for simplicity. In particular, if $Z$ is a
euclidean (real) vector space, we identify $\SS_{Z}$ with the unit
sphere in $Z$. Informally, a manifold with corners is a topological
space locally modeled on the spaces
\begin{equation}\label{eq.def.Rnk}
  \RR_k^n \ede [0,\infty)^k \times \RR^{n-k}.
\end{equation}
For $k,n \in \mathbb{N}=\{0,1\ldots \}$ with $k \leq n$, we let
$\SS_k^{n-1} \subset \RR^n$ be
\begin{equation}\label{eq.not.sphere}
   \SS_k^{n-1} \ede \SS^{n-1} \cap \RR_k^{n} \seq \{ \phi= (\phi_1,
   \ldots \phi_n) \,\vert\ \| \phi \| = 1 \mbox{ and } \phi_i \geq 0
   \mbox{ for } 1 \le i \le k \},
\end{equation}
where $\|.\|$ is the euclidean norm on $\RR^n$.

\begin{remark}\label{rem.first.components}
Let us write $0_V$ for the neutral element of a vector space $V$, when
we want to stress the space to which it belongs. We will often use
maps between subsets of euclidean spaces. As a rule, we will try
not to permute the coordinates and, moreover, our embeddings will be
``first components'' embedings. More precisely, let $k' \le k$ and
$n'-k' \le n-k$, we shall then use with priority the canonical {\em first
components embedding,} namely given by:
\begin{equation}\label{eq.def.can0}
  \begin{gathered}
    \RR^{n'}_{k'} \, \simeq \, [0,\infty)^{k'} \times
      \{0_{\RR^{k-k'}}\} \times \RR^{n'-k'} \times \{0_{\RR^{n-n'}}\}
      \, \subseteq \, [0,\infty)^{k} \times \RR^{n-k} \seq
        \RR^{n}_{k}\\
    (x', x'') \ \mapsto \ (x', 0_{\RR^{k-k'}}, x'', 0_{\RR^{n-n'}}) \,.
  \end{gathered}
\end{equation}
Occasionally, other embeddings (involving permutations of the coordinates) between
these sets will also be considered, in which case they will 
be explained separately. For instance, we shall sometimes find it notationally
convenient to use the {\em canonical permutation of coordinates}
diffeomorphism
\begin{equation}\label{eq.def.can}
  \begin{gathered}
    \can : \RR^n_k \times \RR^{n'}_{k'} \ \simeq \ \RR^{n+n'}_{k+k'}\\
    (x', x'', y', y'') \ \mapsto \ (x', y', x'', y'') \in
         [0,\infty)^{k+k'} \times \RR^{n+n'-k-k'}\,,
  \end{gathered}
\end{equation}
where $x' \in [0,\infty)^{k}$ and $y' \in [0,\infty)^{k'}$. (Compare
    with Equation \eqref{eq.def.Rnk}.) 
\end{remark}

\subsubsection{Charts and atlases}
We shall use suitable charts to define the smooth structure on
manifolds with corners. We proceed as in
the case of smooth manifolds (without corners).
We begin with the following standard definition.
\begin{definition}\label{def.smooth}
Let $U \subset \RR^n_k$ and $V \subset \RR^m_l$ be two open subsets
and $f=(f_1,\ldots,f_m) : U \to V$.  We shall say that:
\begin{itemize}
\item[(a)] $f$ is \emph{smooth} on $U$ if there exists an open
  neighborhood $W$ of $U$ in $\RR^n$ such that $f$ extends to a smooth
  function $\tilde f: W \to \RR^m$.
\item[(b)] $f$ is a \emph{diffeomorphism} between $U$ and $V$ if $f$
  is a bijection and both~$f$ and~$f^{-1}$ are smooth.
\end{itemize}
\end{definition}

\begin{definition}\label{def.chart}
A \emph{corner chart on~$M$} (or simply, {\em ``chart''} in what
follows) is a couple $(U,\phi)$ with $U$ an open subset of~$M$ and
$\phi: U \to \Omega$ a homeomorphism onto an open subset $\Omega$ of
$\RR_k^n$. Let $(U,\phi)$ and $(U',\phi')$ be two corner charts with
values in $\RR_k^n$ and in $\RR^{n'}_{k'}$, respectively. Let $V:= U
\cap U'$. We shall say that the corner charts $(U,\phi)$ and
$(U',\phi')$ are {\em compatible} if
\begin{equation*}
  \phi'\circ \phi^{-1} : \phi(V) \to \phi'(V)
\end{equation*}
is a diffeomorphism (see Definition \ref{def.smooth}) between the open
subsets $\phi(V)\subset \RR^n_k$ and $\phi'(V)\subset\RR^{n'}_{k'}$.
(So $n = n'$ if $V \neq \emptyset$.)
\end{definition}

Given a point $m\in M$ and a corner chart $(U,\phi)$ with $m \in U$,
we can always find a corner chart $(U',\phi')$, $\phi': U'\to
\RR^{n'}_{k'}$, compatible with $(U,\phi)$ such that $\phi'(m)=0$ and
$k'$ is minimal. The least such $k$ is the {\em boundary depth}
of $m$ in $M$. See Subsection \ref{ssec.b.faces} below.

\begin{definition}\label{def.atlas}
A \emph{corner atlas} $\maA=\{(U_a,\phi_a), a \in A\}$ on~$M$ is a
family of compatible corner charts such that $M =\bigcup \limits_{a
  \in A} U_a$. Two corner atlases are called {\em equivalent} if their
union is again a corner atlas. A \emph{manifold with corners} is
defined to be a paracompact Hausdorff space $M$ with an equivalence
class of corner atlases (on~$M$). In the following we will drop the
word ``corner'' before the words ``chart'' and ``atlas.'' In the
context of a manifold with corners, the terms ``atlas'' and ``chart''
will always mean ``corner atlas'' and, respectively, ``corner chart.''
\end{definition}

We stress that we do not require the connected components of
a manifold with corners to have the same dimension.
Let us also remark that a manifold with corners in the above sense is called a
``$t$-manifold'' in \cite[Section~1.6]{MelroseBook},
where the ``$t$'' stands for ``tied''. If $M$ is
manifold with corners, then the union of all its atlases is again an
atlas, the {\em maximal atlas} of~$M$. An open subset of a manifold
with corners is, again, in an obvious way, a manifold with corners.
Many concepts extend from the case of manifolds without corners to
that of manifolds with corners.

\begin{definition}\label{def.new}
Let $f : M \to M'$ be a map between two manifolds with corners. We
will say that $f$ is {\em smooth} if, for any two charts $(U,\phi)$ of
$M$ and $(U',\phi')$ of $M'$, the map $\phi' \circ f \circ \phi^{-1}$
is smooth on its domain of definition $\phi(f^{-1}(U'))$. If~$f$ is a
bijection and both~$f$ and~$f^{-1}$ are smooth, we will say that~$f$
is a {\em diffeomorphism.}
\end{definition}

Here are some examples of manifolds with corners that will be
used in this paper.

\begin{example}\label{ex.corners}
Using the notation from Subsection~\ref{ssec.notation}, we have the
following:
\begin{enumerate}[(i)]
\item Any open subset of $\RR^n_k := [0, \infty)^k \times \RR^{n-k}$
  is a manifold with corners.
\item The sphere orthant $\SS^{n-1}_k := \SS^{n-1} \cap \RR^n_k$ of
  Equation \eqref{eq.not.sphere} is a manifold with corners.
\item Any smooth manifold is a manifold with corners (even if it
  does not have a boundary or any true corners).
\end{enumerate}
\end{example}

The following will be used to introduce \psbmanifolds.

\begin{definition}\label{def.boundary.depth}
Let $I$ be a subset of $\{1, \ldots, n\}$ and $L_I$ be the subset of
$\RR^n_k$ defined by
\begin{equation}\label{eq.p-manifold}
  L_I \ede \{\, x=(x_1,\ldots , x_n) \in \RR^n_k \, \vert \ x_i=0
  \mbox{ if } i \in I\, \} \, .
\end{equation}
The number $b := |I\cap \{1,\ldots,k\}|$ of elements of $I\cap \{1,\ldots,k\}$
will be called the \emph{boundary depth} of $L_I$ in $\RR^n_k$. Similarly, $c:= |I|$ is the 
codimension of $L_I$ in $\RR^n_k$ and $d:=n-c$ is its dimension. 
\end{definition}

\subsection{The boundary and boundary faces of a manifold with corners}
\label{ssec.b.faces}
We now fix some standard terminology to be used in what follows,
extending the local definitions of Definition~\ref{def.boundary.depth}.
In particular, we need the intrinsic definition of the boundary of a
manifold with corners. Recall that above
we defined the \emph{boundary depth} (in $M$) $\depth_M(p)$ of a
point $p\in M$ as the number of non-negative coordinate functions
vanishing at $p$ in any local coordinate chart at $p$. It is the least
$k$ such that there exists a chart near $U$ with values in~$\RR_k^n$.
Let $(M)_k$ be the set of points of~$M$ of boundary depth $k$. 
It is a smooth manifold (no corners). Its connected components are called
the \emph{open} boundary faces (or just the \emph{open} faces) of
codimension (or boundary depth) $k$ of~$M$. The set $(M)_0$ is the 
\emph{interior} of~$M$.
A \emph{boundary face} of  boundary depth $k$
is the closure of an open boundary face of  boundary  depth $k$.
For every boundary face $F$ of codimension~$k$ there is a manifold with corners 
$\FF$ of dimension $n-k$ and an immersion $\iota\colon\FF\to M$ that is 
injective on $\interior{\FF}$ and $\iota(\FF)=F$. However, in general
  $\iota$ is not injective on the boundary faces of $\FF$; an example
  is the teardrop domain of Figure~\ref{fig.teardrop}. For this domain, $M$ does 
  not induce a structure of a manifold with corners on the boundary face $F$. 
  Nevertheless, we will not exclude from
  our consideration boundary faces that are not manifolds with corners; in
  particular, the boundary of the teardrop domain is a boundary face of that domain.

\begin{figure}\label{fig.teardrop}
\begin{center}
  \begin{tikzpicture}[scale=.3]
    \draw[fill=gray!50] (1.5,3) -- (0,0) .. controls (-.61,-1)  and (.5,-2) .. (1.5,-2) --  (1.5,-2)  .. controls (2.5,-2) and (3.61,-1) .. (3,0) -- cycle;
    \node[above] (ABC) at (1.5,3.2) {$\scriptstyle p$}; 
    \draw[fill=black] (1.5,3) circle (.08);
  \end{tikzpicture}
\end{center}
\caption{A teardrop domain: a compact simply-connected domain in~$\RR^2$ with exactly one 
boundary face $F_1$ of  boundary  depth $1$ and exactly one corner $p$ of boundary  depth $2$. The face 
$F_1$ is the image of an immerision of $[0,1]$ that is injective on the open interval 
$(0,1)$, but which maps both ends to the corner~$p$.}
\end{figure}
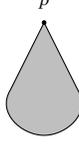

We will denote by $\mathcal{M}_k(M)$
the set of all {\em closed} boundary faces of codimension~$k$. In
particular, the {\em boundary} $\pa M$ of~$M$, defined as the set of
all points of  boundary depth $>0$, is given by
\begin{equation}\label{eq.def.bm}
  \pa M \ede \bigcup_{H \in \mathcal{M}_1(M)} H.
\end{equation}
A boundary face of~$M$ of codimension one will be called a {\em
  hyperface}, as usual. Thus $\pa M$ is the union of the
hyperfaces of~$M$. If $H$ is a hyperface of~$M$ and $0 \le x \in
\CI(M)$ is a function such that $H = x^{-1}(0)$ and $dx \neq 0$ on
$H$, then $x$ is called a {\em boundary defining function} of $H$. As
the example of the teardrop domain shows, not all hypersurfaces have a
boundary defining function. However, each boundary face $F$ of
codimension~$k$ can {\em locally} be represented as 
$$F \cap U \seq \{x_1 = x_2 =
\ldots = x_k = 0\}\,,$$ where $x_j$ are boundary defining functions of
the hyperfaces of $F \cap U$ in $U$.

It is convenient to consider embeddings.

\begin{remark}\label{rem.tildeM}
Every manifold with corners $M$ is contained in a smooth manifold
$\widetilde M$ such that each component of~$M$ is contained in a 
component of $\widetilde M$ of the same dimension 
\cite{ammann.ionescu.nistor:06, JoyceCorners,
  KottkeMelrose, MelroseBook, Melrose92, nistorDesing}. 
We note, however, that $\pa M$ is intrinsically defined, 
Equation \eqref{eq.def.bm}, and,
sometimes, it is {\em not} the {\em topological} boundary $\overline{M}
\smallsetminus \overset{\circ}{M}$ of~$M$, where the closure
$\overline{M}$ and the interior $\overset{\circ}{M}$ are defined as subsets of
$\widetilde M$. For instance, when $M := \{ x \in \RR^n \, \vert
  \ x_n \ge 0, \|x\|< 1\}$ and $\widetilde M = \RR^n$, then $\pa M =
  \{ x \in \RR^n \, \vert \ x_n = 0, \|x\|< 1\}$, whereas the
  topological boundary of~$M$ is $\pa M \cup \{ x \in \RR^n \, \vert
  \ x_n > 0, \|x\| = 1\}$, a bigger set. In fact, we always have that
  $\pa M$ is contained in the topological boundary of~$M$ in
  $\widetilde M$. Unlike $\pa M$, the topological boundary of~$M$ in
  $\widetilde M$ depends on $\widetilde M$. Embeddings are convenient also 
  because one can define
\begin{equation*}
  TM \ede T \widetilde M\vert_{M}\,.
\end{equation*}
Up to a diffeomorphism, $TM$ can be obtained by gluing 
open subsets of the tangent spaces $T(\RR_k^n) := \RR_k^n \times \RR^n$ using an atlas of
$M$. We also let $T_x^+M$ be the set of {\em inward-pointing tangent vectors} of $T_xM$ 
(this includes the vectors tangent to the boundary). It can be defined as
the set of equivalence classes of curves starting at $x$ and
completely contained in $M$. We finally let $T^+M := \bigcup_{x \in M}
T^+_x M$ with its projection map to $M$. Note that $T^+M$ is {\em not} a
fiber bundle, but a fiberwise conical closed subset of the tangent
space.
\end{remark}

\subsection{Submanifolds of manifolds with corners}
We now discuss the two notions of submanifolds of a manifold
with corners -- ``\wsbmanifolds'' and ``\psbmanifolds'' -- used in
this paper. The class of \psbmanifolds{} is the class of
submanifolds with suitable tubular neighborhoods, whereas the class
of \wsbmanifolds{} is the largest class of subsets that have 
some claim to being a ``submanifold'' of a certain kind. 
The concept of a \psbmanifold{} is due to Melrose \cite{MelroseBook}
and is central in what follows, since the class of closed
\psbmanifolds{} is the class of submanifolds
with respect to which we can perform blow-ups.
On the other hand, the concept of a \wsbmanifold, with which we begin, is
weaker than the one considered in \cite{MelroseBook}, 
but also arises naturally in our study. Other concepts
of submanifolds are discussed in Appendix \ref{sec.appendix.submanifold}.

\subsubsection{Submanifolds in the weakest sense: \wsbmanifolds}
We start the discussion with a needed notion of submanifolds, 
called ``\wsbmanifolds'' in order to distinguish them from other
classes of submanifolds of manifolds with corners, see Appendix
\ref{sec.appendix.submanifold}.
The adjective ``weak'' in
``weak submanifolds'' indicates that the class of \wsbmanifolds{} is
the largest class of ``submanifolds'' that we consider, with the
minimal requirement that a submanifold (of a manifold with corners)
be also a manifold with corners in its own way for the induced
structure from the ambient manifold with corners.
As a first step (Definition~\ref{def.weak-s.0}), we introduce the
class of weak submanifolds of $\RR^n$ (or more generally of manifolds without corners 
or boundary). Then as a second step, see Definition~\ref{def.weak-submanifold}, 
we generalize to weak submanifolds in manifolds with corners. Thus in the 
following definition we mainly need the case $M_1=\RR^n$.

\begin{definition}\label{def.weak-s.0}
Let $M_1$ be a {\em smooth} manifold (\ie without boundary or corners) of dimension~$n$. 
A subset $S \subset M_1$ is a \emph{\wsbmanifold} of $M_1$ if, and only if, 
any $x\in S$ is contained in the domain of a chart $\psi:U\to \Omega \subset \RR^n$ 
of $M_1$ with $\psi(S\cap U) = (\RR^m_\ell \times \{0\}) \cap \Omega$, for some 
$0 \leq \ell \leq  m \leq n$. (Hence $m$ is the \emph{dimension} of $S$ at $p$ and 
$n-m$ is the \emph{codimension} of $S$ in $M_1$ at $p$.)
\end{definition}

A submanifold $S$ with corners of some smooth manifold $M_1$ inherits from $M_1$
the structure of a manifold with corners. In manifolds without corners the definition 
of a ``weak submanifolds'' coincides with what often is considered as a submanifold 
with (boundary or) corners.
The definition gets a bit more involved when we  generalize to the 
case when $M_1$ is replaced with a manifold $M$ with corners, thus obtaining a generalization 
of the submanifold property to the category of manifolds with corners.

\begin{definition}
\label{def.weak-submanifold}
A subset $S$ of a manifold with corners $M$ is a \emph{\wsbmanifold{} of~$M$} if, for 
every $p \in S$, there is a chart $\phi : U \to \Omega \subset \RR^n_k$ on~$M$, 
$0 \le k \le n$, such that
\begin{enumerate}
\item\label{def.weak-submanifold.i} $p \in U$ and
\item\label{def.weak-submanifold.ii} $\phi(S \cap U)$ is 
a weak submanifold of $\RR^n$, see Definition \ref{def.weak-s.0}.
\end{enumerate}
\end{definition}

Any weak submanifold of a manifold with corners
carries an induced structure of a manifold with corners, 
see Remark~\ref{rem.sbmdf-corner.mfd-corners} \eqref{rem.sbmdf-corner.mfd-corners.i} 
below for details. In particular, the \emph{dimension $\dim_p(S)$ of $S$ at $p$} is, from
the definition,  
the dimension of $\phi(S \cap U)$ at $\phi(p)$. (Recall that the function $\dim_p(S)$, while
locally constant, it is not required to be constant on a manifold
with corners $S$.)
The concept of a \wsbmanifold{} is more general than the concept of a ``submanifold''
considered in \cite{MelroseBook} and recalled in Definition~\ref{def.submanifold-gen} in the
Appendix. Roughly speaking, compared to many other definitions of submanifolds in manifolds 
with corners in the literature,  the corner structure of \wsbmanifolds{} has weaker 
compatibility properties with the corner structure of the ambient 
manifold. See Examples~\ref{ex.weak.submanifold} for examples of a \wsbmanifolds{}
that are not submanifolds in the sense of Definition~\ref{def.submanifold-gen}.
One can reformulate Definition~\ref{def.weak-submanifold} as follows.

\begin{lemma}\label{lemma.krit.weak.submanifold}
The subset $S \subset M$ is a \emph{\wsbmanifold} of~$M$ if, and only if, $M$ can be 
extended to a smooth manifold $\widetilde M$ (without boundary or corners), such that $S$ 
is a \wsbmanifold{} of $\widetilde M$. That is, for any $p\in S$ there are numbers 
$0\leq\ell\leq m\leq n=\dim M=\dim \widetilde M$ and a chart 
$\tilde\phi:\widetilde U\to\widetilde \Omega \subset \RR^n$ of $\widetilde M$ with $p\in \widetilde U$ and
\begin{equation}\label{eq.def.wsubm.alternative}
  \tilde\phi(S \cap \widetilde U)\seq \left(\RR^{m}_{\ell}\times \{0\}\right) \cap
  \widetilde \Omega \,.
\end{equation}
\end{lemma}

In order to understand the statement of this lemma, note that $\phi$ in 
Definition~\ref{def.weak-submanifold} and~$\tilde\phi$ in Lemma~\ref{lemma.krit.weak.submanifold} 
have properties that are, in general, incompatible. The map $\phi$ turns the boundary components 
of~$M$ into (open subsets of) coordinate hyperfaces of $\RR^n_k$, while the map 
$\tilde\phi$ turns the submanifold into (an open  subset of) 
$\RR^{m}_{\ell}\times \{0\}$. For general weak submanifolds, both properties cannot 
be achieved simultaneously.

\begin{proof}
On the one hand, $\bigl(\RR^{m}_{\ell}\times \{0\}\bigr)\cap \widetilde \Omega$ is a 
\wsbmanifold{} of $\widetilde\Omega$. The ``if'' part thus follows from the fact that 
the property of being a \wsbmanifold{} in the sense of Definition~\ref{def.weak-s.0} is 
invariant under restrictions to open subsets and under diffeomorphisms, which we apply 
to the diffeomorphism~$\tilde\phi$. On 
the other hand, if we are in the situation of Definition~\ref{def.weak-submanifold}, then  
$\phi(S \cap U)$ is a \wsbmanifold{} of $\RR^n$. Let $\psi$ be a chart for this 
\wsbmanifold{} in the sense of Definition~\ref{def.weak-s.0}. We  may assume, without loss 
of generality, that $\psi$ is defined on $\phi(S \cap U)$. Let $\phi_1:U_1\to \Omega_1$ be 
a smooth extension of $\phi:U\to \Omega$ to a chart of some extension $\widetilde M$ 
of $M$ with properties as above with $U= U_1\cap M$. 
Then $\tilde\phi:=\psi\circ \phi_1$ yields a map with the properties of Equation~\eqref{eq.def.wsubm.alternative} of our lemma. 
This yields the ``only if'' part and thus completes the proof.
\end{proof}

It is clear from the construction above, that in the lemma we can choose $\widetilde M$ 
such that~$M$ is a closed subset of~$\widetilde M$.

\pagebreak[3]
\begin{remarks}\label{rem.sbmdf-corner.mfd-corners}\ \nopagebreak
\begin{enumerate}[\kern6mm (a)]
\item\label{rem.sbmdf-corner.mfd-corners.i} 
Let $S$ be a weak submanifold of~$M$, then $S$ is covered by charts
  $\widetilde\phi\colon U\to \Omega$ of an extension $\widetilde M$ of~$M$ as in 
  \eqref{eq.def.wsubm.alternative}. Any $\tilde\phi$ yields
  a chart of $S$:
  $$\psi:= \tilde\phi|_{S\cap \widetilde U}:{S\cap \widetilde U}\to \left\{x\in \RR^{n'}_{k'}
  \,\big|\, (x,0)\in \widetilde\Omega\right\}.$$
The set of a such charts $\psi$ is an atlas of
$S$, the \emph{induced atlas on $S$}. With this atlas, the set $S$ is
a manifold with corners.
Thus any \wsbmanifold{} is a manifold with corners in its own.

\item\label{rem.sbmdf-corner.mfd-corners.ii} If $M$ is manifold with corners and $S$ a 
subset of the interior $M_0$ of $M$. Then $S$ is a weak submanifold of $M$ in the sense 
of Definition~\ref{def.weak-submanifold}, if and only if it is a weak submanifold in the 
sense of Definition~\ref{def.weak-s.0}.

\item\label{rem.sbmdf-corner.mfd-corners.iii} Definition~\ref{def.weak-s.0} -- in which we assumed 
that $M_1$ has no corners and no boundary -- is equivalent to several classical definitions in the 
literature. For example, a subset $S \subset M_1$ is a weak submanifold, if and only if it is a 
submanifold in the sense of Melrose, see 
Definition~\ref{def.submanifold-gen}. As discussed above and in Examples~\ref{ex.weak.submanifold}, 
this no longer holds for submanifolds of manifolds with corners.

\item\label{rem.sbmdf-corner.mfd-corners.iv} 
In view of our experience with the preprint version of this article, we want to comment 
on a possible missunderstanding of Definition~\ref{def.weak-submanifold}. It is a priori obvious that $\phi(S\cap U)$ is a subset of $\RR^n_k$, so one might think, that one could or even should replace 
$\RR^n$ in Definition~\ref{def.weak-submanifold} \eqref{def.weak-submanifold.ii} by $\RR^n_k$. However, 
this version would run into the obvious trouble that $\RR^n_k$ is manifold with corners and weak submanifold  
of manifolds with corners had not yet been defined before.
Thus Definition~\ref{def.weak-submanifold} would be circular. The possible missunderstanding is related, 
but more subtle. In view of item \eqref{rem.sbmdf-corner.mfd-corners.iii} one might assume that one may 
replace Definition~\ref{def.weak-submanifold} \eqref{def.weak-submanifold.ii} by\smallskip\\
  \kern3mm (2$'$)\quad $\phi(S \cap U)$ is 
  a submanifold of $\RR^n_k$ in the sense of Definition~\ref{def.submanifold-gen},\smallskip\\
 as this property has already been defined in the literature.
 However, this replacement would yield a definition equivalent to  Definition~\ref{def.submanifold-gen} 
 and thus a more restrictive notion of submanifold than ours.

\item\label{rem.sbmdf-corner.mfd-corners.v} 
Melrose's more restictive definition of a submanifold given in 
Definition~\ref{def.submanifold-gen} is not sufficiently general for our purposes. 
Indeed, the set $S$ in Example~\ref{ex.weak.submanifold} \eqref{ex.weak.submanifold.i} 
is a \wsbmanifold{} of $\RR^2_1$, but not a \Msbmanifold. Similarly, the graph blow-up 
is a \wsbmanifold{} of the product considered in Definition~\ref{def.unres.blowup}, 
but in general it is not a \Msbmanifold, see \cite{koenig.master} for details.

\item\label{rem.sbmdf-corner.mfd-corners.vi}
Any weak submanifold $S$ of a manifold
with corners $M$ is \emph{locally closed}, \ie  it is the
intersection of a closed subset with an open subset. In order to prove this we embed~$M$ 
into a manifold without corners and without boundary $\widetilde M$ with $M$ closed in $\widetilde M$,
and we consider an atlas $\maA$ of charts $\tilde\phi:\widetilde U\to \widetilde\Omega$ as in Lemma~\ref{lemma.krit.weak.submanifold}.
Using that $\left(\RR^{m}_{\ell}\times \{0\}\right) \cap  \widetilde \Omega$ is closed in $\widetilde \Omega$,
we obtain that $S\cap \widetilde U$ is closed in $U\seq \widetilde U\cap M$. 
Thus $S$ is covered by charts of $M$ in which $S$ is closed, thus it is locally closed in $M$.
\end{enumerate}
\end{remarks}

We need \wsbmanifolds{} in view of the following proposition.

\begin{proposition}\label{prop.smfd.crit}
Let $N$ and $M$ be manifolds with corners and let $f:N\to M$ be a smooth
map that is an immersion and a homeomorphism onto its image. Then
$f(N)$ is a \wsbmanifold{} of~$M$ (Definition~\ref{def.weak-submanifold})
and $f : N \to f(N)$ is a diffeomorphism.
\end{proposition}

\begin{proof}
By extending the charts of~$M$, we can find an extension $\widetilde M$ as 
in Lemma~\ref{lemma.krit.weak.submanifold}, \iet $M$ is then a submanifold
with corners of codimension~$0$ in a manifold $\widetilde M$ without corners 
or boundary. Similar we can extend $N$ to a manifold~$\widetilde N$ without 
corners or boundary of the same dimension, and if $\widetilde N$ is sufficiently 
small, then we can find an immersion $\tilde f:\widetilde N\to \widetilde M$ 
extending $f$. The injectivity of $f$ and its homeomorphism property imply
(after passing to a possibly smaller $\widetilde N\supset N$) that the map 
$\tilde f$ is injective and satisfies the homeomorphism property.
  
Note that the homeomorphism property implies that we have a local
statement, that is, we can restrict ourselves to small neighborhoods in $N$,  
$\widetilde N$, $M$ and, respectively, $\widetilde M$.
For $p\in N$, let $\rho:V_0\to W_0\subset \RR^m_\ell$ be a chart of $N$ containing $p$, 
$m=\dim N$, with $W_0$ open in $\RR^m_\ell$,
that extends to a chart $\tilde \rho: \widetilde V_0\to   \widetilde W_0\subset \RR^m$ 
of $\widetilde N$. We may extend the immersion 
$\tilde f \circ \tilde\rho^{-1}:\widetilde W_0 \to  \widetilde M$ to a diffeomorphism
$\Phi:\widetilde W_0 \times B_\epsilon(0,\RR^{n-m}) \to U_0$ where $U_0$ is open in 
$\widetilde M$. Then $\tilde \phi:=\Phi^{-1}$ is a chart of $\widetilde M$ and, if 
$\widetilde V_0$ and $\epsilon$ were chosen sufficiently small, then we may verify
$$\tilde\phi\bigl(f(N)\cap U_0\bigr) = (\RR^m_\ell\times\{0\}\bigr)\cap 
\bigl(\widetilde W_0 \times B_\epsilon(0,\RR^{n-m}) \bigr).$$
The statements of the proposition then follow from Lemma~\ref{lemma.krit.weak.submanifold}.
\end{proof}

We obtain the following easy consequence.

\begin{corollary}\label{cor.smfd.crit}
Let $N$ and $M$ be manifolds with corners and $f:N\to M$ be a smooth
map. If there is a smooth map $F:M\to N$ with $F\circ f=\id_N$, then
$f(N)$ is a \wsbmanifold{} of~$M$, \ie a submanifold in the sense of
Definition~\ref{def.weak-submanifold}.
\end{corollary}

\begin{proof}
The relation $\id_{T_xN}=d_{f(x)}F \circ d_xf$ implies that
$d_xf:T_xN\to T_{f(x)}M$ is injective. As $F|_{f(N)}$ is continuous,
$f$ is a homeomorphism onto its image.
\end{proof}

Proposition~\ref{prop.smfd.crit} can also be used to prove the
following property: If $P$ and $Q$ are \wsbmanifolds{} of~$M$ and
$P\subset Q$, then $P$ is a \wsbmanifold{} of~$Q$.
Thus \wsbmanifolds{} have nicer categorical properties than
submanifolds in the sense of Definition~\ref{def.submanifold-gen}.
In the categorical language, the above property is
expressed as follows: if we consider the category whose objects are
manifolds with corners and whose morphisms are inclusions as a weak
submanifold, then this is a \emph{full} subcategory of the category of
sets with the inclusions as morphisms. On the other hand,
this property is not valid for submanifolds in the sense of 
Definition~\ref{def.submanifold-gen} (see 
Examples~\ref{ex.weak.submanifold}).

\subsubsection{Submanifolds with tubular neighborhoods:
  \psbmanifolds{}} We now recall Melrose's definition of a \psbmanifold{}
of a manifold with corners $M$ \cite{Kottke-Lin, MelroseBook,
  VasyReg}. In our paper, \psbmanifolds{} are of central importance,
as we blow-up manifolds with corners along closed \psbmanifolds.

Recall the subsets $L_I \subset \RR_k^n$ of Definition \ref{def.boundary.depth}.
After reordering the coordinates, $L_I$ is the first factor of $\RR^n_k
\cong \RR^d_{k-b}\times \RR^c_{b}$, in the sense that $L_I$ is mapped
to $\RR^d_{k-b}\times \{0\}$. The sets $L_I$ are the local models for \psbmanifolds{}
\cite[Definition~1.7.4]{MelroseBook}. However, in the following definition, 
we do not reorder the coordinates.

\begin{definition}
\label{def.psubmanifold}
A subset $P$ of a manifold with corners $M$ is a
\emph{\psbmanifold{}} if, for every $x \in P$, there exists a chart
$(U, \phi)$ with $x \in U$ and $I \subset \{1, 2, \ldots, n\}$ such
that
\begin{equation*}
  \phi(P \cap U) \seq L_I \cap \phi(U),
\end{equation*}
with $L_I$ as defined in Equation \eqref{eq.p-manifold}. The number
$n-|I|$ (respectively, $|I|$, respectively, $\depth_{M}(p) :=  |I\cap
\{1,\ldots,l\}|$) will be called the \emph{dimension}
(respectively, the \emph{codimension} of $P$ in $M$ at $x$, respectively, the
\emph{boundary depth} of $P$ in $M$ at $x$). We allow \psbmanifolds{} $Y$ of
non-constant dimension. We define $\dim (Y)$ as the maximum of the
dimensions of the connected components of~$Y$ and $\dim \emptyset = - \infty$.
\end{definition}

Of course, this definition extends the previous definition of
the boundary depth of $L_I$ in $\RR_k^n$ as the boundary depth of any interior 
point of $L_I$ in $\RR_k^n$. Obviously all \psbmanifolds{} are \wsbmanifolds{} 
(and submanifolds in the sense of Definition~\ref{def.submanifold-gen}),
and the definition of the dimension of~$P$ in~$x$ coincides with the dimension 
already defined above.

\begin{remark}\label{rem.depths}
The codimension and the boundary depth of a \psbmanifold{} $P$
at $p \in P$
(introduced in Definition~\ref{def.psubmanifold}) are locally constant functions in $p$. 
(The dimension of $P$ at $p$ is also locally constant 
in $p \in P$, but this is true in general for \wsbmanifolds.)
For any interior point $x$ in $P$ and $\epsilon > 0$ small enough, these
numbers are the codimension (respectively, the the boundary depth,
respectively, the dimension) of the intersection
$B_\epsilon(x)\cap P$ in $M$. More generally: if $P$ is a
\psbmanifold{} of~$M$ with boundary depth $d$ on the component of $x
\in P$, and if $x$ is a (boundary) point of  boundary  depth $e$ in $P$, then $x$
has  boundary depth $d+e$ in $M$. In particular, for a \psbmanifold{} $P
  \subset M$, the difference of  boundary depths $\depth_M (x) - \depth_P(x)$ is
  constant on the connected components of~$P$.
\end{remark}

\begin{example}\label{ex.first.comp}
This definition of a \psbmanifold{} comes from \cite{MelroseBook}.
Note that ``$p$'' is used as an abbreviation for ``product'',
reflecting the fact that, locally in coordinate charts,
\psbmanifolds{} are a factor of the product $\RR^n_k \simeq
\RR^{n_1}_{k_1}\times \RR^{n_2}_{k_2}$. More precisely,
let $\ell \le k$ and $m - \ell\le n - k$, so that the canonical
first components inclusion $j_0 : \RR^{m}_{\ell} \to \RR^n_k$ 
of Equation~\eqref{eq.def.can0} is defined. Then $j_0(\RR^m_\ell)$
is a \psbmanifold{} of $\RR^n_k$ and $j_0(\SS^{m-1}_{\ell}) = \SS^{n-1} 
\cap j_0(\RR^m_\ell)$ is a \psbmanifold{} of $\SS^{n-1}_{k} = \SS^{n-1} 
\cap \RR^{n}_{k})$. In fact, it is the first factor with $n_1=m$ and $k_1=\ell$.
A more general concept, that of an ``interior binomial subvariety,'' was 
introduced and studied in~\cite{KottkeMelrose}.
\end{example}

\begin{remark}
Let $P \subset M$ be a \psbmanifold{}. Then it is possible that
$P \subset F$, for $F$ a non-trivial face of~$M$. If $P$ is
connected, then the boundary depth of $P$ is the boundary depth of
the smallest closed face $F$ of~$M$ containing $P$. The
earlier notion of a ``submanifold of a manifolds 
with corners'' used by some of us \cite{ammann.ionescu.nistor:06, aln1, aln2}
is equivalent to being a \psbmanifold{} of boundary depth $0$ 
(\ie no connected component of $P$ is  contained in the boundary of~$M$).
They have an intrinsic definition (not using local coordinate charts).
\end{remark}

We shall need the following lemma.

\begin{lemma}\label{lemma.submanifolds}
Let $P \subset Q \subset M$ be manifolds with corners.
\begin{enumerate}[\kern5mm \rm (i)]
  \item\label{subm.ii} If both $P$ and $Q$ are \psbmanifolds{} of
    $M$, then $P$ is a \psbmanifold{} of $Q$.
  \item\label{subm.iii} If $P$ is a \psbmanifold{} of $Q$ and $Q$ is a
    \psbmanifold{} of~$M$, then $P$ is a \psbmanifold{} of~$M$.
\end{enumerate}
\end{lemma}

\begin{proof}
In order to prove \eqref{subm.ii}, we consider functions $x^1,\ldots, x^\ell$
defining the \psbmanifold{} $P$ of codimension $\ell$ in $M$ locally in
a neighborhood of $x\in P$. Choose $I\subset \{1,\ldots,\ell\}$ such
that $(dx^i|_p)_{i\in I}$ is a basis of $T^*_xQ$. Then in a possibly
smaller neighborhood, the functions $(x^i)_{i\in I}$ define~$P$ as a
\psbmanifold{} of~$Q$.

For~\eqref{subm.iii}, we consider functions $x^1,\ldots,x^k$ locally defining $P$
as a \psbmanifold{} of $Q$. We extend these functions to locally
defined functions on~$M$. Then we choose functions
$x^{k+1},\ldots,x^l$ defining $Q$ locally as a \psbmanifold{} of
$M$. Then $x^1,\ldots,x^l$ locally define~$P$ as a \psbmanifold{} of
$M$.
\end{proof}

\begin{example}
The diagonal $\Delta_N$ in Example~\ref{example.diagonal} is not a
\psbmanifold. If $N$ is the $2$-dimensional closed disc, then with
arguments analogous to Remark~\ref{rem.Bernd}, the diagonal is not a
\psbmanifold{} of $N\times N$. Alternatively, one could argue using
\cite{MelroseBook} (see Remark~\ref{rem.further.submfd}).
\end{example}

\subsubsection{The normal bundle of \psbmanifolds}
We now introduce some standard concepts that will be important in the 
definition of
the blow-up of a manifold with corners along a \psbmanifold{}. Recall
that $T^+M$ denotes the set of inward pointing tangent vectors to $M$,
see Remark~\ref{rem.tildeM}.

\begin{definition}\label{def.NPM}
Let $P \subset M$ be a \psbmanifold{} of the manifold with
corners~$M$. 
\begin{enumerate}[\kern3mm (i)]
\item The quotient $N^MP := TM\vert_P/TP$ is called the {\em normal
  bundle of $P$ in $M$.} 
  
\item The image $N^M_+P$ of $T^+M\vert_{P}$ in $N^M
P$ is called the {\em inward pointing normal bundle of $P$ in
  $M$.}
\item The set $\SS(N^M_+P)$ of directions 
in $N^M_+P$ is called the set of \emph{inward pointing spherical normal 
bundle of $P$ in $M$.} 
\end{enumerate}
\end{definition}

\begin{remark}\label{rem.SNMP}
As we have already noticed in Remark \ref{rem.tildeM}, $T^+M \vert_{P}
\to P$ is {\em not} a (locally trivial) fiber bundle over $P$. 
(Recall that $N^M_+ P$ is the image of $T^+ M\vert_P$ in $N^M P$ via 
the projection $TM\vert_P \to N^M P := TM\vert_P/TP$.) Then the projection 
map $N^M_+P \to P$ does define a (locally trivial) fiber bundle
structure over $P$ on $N^M_+P$, precisely because $P$ is a
\psbmanifold. The fibers of this bundle are modeled on $\RR^\ell_r$,
where $\ell$ is the codimension of $P$ and~$r$ is its boundary
depth. Thus the name of ``inward pointing normal bundle of $P$ in
$M$'' for $N^M_+P$ is justified.  Consequently, $\SS(N^M_+P)$, the
inward pointing spherical normal bundle of~$P$ in~$M$, is also a
(locally trivial) fiber bundle over~$P$ with projection $\SS(N^M_+P)
\to P.$ If $M$ is endowed with a Riemannian metric, then, we can endow 
$N^M_+ P\to P$ with the quotient metric, and hence 
we can of course identify $\SS(N^M_+P)$ with the set of unit vectors in $N^M_+P$. 
In particular, if, at $p \in P$, we have $\dim_p(P) = \dim_p(M)$, then the 
total space of this fiber  bundle is $\SS_p(N^M_+P) = \emptyset$.
\end{remark}


We complete this section with a related remark. We notice, in particular, the
existence of suitable ``tubular neighborhoods'' for \psbmanifolds.

\begin{remark}\label{eq.tubular.nbhd}
Let $P \subset M$ be a \psbmanifold{} in the manifold with corners
$M$. If $M$ is compact, then $P$ has a neighborhood $V_P \subset M$
such that $V_P$ is diffeomorphic to the closed cone $N^M_+P$ via a
diffeomorphism that sends $P$ to the zero section of $N^M_+P \to M$
and induces the identity at the level of normal bundles. This
was proved in \cite[Proposition~2.10.1]{MelroseBook}, under the
additional assumption that $P$ be closed. Moreover, the condition
that~$M$ be compact or that~$P$ be closed is not necessary, provided the right 
definition of tubular neighborhood is chosen (as neighborhood that 
is a disc bundle over $P$, which, in general, will not be the set of 
points of distance $<\epsilon$ for some Riemannian metric and some 
small~$\epsilon$). Then $N^M_+P$ is a cone with
corners in $N^MP$. Generalizing Example \ref{ex.corners}, we obtain
that all of the subsets $N^MP$, $N^M_+P$, and $S^+(N^MP)$ of $TM$ 
introduced in the last definition are manifolds with corners.
This is because the
property of being a manifold with corners is a local property and the
product of manifolds with corners is again a manifold with corners.
\end{remark}

We finish the discussion of \psbmanifolds{} 
by stressing that we allow the different connected components of a 
\psbmanifold{} to have \emph{different} dimensions. This is convenient when
considering intersections and when defining blow-ups.

\section{The blow-up for manifolds with corners}\label{sec.blowup.mfd.corners}
We now recall and study the blow-up of a manifold $M$ with corners by a 
\emph{closed} \psbmanifold. Our definition is the same as the one
in \cite{ACN, Kottke-Lin, MelroseBook}.

\subsection{Definition of the blow-up and its smooth structure} 
To define the blow-up, we first define the underlying set
and then we will define its smooth structure using as a model
the case of euclidean spaces.
The disjoint union of two subsets $A$ and $B$ will be
typically denoted $A \sqcup B$, as usual.

\subsubsection{Definition of the blow-up as a set}
We now define the underlying set of the blow-up $[M: P]$ of a manifold with
corners $M$ with respect to a closed \psbmanifold{} $P$ by replacing $P$ with the 
inward spherical normal bundle $\SS(N_+^MP)$ of $P$ in $M$ using the
disjoint union $\sqcup$.

\begin{definition}\label{def.blow-up}
Let $M$ be a manifold with corners and $P$ be a \emph{closed}
\psbmanifold{} of~$M$. Let $\SS(N^M_+P)$ be the inward pointing
spherical normal bundle of $P$ in $M$ (Definition \ref{def.NPM}).
As a set, we define the \emph{blow-up $[M: P]$ of~$M$ along $P$} 
(or \emph{with respect to $P$}) as the disjoint union 
\begin{equation*}
  [M:P] \ede (M\setminus P) \sqcup \SS(N^M_+P)\,.
\end{equation*}
The blow-down map $\beta = \beta_{M,P} : [M:P] \to M$ is defined as the
identity map on $M\setminus P$ and as the fiber bundle projection
$\SS(N^M_+P) \to P$ on the complement.
\end{definition}

Let us comment now on the definition of the blow-up.

\begin{remark}
The blow-up $[M: P]$ is therefore not defined if $P$ is not closed,
but we allow $P$ to consist of the disjoint union of several closed,
connected \psbmanifolds{} of~$M$ of possibly different dimensions.
Also, we are allowing the submanifold $P$ to be empty or to 
have components of the same dimension as $M$, which will be convenient
when considering iterated blow-ups, since it is not always possible to
arrange a family of submanifolds in a strictly increasing order
of dimensions. The components of $P$ of the same dimension with $M$ will,
of course, be connected components of~$M$ (since we have
assumed that $P$ is a closed \psbmanifold{} of~$M$). If $P\neq\emptyset$ is 
a union of connected components of~$M$, then $[M: P ] = M \smallsetminus P$
and the blow-down map $\beta_{M, P} : [M: P] \to M$ is not surjective
in this case (see Remark \ref{rem.SNMP}). In particular $[M:M]=\emptyset$, 
and hence the case $P = \emptyset$ is also needed when considering iterated blow-ups. 
In that case, $[M:\emptyset]= M$. 
\end{remark}

A general approach to smooth structures on the blow-up is contained in
\cite{KottkeMelrose}. Here we recall an approach that suffices for our
needs. We begin with the case of open subsets of a model space
$\RR^n_k$.

\subsubsection{The blow-up of the local models}
We now discuss various issues concerning the blow-up in local 
coordinates.
In the following, let $I_j$, $j=1,2,\ldots n$,
denote either $\RR$ or $[0,\infty)$. We will write $N_1\cong N_2$ if
$N_1$ is a \psbmanifold{} of $I_1 \times I_2 \times \cdots \times
I_n \subset \RR^n$ and if there is a permutation $\sigma$ of the
components of $\RR^n$ that induces a diffeomorphism from $N_1$ to
the \psbmanifold{} $N_2$ of $I_{\sigma(1)}\times I_{\sigma(2)} \times
\cdots\times I_{\sigma(n)}\subset \RR^n$. By contrast, when we write 
$N_1\simeq N_2$, we will merely state that the indicated manifolds 
are diffeomorphic, without including further information on the 
diffeomorphism. In particular, $N_1\cong N_2$ implies $N_1\simeq N_2$. 
To start with, the blow-up $[\RR^{n}_{k} \times \RR^{n'}_{k'} : 
\RR^n_k \times \{0\}]$ of $\RR^{n}_{k} \times \RR^{n'}_{k'} \cong
\RR^{n+n'}_{k+k'}$ along its \psbmanifold{} $\RR^n_k \times \{0\} =
\RR^n_k \times \{0_{\RR^{n'}}\}$ is, by Definition \ref{def.blow-up},
the set
\begin{equation}\label{eq.Blow.up}
  \begin{gathered}
  {[ \RR^{n}_{k} \times \RR^{n'}_{k'} : \RR^n_k \times \{ 0 \} ]}
  \ede \Big (\RR^{n}_{k} \times \RR^{n'}_{k'} \smallsetminus
  \RR^{n}_{k} \times \{0\} \Big ) \sqcup \RR^n_k \times
  \SS^{n'-1}_{k'}\\
  \seq \RR^n_k \times \left(\SS^{n'-1}_{k'} \sqcup \big (\RR^{n'}_{k'}
  \setminus \{0\} \big ) \right) \,.
  \end{gathered}
\end{equation}
Let us consider the map
\begin{equation}\label{eq.kappa}
  \begin{gathered}
  \kappa: \RR^{n}_k \times \SS^{n'-1}_{k'} \times [0, \infty)
    \ \to\ \RR^n_k \times \left(\SS^{n'-1}_{k'} \sqcup (\RR^{n'}_{k'}
    \setminus \{0\})\right) \,,\\
  \kappa (x, \xi, r) \ede
    \begin{cases}
        \ (x, \xi) \in \RR^n_k \times \SS^{n'-1}_{k'} & \mbox{if } r=0
        \\
      \, (x, r \xi) \in \RR^n_k \times (\RR^{n'}_{k'} \setminus \{0\})
      & \mbox{if } r > 0\,.
    \end{cases}
  \end{gathered}  
\end{equation}
The map $\kappa$ is immediately seen to be a bijection and we will use
it to endow $[ \RR^{n}_{k} \times \RR^{n'}_{k'} : \RR^n_k \times
  \{0\}]$ with the structure of a manifold with corners induced from
$\RR^{n}_k \times \SS^{n'-1}_{k'} \times [0, \infty)$. Under this
  diffeomorphism, the blow-down map becomes
\begin{equation}\label{eq.Blow.down2}    
  \beta : \RR^{n}_k \times \SS^{n'-1}_{k'} \times [0, \infty)\to
    \RR^{n}_{k} \times \RR^{n'}_{k'} \, , \quad \beta(x, \xi, r) \ede
    (x ,r \xi) \,.
\end{equation}
The blown-up space $ [\RR^{n}_{k} \times \RR^{n'}_{k'} : \RR^n_k
  \times \{0\}]$ is thus a space of ``generalized spherical
coordinates.''
  
If $U \subset \RR^{n}_{k} \times \RR^{n'}_{k'}$ is an open subset, we
endow
\begin{equation}\label{eq.Blow.up3}
  [U : U \cap (\RR^n_k \times \{0\})] \seq \beta^{-1}(U) \subset [
    \RR^{n}_{k} \times \RR^{n'}_{k'} : \RR^n_k \times \{0\}]
\end{equation}
with the induced structure of a manifold with corners.

\subsubsection{The topology and smooth structure of the blow-up}
The following lemmas will allow us to define a manifolds with corners
structure on blow-ups.

\begin{lemma}\label{lemma.ACN}
Let $P_i \subset M_i$, $i=1,2$, be closed \psbmanifolds{} and let
$\phi : M_1 \to M_2$ be a diffeomorphism such that $\phi(P_1) =
P_2$. Then there exists a unique map $\phi^\beta : [M_1:P_1] \to
[M_2:P_2]$ that is bijective and makes the following diagram commute
\begin{equation*}
\xymatrix{ [M_1: P_1] \ar[r]^{\phi^\beta} \ar[d]_{\beta_{M_1,P_1} }&
  [M_2:P_2] \ar[d]^{\beta_{M_2,P_2}} \\
  M_1 \ar[r]^\phi & M_2. }
\end{equation*}
This construction is functorial, in the sense that $(\phi \circ
\psi)^{\beta} = \phi^{\beta} \circ \psi^{\beta}$. If $M_i$ are open
subsets of $\RR^n_k$, then $\phi^\beta$ is a diffeomorphism.
\end{lemma}

\begin{proof}
The existence, uniqueness, and the functorial character of
$\phi^{\beta}$ follows from the definition of the blow-up. The fact
that $\phi^\beta$ is smooth if $M_i$ are open subsets of the model
space $\RR^{n}_{k}$ is the content of Lemma 2.2 of \cite{ACN}.
\end{proof}

Recall the maps $\phi^\beta$ of Lemma \ref{lemma.ACN}.

\begin{lemma}\label{lemma.atlas}
Let $\maA = \{(U_a, \phi_a) \, \vert \ a \in A\}$ be an
\emph{atlas} on a manifold with corners $M$ (Definition
\ref{def.atlas}). Let $P \subset M$ be a closed \psbmanifold{} and
$\beta = \beta_{M, P} : [M:P] \to M$ be the blow-down map. We endow
$[M: P]$ with the smallest topology that makes all the maps
$\phi_a^\beta$, $a \in A$, continuous ($\phi_a^\beta$ is defined on
$\beta^{-1}(U_a)$, see Lemma \ref{lemma.ACN}). Then
\begin{equation*}
  \beta^*(\maA) \ede \{(\beta^{-1}(U_a), \phi_a^\beta) \, \vert \ a
  \in A\}
\end{equation*}
is an atlas on $[M:P]$. If we take another
atlas $\maA'$ of~$M$ that is compatible with $\maA$, then
$\beta^*(\maA)$ and $\beta^*( \maA')$ will be compatible atlases
on $[M:P]$.
\end{lemma}

\begin{proof} This follows from Equation \eqref{eq.Blow.up3} and
  Lemma \ref{lemma.ACN}.
\end{proof}

Lemma \ref{lemma.atlas} thus yields the desired smooth structure on
$[M:P]$ that is moreover canonical (independent of any choices).

\begin{definition}\label{def.smooth.s}
Let $M$ be a manifold with corners and $P \subset M$ be a
closed \psbmanifold{}. We endow $[M:P]$ with the smooth structure defined by
the atlas $\beta^*(\maA)$ obtained from Lemma \ref{lemma.atlas},
for any atlas $\maA$ on~$M$.
\end{definition}

The following is a consequence of the definition.

\begin{corollary} \label{cor.dense}
With the notations of Definition \ref{def.smooth.s},
$M \smallsetminus P$ is dense in $[M : P]$.
\end{corollary}

This result remains true even if $P$ is a union of connected components
of~$M$, in which case $[M: P] = M \smallsetminus P$.
The smooth structure on $[M: P]$ is natural in the following strong
sense.

\begin{proposition}\label{prop.action}
With the notation of Lemma \ref{lemma.ACN}, we have that the map
$\phi^\beta$ is a diffeomorphism (in general, not just in the case of
open subsets of Euclidean spaces).
\end{proposition}

\begin{proof}
If $\maA$ is an atlas on $M_2$, then the pull-back of
$\beta^*(\maA)$ to $[M_1: P_1]$ is an atlas.
\end{proof}

The blow-up $[M:P]$ of a manifold with corners is thus again a
manifold with corners.

\subsection{Exploiting the local structure of the blow-up}
The local character of the definition of the smooth structure of the
blow-up $[M:P]$ of the manifold with corners $M$ along a
closed \psbmanifold{} $P$ means that most of the proofs involving blow-ups
can be conveniently treated by first treating the model case $P :=
\RR^n_k = \RR^n_k \times \{0\} \subset \RR^{n}_{k} \times
\RR^{n'}_{k'} = M$. To simplify notation, we shall often omit factors
of the form $\{0\}$ when there is no danger of confusion. This is the
case with the following results.

\subsubsection{The blow-down map is proper}
We shall need to prove that certain maps are closed.
This will be
conveniently done by proving that they are proper, since a proper map
between manifolds with corners is closed. In particular, we will show
that the blow-down map is proper.

Let $f : X \to Y$ be a continuous map between two Hausdorff
spaces. Recall that $f$ is called {\em proper} if $f^{-1}(K)$ is
compact for every compact subset $K \subset Y$. For instance, the map
$\beta$ of Equation \eqref{eq.Blow.down2} is immediately seen to be
proper.

\begin{corollary}\label{cor.beta.proper}
Let $P$ be a closed \psbmanifold{} of a manifold with corners $M$. The
blow-down map $\beta_{M, P} : [M:P] \to M$ is proper.
\end{corollary}

\begin{proof} 
Using Lemma~\ref{lemma.loc.prop} from the Appendix, we see that we can
treat the problem in local coordinates. Then, in local coordinates,
the blow-down map is given by Equation~\eqref{eq.Blow.down2}, which is
a proper map, as we have already pointed out.
\end{proof}

\subsubsection{Blow-ups and products}
We have a simple, convenient behavior of the blow-up with respect to
products. 
\begin{lemma}\label{lemma.product} 
Let $M$ and $M_1$ be two manifolds with corners and $P$ be a closed
\psbmanifold{} of~$M$. Then $P \times M_1$ is a closed
\psbmanifold{} of $M \times M_1$ and there exists a canonical
diffeomorphism $[M \times M_1 : P \times M_1] \simeq [M: P] \times M_1$
such that the following diagram commutes:
\begin{equation}\label{eq.CD}
\begin{CD}
  [M \times M_1 : P \times M_1] @>{ \simeq }>>[M: P] \times M_1 \\
  @V{ \beta_{M \times M_1, P \times M_1}}VV
  @VV{\beta_{M,P} \times \id}V \\
   M \times M_1 @>{ \id }>> M \times M_1 \,.
\end{CD}
\end{equation}
\end{lemma}

\begin{proof}
Since the result is a local one and $P$ is a closed \psbmanifold{} of~$M$,
it is enough to treat the case
\begin{align*}	
     M & \ede \RR^{m}_{k_m} \times \RR^{p}_{k_p} \\
     P\, & \ede \{ 0_{\RR^m} \} \times \RR^{p}_{k_p} \subset M \\
     M_1 & \ede \RR^{l}_{k_l} \,.
\end{align*}
In this local treatment, we will again write $\cong$ to stress that a given
diffeomorphism is given by a permutation of coordinates, as in Equation \eqref{eq.def.can}.

With this choice, we see that $P \times M_1$ is a closed \psbmanifold{} of $M
\times M_1$. We have natural diffeomorphisms with the first one
  being obtained from the definition of the blow-up, Definition
  \ref{def.blow-up}, and the last being induced by suitable
  permutations of coordinates
\begin{multline*}
  [M \times M_1: P \times M_1 ] \ = \ [ \RR^{m}_{k_m} \times
      \RR^{p}_{k_p} \times \RR^{l}_{k_l} : \{ 0_{\RR^m} \} \times
      \RR^{p}_{k_p} \times \RR^{l}_{k_l} ]  \\
  \ = \ \SS^{m-1}_{k_m} \times \RR^{p}_{k_p} \times
    \RR^{l}_{k_l} \sqcup \left((\RR^{m}_{k_m} \times
      \RR^{p}_{k_p} \times \RR^{l}_{k_l} ) \setminus (\{ 0_{\RR^m} \}
    \times \RR^{p}_{k_p} \times \RR^{l}_{k_l} ) \right)  \\
  \ \simeq \
  \SS^{m-1}_{k_m} \times [0, \infty) \times \RR^{p}_{k_p} \times \RR^{l}_{k_l} 
  \ \cong\  \SS^{m-1}_{k_m} \times \RR^{p+l+1}_{k_p+k_l+1}
\end{multline*}
and 
\begin{multline*}
  [M:P] \ = \ [\RR^{m}_{k_m}\times \RR^{p}_{k_p}: \{ 0_{\RR^m} \}
    \times \RR^p_{k_p}]\\
  \ = \ \SS^{m-1}_{k_m} \times \RR^p_{k_p} \sqcup \left (
  (\RR^{m}_{k_m}\times \RR^{p}_{k_p}) \setminus (\{ 0_{\RR^m} \}
  \times \RR^{p}_{k_p}) \right ) \\
  \simeq \ \SS^{m-1}_{k_m} \times [0, \infty) \times \RR^p_{k_p}
    \ \cong\ \SS^{m-1}_{k_m} \times \RR^{p+1}_{k_p+1} \,.
\end{multline*}
The desired diffeomorphism $[M \times M_1 : P \times M_1]
\stackrel{\simeq}{\longrightarrow}[M : P ] \times M_1$ is then induced
by the above diffeomorphisms and the canonical permutation of
coordinates diffeomorphism $\SS^{m-1}_{k_m} \times \RR^{p+1}_{k_p+1}
\times \RR^{l}_{k_l} \stackrel{\cong}{\longrightarrow} \SS^{m-1}_{k_m}
\times \RR^{p+l+1}_{k_p+k_l+1}$ of Equation \eqref{eq.def.can}.
\end{proof}

The functoriality property of Lemma \ref{lemma.ACN} and Lemma \ref{lemma.product}
then gives the following result.

\begin{proposition}\label{prop.group.action}
Let $G$ be a Lie group acting smoothly on~$M$ (that is, such that the 
action map $G\times M\to M$ is smooth). Let $P \subset M$ be a closed \psbmanifold{} such
that $g(P) = P$ for all $g \in G$. Then the action of $G$  lifts to a
smooth action of $G$ on $[M:P]$.
\end{proposition}

\begin{proof}
We extend the action map to a diffeomorphism
\begin{equation*}
    a : G\times M \to G \times M\,, \quad (g,x)\mapsto (g,gx)\,.
\end{equation*}  
Note that $G\times P$ is a closed \psbmanifold{} of $G\times M$. Thus our functorial
property implies that this map lifts to a map 
\begin{equation*}
    [G \times M : G \times P] \to [G\times M:G\times P] \, .
\end{equation*}
Using the natural diffeomorphism  $[G\times M: G\times P] \simeq G \times [M:P]$
of Lemma \ref{lemma.product}, we obtain a smooth map
$$\hat a: G\times [M:P]\to G\times [M:P].$$
Thus the second component $\hat a_2$ defines a smooth map  $G\times [M:P]\to [M:P]$. 
It satisfies all axioms of an action since it is an extension of the action map for 
the action of $G$ on $M\setminus P$ and since $M\setminus P$ is dense in $[M:P]$.
\end{proof}

\subsection{Cleanly intersecting families and liftings}
In order to avoid complications, we shall consider
the blow-up with respect to ``cleanly intersecting families''
of closed \psbmanifolds{} and ``clean semilattices,'' some concepts
that we recall below.

\subsubsection{Clean intersections}
We continue to exploit the local structure of the blow-up. Recall the
following standard definition.

\begin{definition}
\label{def.clean.intersection}
Let $M$ be a manifold with corners and $X_1, X_2, \ldots, X_k \subset
M$ be \psbmanifolds. We shall say that $X_1, X_2, \ldots, X_k$ {\em
  have a clean intersection} or that they {\em intersect
    cleanly} if
\begin{enumerate}[\kern5mm (i)]
  \item\label{def.clean.intersection.i} $Y := X_1 \cap X_2 \cap \ldots
    \cap X_k $ is a \psbmanifold{} of~$M$ (possibly empty),
  \item\label{def.clean.intersection.ii} for all $x \in Y$, $T_x Y = T_x X_1 \cap T_x X_2 \cap \ldots
      \cap T_x X_k$.
\end{enumerate}
\end{definition}

We consider conditions~\eqref{def.clean.intersection.i} and~\eqref{def.clean.intersection.ii} of  Definition~\ref{def.clean.intersection} to be automatically satisfied if $Y :=
X_1 \cap X_2 \cap \ldots \cap X_k = \emptyset$. Similar conditions
appear in \cite{ACN}, Definition~2.7, where they were used to define a
\emph{weakly transversal family} of connected submanifolds with
corners. We shall need also the notion of a ``cleanly intersecting
family'' (Definition~\ref{def.clean.lattice}), which roughly states
that every subfamily intersects cleanly.

\begin{lemma}\label{lem.clean.inters.p-subm}
Let $P$ and $Q$ be \psbmanifolds\ of~$M$ intersecting
cleanly. Then $P \cap Q$ is a \psbmanifold\ of $Q$ (and also for $P$). 
\end{lemma}

\begin{proof}
According Definition~\ref{def.clean.intersection}
\eqref{def.clean.intersection.i} $P\cap Q$ is a \psbmanifold\ of
$M$. Then Lemma~\ref{lemma.submanifolds} \eqref{subm.ii} states that
$P\cap Q$ is also a \psbmanifold{} of $Q$.
\end{proof}

\subsubsection{Liftings of \psbmanifolds{} to blow-ups}
We now consider the lifting of suitable \psbmanifolds{} in $M$ to $[M: P]$ as in 
\cite{KottkeMelrose, MelroseBook}. The local model for such lifts is given by the 
following lemma. 

\begin{lemma}\label{lemma.lin.lift}
Let $k' \le k''$ and $n' - k'\le n'' - k''$, so that the canonical
first components inclusion $j_0 : \RR^{n'}_{k'} \to \RR^{n''}_{k''}$ 
of Equation \eqref{eq.def.can0} is defined, and let $j := (\id \times j_0 , 0) : 
\RR^{n}_{k} \times \RR^{n'}_{k'}  \to \RR^{n}_{k} \times \RR^{n''}_{k''} \times \RR_\ell^p$.
Then there is a unique smooth map $j^{\beta}$ such that the diagram
\begin{equation*}   
\begin{CD}
   [\RR^{n}_{k} \times \RR^{n'}_{k'} : \RR^n_k \times \{0\}]
   @>{j^{\beta}}>> [ \RR^{n}_{k} \times \RR^{n''}_{k''} \times \RR_\ell^p : \RR^n_k \times
     \{0\} \times \RR_\ell^p ] \\
   @V{ \beta_{\RR^{n}_{k} \times \RR^{n'}_{k'} , \RR^n_k \times
      \{0\}}}VV @VV{\beta_{\RR^{n}_{k} \times \RR^{n''}_{k''} \times \RR_\ell^p ,
      \RR^n_k \times \{0\} \times \RR_\ell^p
      }}V \\
   \RR^{n}_{k} \times \RR^{n'}_{k'} @>{j}>>\RR^{n}_{k} \times
   \RR^{n''}_{k''} \times \RR_\ell^p\,.
\end{CD}
\end{equation*}
commutes ($0 \le \ell \le p$). Moreover, $j^\beta$ is a diffeomorphism onto its image,
which is a closed \psbmanifold{} of $[ \RR^{n}_{k} \times \RR^{n''}_{k''} : \RR^n_k \times
     \{0\}]$.
\end{lemma}

The definition of the map $j^\beta$ extends the definition
in Lemma~\ref{lemma.ACN}.


\begin{proof} 
If we can prove the result for $p = \ell = 0$, then Lemma~\ref{lemma.product}
will give the result in general. 
Clearly, $j_0$ is a particular case of the maps $j$ (more precisely the case $n=k=0$), but we will treat it first. Let us first notice that
in the statement, we have not specified to which sets the various neutral 
elements belong. For instance, the desired map $j_0^\beta$
acts as $j_0^\beta : [\RR^{n'}_{k'} : \{0_{\RR^{n'}} \}] \to [\RR^{n''}_{k''} : 
\{0_{ \RR^{n''} } \}]$. However, in what follows, we will drop the indices of 
the neutral elements, as there is no danger of confusion. To define
$j_0^\beta$ and to obtain its properties, we shall use the definition of the
blow-up in local coordinates and the setting of Example \ref{ex.first.comp} to define
\begin{equation*} 
  j_0^\beta : [\RR^{n'}_{k'} :\{0\}] \simeq \SS^{n'-1}_{k'} \times [0, \infty)
    \stackrel{i\times \id}{-\!\!-\!\!\!\longrightarrow}\SS^{n''-1}_{k''} \times
             [0, \infty) \simeq [ \RR^{n''}_{k''} : \{0\}] \,,
\end{equation*}
where $i:\SS^{n'-1}_{k'} \to \SS^{n''-1}_{k''}$ is the restriction of $j_0$.
With this definition, we obtain that $j_0^\beta$ fits into the commutative
diagram
\begin{equation*}  
\begin{CD}
   [\RR^{n'}_{k'} : \{0\}] @>{j^\beta_0}>> [\RR^{n''}_{k''} : \{0\}] \\
  @V{ \beta_{\RR^{n'}_{k'}, \{0\}}}VV @VV{\beta_{\RR^{n''}_{k''},
      \{0\}}}V \\
 \RR^{n'}_{k'} @>{j_0}>>\RR^{n''}_{k''}\,,
\end{CD}
\end{equation*}
As in Example \ref{ex.first.comp}, $j_0^\beta$ is seen to be a diffeomorphism onto 
its image, which is a closed \psbmanifold.

Coming back to the general case, let $\id:\RR^{n}_{k}\to \RR^{n}_{k}$ be the identity 
map. We then let $j^{\beta}$ be the composition of $\id \times j_0^\beta : 
\RR^{n}_{k} \times [\RR^{n'}_{k'} :\{0\}] \to \RR^{n}_{k} \times [ \RR^{n''}_{k''} : \{0\}]$ with
the canonical diffeomorphisms $\RR^{n}_{k} \times [\RR^{n'}_{k'} :\{0\}] 
\simeq  [\RR^{n}_{k} \times \RR^{n'}_{k'} : \RR^{n}_{k} \times\{0\}]$
and $\RR^{n}_{k} \times [ \RR^{n''}_{k''} : \{0\}] \simeq  [ \RR^{n}_{k} \times \RR^{n''}_{k''} : 
\RR^{n}_{k} \times \{0\}]$ of Lemma~\ref{lemma.product}.
Then $j^\beta$ is smooth,
by its definition. Lemma~\ref{lemma.product} also gives the
commutativity of the diagram and proves that the
image of $j^\beta$ is a \psbmanifold, and hence that $j^\beta$ is
a diffeomorphism onto its image.
\end{proof}

We have the following result on the blow-up of \psbmanifolds, due, in
part, to Melrose \cite[Chapter 5, Section 7]{MelroseBook}. A proof in
a slightly less general setting can be found also in Proposition~2.4
of \cite{ACN}. For a \psbmanifold{} $P \subset M$, recall the
definition of $\SS(N^M_+P)$, the {\em inward pointing normal bundle of
  $P$ in $M$} from Definition \ref{def.NPM}.

\begin{proposition}\label{prop.beta.m1}
Let $P$ and $Q$ be closed \psbmanifolds\ of~$M$ intersecting
cleanly. Let $j : Q \to M$ be the inclusion. There exists a unique continuous map
\begin{equation*}
  j^{\beta} : [Q : P \cap Q] \to [ M : P ] \,
\end{equation*}
\begin{equation*}
\end{equation*}
such that $\beta_{M, P} \circ j^\beta = j \circ \beta_{Q, P \cap Q}$.
The map $j^\beta$ is an injective immersion and its image is a closed \psbmanifold. Moreover
\begin{equation*}
   \overline{\beta_{M,P}^{-1}(Q \smallsetminus P)} \seq j^\beta([Q: P \cap Q]).
\end{equation*}
\end{proposition}

In the following we will identify $[Q : P \cap Q]$ with  $j^{\beta} 
([Q : P \cap Q])$.

\begin{proof} 
Let $P_0 := P \cap Q$, to simplify the notation. To define the map
\begin{equation*}
  j^{\beta} : [Q : P_0] \ede (Q \smallsetminus P_0)
  \sqcup \SS(N^Q_+ P_0 ) \to (M\smallsetminus P) \sqcup
  \SS(N^M_+P) \, =: \, [ M : P ]\,,
\end{equation*}
it is enough to define it on $Q \smallsetminus P_0$ and on $\SS( N^Q_+ P_0 )$.
First, on $Q \smallsetminus P_0$, we let
$j^\beta$ to be the inclusion $Q \smallsetminus P_0  := Q \smallsetminus (P \cap Q) 
\to M \smallsetminus P$. 
The inclusion of $Q$ into $M$ also induces an inclusion $TQ
\to TM$, extending the inclusion $TP_0 \to T P$. Since
$T P_0 = TP \cap TQ$, by the assumption that $P$ and $Q$ intersect
cleanly, we can pass to quotients to obtain an
vector bundle map
\begin{equation*}
  \psi : N^Q P_0 \ede TQ|_{P_0}/TP_0 \seq TQ|_{P_0}/(TP\cap TQ) \, 
  \to \, TM|_P/TP \, =: \, N^M P\,,
\end{equation*}
which is injective, immersive, and a homeomorphism onto its image,
which is, moreover, a closed subset of $N^M P$. Next, the map $\psi$ restricts to an embedding 
$\SS(N^Q_+ P) \to \SS(N^M_+P)$, which defines $j^\beta$ on $\SS(N^Q_+ P)$. This 
completes the definition of $j^\beta$. This definition shows right away that 
$j^\beta$ is injective and that it satisfies the relation $\beta_{M, P} \circ j^\beta 
= j \circ \beta_{Q, P_0}$.

We next want to show that $j^\beta$ is smooth and that $j^\beta([Q: P \cap Q])$ is a  
\psbmanifold{} of $[M: P]$. The locality of the blow-up and from Lemma \ref{lemma.product} it follows 
that it is sufficient to prove this in the case
$Q := \RR^{n}_{k} \times \RR^{n'}_{k'} \times \{0\} \times \{0\}
\subset M := \RR^{n}_{k} \times \RR^{n'}_{k'} \times \RR^{n''}_{k''} \times \RR_\ell^p$
and $P := \RR^{n}_{k} \times \{0\} \times  \{0\} \times \RR_\ell^p$. Lemma \ref{lemma.lin.lift} 
(a standard calculation in local coordinates, but note that $n''$ in that
Lemma is $n' + n''$ in our definition) then gives the desired statement
(that $j^\beta$ is smooth and that $j^\beta([Q: P_0])$ is a \psbmanifold,
where, we recall, $P_0 = P \cap Q$).
To prove that $j^\beta([Q : P_0])$ is closed in $[M: P]$, we can use again Lemma \ref{lemma.product}.
Indeed, let $x_n \in j^\beta([Q : P_0])$ be a sequence convergent to some
$y \in [M: P]$. We want to show that $y \in j^\beta([Q : P_0])$.
Let $z := \beta_{M, P}(y)$, which is the limit of $\beta_{M, P}(x_n)$.
Then the result follows using local coordinates around $z$ compatible with
Lemma \ref{lemma.lin.lift}.


We now come back to the general case. Since $j^\beta$ is smooth, it is also 
an immersion, by its definition, since we have already
noticed that it is an immersion on the two sets defining $[Q : P \cap Q]$.
(An alternative argument would be to use the local description 
of the blow-up in terms of half-spaces as in \cite{ACN}.) It follows that
$j^\beta$ is a diffeomorphism onto its image. We also obtain that
$j^\beta$ is unique, since it is continuous (even differentiable!) 
and since $Q \smallsetminus (P \cap Q)$ is dense in $[Q : P \cap Q]$, 
by Corollary \ref{cor.dense}. The equality $\overline{\beta_{M,P}^{-1}
(Q \smallsetminus P)} \seq j^\beta([Q: P \cap Q])$ follows since the right
hand side is a closed set containing $Q \smallsetminus (P \cap Q)$.
\end{proof}

\begin{definition}\label{def.pull-back}
Let $P$ be a closed \psbmanifold{} of~$M$ and $Q$ be a closed subset of~$M$.
The {\em lifting} $\beta_{M, P}^* (Q)$ of $Q$ in $[M:P]$ is defined by
\begin{equation*}
  \beta_{M, P}^* (Q) \ede \overline{\beta_{M,P}^{-1}(Q \smallsetminus P)}\,.
\end{equation*}
(The closure is in $[M: P]$.)
\end{definition}

\begin{remark}\label{betaMPQremark}
Of course, in general, $\beta_{M, P}^* (Q)$ will not be a 
submanifold of some sort (of $[M: P]$), even if $Q$ is one. 
However, if $Q$ is a closed \psbmanifold{} of~$M$ and $P$ and $Q$
intersect cleanly, then $\beta_{M, P}^* (Q) \simeq j^\beta \big ( [Q: P \cap Q] \big )$
and will hence be a 
\psbmanifold{} of~$M$, by Proposition \ref{prop.beta.m1}.
Also, it should be pointed out that, if $Q\subset P$, the definition 
above of $\beta_{M, P}^* (Q)$ differs 
from Melrose's one in \cite[Chapter~5, Section~7]{MelroseBook}. 
More precisely, if $Q \subset P$, our definition is such that 
$\beta_{M,P}^*(Q) = \emptyset$, whereas it is $\beta_{M,P}^{-1}(Q)$ 
in Melrose's unpublished book. In part for this reason, we will avoid 
considering the case $Q \subset P$
when dealing with $\beta_{M,P}^*(Q)$. A first advantage of our definition is
that it preserves the inclusions. Another advantage of our
definition is that it is local in the sense that, for any open subset 
$U\subset M$, we have
\begin{equation*}
  \beta_{U,P\cap U}^*(Q\cap U)  \seq \beta_{M, P}^* (Q)\cap  
  \beta_{M, P}^{-1} (U)\,.
\end{equation*}
Moreover, in Melrose's definition, locality may fail if $Q$ is not connected.  
However, with our definition, in general $\bigl[[M: Q] : P\bigr] \neq \bigl[[M: P] : Q\bigr]$
if $Q \subsetneq P$.
\end{remark}

\section{The iterated and the graph blow-ups}\label{sec.graph.blowup}

We introduce, study, and compare 
in this section two types of blow-ups with respect 
to more than one submanifold: the {\em iterated blow-up} and the 
{\em graph blow-up}. The graph blow-up $\bl{M: \maP}$ 
has the advantage that it is defined in great generality and is
obviously independent of the order on the family of closed \psbmanifolds{}
$\maP$, up to an isomorphism. However, it is not clear whether
it has a smooth structure. The iterated blow-up, on the other
hand, comes with a smooth structure, but may not be always
defined and, at least in our approach, depends on the order
on the elements of the family $\maP$.

\subsection{Definition of the iterated blow-up}
Recall the definition of the lifting $\beta^* (Q) = \beta_{M, P}^*(Q)
:= \overline{Q \smallsetminus P} \subset [M : P]$ (closure in $[M:
  P]$), Definition \ref{def.pull-back}. We fix a manifold with
corners $M$. We now introduce the {\em iterated version of the
  blow-up}; it is different from the one in \cite{Kottke-Lin, MelroseBook}
  since we are using a different type of pull-back $\beta^*$. 

\begin{definition}\label{def.iterated.bu}
Let $\maP := (P_i)_{i=1}^k$, $P_i \subset M$, be a $k$-tuple of subsets of~$M$ with $P_1$ a closed \psbmanifold{} of~$M$, 
let $\beta_1 \ede \beta_{M, P_1} : [M:P_1] \to M$ be the associated blow-down map, and let
  $\maP' := \bigl(\beta_1^*(P_2), \beta_1^*(P_2), \ldots, \beta_1^*(P_k)\bigr)$,
  which is a $(k-1)$-tuple).
If $k = 1$ or if $\big[[M:P_1]: \maP' \big]$ is defined, then we define 
by induction on $k$ the \emph{iterated 
blow-up $[M : \maP]$ of~$M$ with respect to} or \emph{along}
$\maP$ by
\begin{equation*}
  [M: \maP] \seq [M : (P_i)_{i=1}^k] \ede
  \begin{cases}
    \ [M: P_1] & \mbox{ if }\ k = 1\,,\\
   \,  \big[[M:P_1]: \maP' \big]  &
   \mbox{ if }\ k > 1\,.
   \end{cases}
\end{equation*}
\end{definition}

Let us now discuss this definition.

\begin{remark} \label{rem.repetition} 
The following comments should be compared with the analogous ones
for the graph-blow-up (Remark \ref{rem.repetition2}).
\begin{enumerate}
\item We can remove from the $k$-tuple $\maP$ any element $P_j$ for which there 
exists a {\em larger} element $P_i$ preceding it (i.e., there exists $i < j$
with $P_i \supset P_j$) without affecting the resulting blown-up space. 
To see that, we notice that $\beta_{M, P}(Q) = \emptyset$ 
if $Q \subset P$ and $[Q: \emptyset] = Q$. 

\item Similarly, we can also remove any 
occurence of the empty set from the $k$-tuple without affecting the resulting
blow-up. In particular, removing or adding
repetitions to a $k$-tuple $\maP$ or removing or adding empty sets to
the $k$-tuple  will not affect whether $[M: \maP]$
is defined or not. (When removing repetitions, we always keep
the first element in the repeated sequence and remove only the ones following it.)

\item
Let us say that the $k$-tuple  $\maP$ is 
{\em degenerate} if it contains repetitions. Otherwise,
we shall say that it is {\em non-degenerate}. Thus, we can
always replace our $k$-tuple  $\maP$ with non-degenerate one without
changing the result of the iterated blow-up (in particular, not changing
the whether it is defined or not). 
We then notice that $(k-1)$-tuple $\maP'$ may be degenerate even if $\maP$ is 
non-degenerate. An example is provided by
the triple $\maP := (P, Q, P \cup Q)$, where $P$ and $Q$ are {\em disjoint}
\psbmanifolds{} of~$M$. 

\item We also notice that giving a 
$k$-tuple without repetitions of \psbmanifolds{} of~$M$ is equivalent to giving a totally 
ordered finite set (or family) of \psbmanifolds{} of~$M$ and that, when considering blow-ups, 
it is enough to consider only such totally ordered sets.

\item
Finally, in \cite{Kottke-Lin}, Kottke has introduced a {\em size order} as
one that satisfies
\begin{equation*}   
  P_i \subsetneq P_j \ \Rightarrow i < j
\end{equation*}
and used it to define his version of iterated blow-up.
(His version of iterated blow-up, however, is different from ours since he considered 
a different pull-back map $\beta^*$.) The above discussion shows that,
given a $k$-tuple $\maP$, there always exists another $j$-tuple $\maP_0$, $j \le k$,
that is size-ordered and such that $[M: \maP] = [M: \maP_0]$ (this statement
subsumes the fact that if one of these blow-ups is defined, then so is the other).
\end{enumerate}
\end{remark}


We shall also write
\begin{equation*}
   [M: \maP] \seq [M : (P_i)_{i=1}^k] \, =: \, [M : P_1,P_2, \ldots, P_k] \,,
\end{equation*}
and hence, using the pull-back by the map $\beta_1$, we have
\begin{equation*}
    [M : P_1,P_2, \ldots, P_k] \ede \bigl[[M:P_1] :
      \beta^{*}_1(P_{2}),\ldots, \beta^{*}_1(P_k)\bigr]\,.
\end{equation*}
We generalize this relation in the following remark.

\begin{remark}\label{rem.cond.ex}
Let us define by induction first $\beta_1 := \beta_{M, P_1} : [M: P_1] \to M$,
and $\gamma_1 := \beta^{*}_1$ and then, for $k \ge 2$, 
\begin{equation*}
  \beta_k \ede \beta_{[M,P_1,P_2, \ldots, P_{k-1}], \gamma_{k-1}^*(P_k) } : [M :
    P_1,P_2, \ldots, P_k] \to [M : P_1, P_2, \ldots, P_{k-1}]
\end{equation*}
and $\gamma_k := \beta_k^{*} \circ \gamma_{k-1} = \beta_k^{*} \circ ... \circ \beta_1^{*}$.
(In particular, $\beta_2 := \beta_{[M: P_1], \beta_1^*(P_2)}
: [[M:P_1] : \beta_1^*(P_2)] \to [M: P_1]$.)
We then have
\begin{align*}
  [M : P_1,P_2, \, \ldots \, , P_j]\ & =\ \bigl[[M:P_1] :
    \gamma_1(P_{2}),\ldots, \gamma_1(P_j)\bigr]\\
  &=\ \Bigl[\bigl[[M:P_1] : \gamma_1( P_{2})\bigr] : \gamma_2(P_3) , \, \ldots \, ,
    \gamma_{2}(P_j)\Bigr] \\
  & = \ \, \ \ldots \\
  & =\ [[\, \ldots [[[M:P_1] : \gamma_1( P_{2})]: \gamma_2(P_3) ]
      \ \ldots \ ] : \gamma_{j-1}(P_j)]\,.
\end{align*}
Note that $[M : P_1,P_2, \ldots, P_j]$ is always defined if $j
= 1$. Then the condition that the iterated blow-up $[M : P_1,P_2,
  \ldots, P_j]$ be defined can then be formulated by induction as
follows:
\begin{enumerate}[(i)]
  \item the iterated blow-up $[M: P_1, \ldots ,P_{j-1}]$ is defined,
    and
  \item the lift $\gamma_{j-1}(P_j)$ is a closed
    \psbmanifold{} of $[M: P_1, \ldots ,P_{j-1}]$.
\end{enumerate}
\end{remark}

%

\subsection{Disjoint submanifolds}
A particularly simple instance of the iterated blow-up is when
the corresponding manifolds are disjoint. Nevertheless, this is an
important case that will be discussed in this subsection. This
will also alow us to introduce and prove the first factorization
lemma, Lemma \ref{lemma.factorization1}. It also allows us to deal
with the case of a closed \psbmanifold{} that has connected components of 
different dimensions, which, we recall, is allowed. In particular, 
blowing-up with respect to such a manifold amounts, as we
will see, to blowing up successively with respect to each component.

We need first to discuss the gluing of open subsets. Let us assume
that we have two manifolds with corners $M_1$ and $M_2$ and that $U_i
\subset M_i$ are open subsets ($i=1,2$). Let us also assume that we
are given a diffeomorphism $\phi : U_1 \to U_2$. Then we define
\begin{equation}
  \begin{gathered}
    M_1 \cup_{\phi} M_2 \ede (M_1 \sqcup M_2)/\{ x \equiv \phi(x) \,
    \vert \ x \in U_1\}\,,\\
   M_1 \cup_{\id} M_2 \, =: \, M_1 \cup_{U_1} M_2\,, \ \ \mbox{ if
   }\ U_1 = U_2 \ \mbox{ and }\ \phi \mbox{ is the identity map } \id\,.
   \end{gathered}
\end{equation}
If $\phi$ is the identity, we shall call $M_1 \cup_{U_1} M_2$ the {\em
  union of $M_1$ and $M_2$ along $U_1 = U_2$.} Under favorable
circumstances (but not always), $M_1 \cup_{\phi} M_2$ is also a
manifold with corners.

We have the following simple lemma.

\begin{lemma}\label{lemma.union}
Let $M$ be a manifold with corners (and hence Hausdorff) and $M_i
\subset M$, $i=1,2$, be open subsets with $U:=M_1\cap M_2$ and
$M_1\cup M_2=M$. Then there exists a unique structure of a manifold
with corners on $M_1 \cup_U M_2$ that induces the given smooth
structures on $M_i$, and hence we have a canonical
diffeomorphism $M_1 \cup_U M_2 \simeq M$.
\end{lemma}

\begin{proof}
Let $\maA_i$ be an atlas for $M_i$. Then their union $\maA
  := \maA_1 \cup \maA_2$ is an atlas for~$M$. It is also an
  atlas for any manifold with corners structure on $M_1 \cup_U M_2$
  that induces the given one on each $M_i$. Hence the desired manifold
  with corners structure on $M_1 \cup_U M_2$ is given by the union
  $\maA := \maA_1 \cup \maA_2$.
\end{proof}

This lemma allows us to ``commute'' the blow-ups with
respect to disjoint closed \psbmanifolds. We thus have the following
simple result, see for example \cite{ACN, Kottke-Lin, MelroseBook}.

\begin{lemma}\label{lemma.disjoint}
Let $P$ and $Q$ be closed \psbmanifolds{} of~$M$ and $\beta_{M, Q} : [M:Q] \to M$ be 
the blow-down map. Assume that $P \cap Q = \emptyset$. Then $\beta_{M, Q}^*(P) :=
\beta_{M, Q}^{-1}(P) = P$ and the iterated blow-up $[M: Q, P] := \bigl[[M: Q]: P\bigr]$ is
defined and diffeomorphic to $([M:Q] \smallsetminus P) \sqcup_{M
  \smallsetminus (P \cup Q)} ([M:P] \smallsetminus Q)$, the union of
$[M:Q] \smallsetminus P$ and $[M:P] \smallsetminus Q$ along
$M\smallsetminus (P \cup Q)$, a common open subset. In particular, by symmetry, 
\begin{equation*}
  \bigl[[M:P]:Q\bigr]  \ \seq \ [M: P \cup Q] \ \seq \ \bigl[[M:Q]:P\bigr] \,,
\end{equation*}
with the same smooth structure.
\end{lemma}

The idea of the proof is to use the fact that the
construction of the blow-up is local in nature.

\begin{definition}
Suppose $f_i : X \to Y_i$, $i=1, \ldots , N$, are continuous maps. We
say $(f_1,\ldots,f_N): X \to \prod_{i=1}^N Y_i$, $x\mapsto
(f_1(x),\ldots,f_N(x))$ is proper in each component if each $f_i$ is
proper.
\end{definition}

We shall need the following ``factorization'' lemma.
The existence of the map $\zeta_{M, Q, P}$ of follows from Lemma \ref{lemma.disjoint} or
from a more general result in \cite{Kottke-Lin}.

\begin{lemma}[The first factorization Lemma] 
\label{lemma.factorization1}
Let us assume that $P$ and $Q$ are closed, {\em disjoint}
\psbmanifolds{} of~$M$. Then there exists a unique, smooth, natural
map
\begin{equation*}
  \zeta_{M, Q, P} : \bigl[[M: Q]: P\bigr] \ \to \ [M : P]
\end{equation*}  
that restricts to the identity on $M \smallsetminus (P \cup
  Q)$.
Moreover, the product map
\begin{equation*}
  \maB_{M, Q, P} \ede ( \zeta_{M, Q, P}, \beta_{[M:Q], P}) :
  \bigl[[M: Q]: P\bigr] \to [M : P] \times [M : Q]
\end{equation*}
is proper in each component. Its image is a \wsbmanifold{} in the
sense of Definition~\ref{def.weak-submanifold} and $\maB_{M, Q, P}$ is
a diffeomorphism onto its image.
\end{lemma}

Again, our main focus lies on the case $\dim (P) < \dim (M)$ and $\dim (Q) < \dim (M)$. 
The statements, however remain (trivially) true if one of these dimensions 
is equal to $\dim (M)$ (equivalently, if a connected component of~$M$ is contained in 
$P$ or $Q$). Then this component is removed both from $\bigl[[M: Q]: P\bigr]$ 
and from  $[M : P]$ or $[M : Q]$.

\begin{proof}
Lemma \ref{lemma.disjoint} states that $\bigl[[M: Q]: P \bigr] = [M :
  P \cup Q] = \bigl[[M: P]: Q\bigr]$. This gives $\zeta_{M, Q, P} =
\beta_{[M: P], Q}$.  In particular, $\zeta_{M, Q, P}$ is proper, by
Corollary \ref{cor.beta.proper}. The map $\beta_{[M: Q], P}$ is proper
by Corollary \ref{cor.beta.proper}. As $P$ and $Q$ are disjoint, at
each point, at least one component of $\maB_{M, Q, P} = ( \zeta_{M, Q,
  P}, \beta_{[M:Q], P})$ is a local diffeomorphism. Thus $\maB_{M, Q,
  P}$ is an immersion. As it is injective and proper, it is a
homeomorphism onto its image. Proposition~\ref{prop.smfd.crit} implies
that the image is thus a \wsbmanifold{} and that $\maB_{M, Q, P}$ is
a diffeomorphism onto its image.
\end{proof}

By iterating the above lemma, we obtain the following consequence.

\begin{corollary}\label{cor.map.to.graph}
Let $\maP := (\emptyset,P_1, P_2, \ldots, P_k)$ be a family of closed,
 {\em disjoint} \psbmanifolds{} of a manifold with
 corners $M$.
Then we have canonical diffeomorphisms inducing the
identity on $M \setminus\unionP := M \smallsetminus \bigcup_{j=1}^k P_j$ between the
usual blow-ups and the graph blow-up (Definitions \ref{def.blow-up}
and \ref{def.unres.blowup}):
\begin{equation*}
  [[\ldots [[ M: P_1] : P_2]: \, \ldots \, : P_{k-1} ]: P_k] \, \simeq
  \, [M: \bigcup_{j=1}^k P_j] \, \simeq \, \bl{ M: \maP }\,.
\end{equation*}
\end{corollary}

\begin{proof}
This follows by induction from Lemmas \ref{lemma.disjoint} and
\ref{lemma.factorization1} since $P_j$ identifies naturally with a
\psbmanifold{} of $[[\ldots [[ M: P_1] : P_2]: \, \ldots \, : P_{j-2}
  ]: P_{j-1}]$.
\end{proof}

\subsection{Clean semilattices}\label{subsec.clean.semilattices}
We now investigate the iterated blow-up $[M: (P_i)_{i=1}^k]$ of a
manifold with corners $M$ with respect to a suitable cleanly
intersecting finite {\em totally ordered}
family of closed \psbmanifolds{} of~$M$. (If we arrange the elements of
this family according to the total order, then we obtain a $k$-tuple.)

\begin{definition}
\label{def.clean.family}
Let $\maF$ be a locally finite (unordered) set of closed 
\psbmanifolds{} of~$M$. We shall say that $\maF$ is a \emph{cleanly intersecting family}
if any $X_1, X_2, \ldots, X_j \in \maF$ have a clean intersection
(Definition~\ref{def.clean.intersection}).
\end{definition}

Recall that a \emph{meet semilattice} (or, simply, \emph{semilattice} in
what follows) is a partially ordered set $\maL$ such that, for every
two $x, y \in \maL$, there is a greatest common lower bound $x \cap y
\in \maL$ for $x$ and $y$. We shall consider only semilattices of
subsets of a given set where the order is given by $\subset$ and where
$x \cap y$ is the usual intersection of sets. We can now introduce the
semilattices we are interested in. We let $2^M$ denote the set of
all subsets of~$M$. Thus all our semilattices
will be subsets of~$M$ that are stable for intersection. The order given
by inclusion will not play a role and will thus be ignored.

\begin{definition}\label{def.clean.lattice} 
A semilattice $\maS \subset 2^M$  will be called \emph{clean} if it is a 
cleanly intersecting locally finite set of
closed \psbmanifolds\ of~$M$.
\end{definition}

We stress that this definition requires that all elements of a clean semilattice
are closed \psbmanifolds\ of~$M$. It is easier to deal with semilattices since a 
semilattice of closed \psbmanifolds{}
is clean if, and only if, any two members of the semilattice intersect
cleanly. For simplicity of the notation, we shall consider only
semilattices $\maS \subset 2^M$ (the set of subsets of~$M$) with $\emptyset \in \maS$. 
This changes nothing in our results, by Remark \ref{rem.repetition}, but avoids us 
treating separately the cases $\emptyset \in \maS$ and $\emptyset \notin \maS$ 
in proofs. 

The concept of a clean semilattice introduced here is very closely related
to that of a weakly transversal family considered in \cite[Definition~2.7]{ACN}, 
except that in that paper, the authors considered only
\psbmanifolds\ that were {\em not} contained in the boundary.
Similar concepts were also considered in \cite[Theorem 3.2]{Kottke-Lin} 
and in \cite{VasyReg}. The case considered in \cite[Sec.~2]{VasyReg} was
that when all \psbmanifolds\ with respect to which we blow-up \emph{are 
  contained in the boundary.}

The following result was proved in special cases
in \cite{ACN} and in \cite[Lemma~2.7]{VasyReg} with similar proofs. 
(See also \cite{Kottke-Lin}.) In \cite{MelroseBook}, Melrose proved that a 
the lift of a normal family remains a normal family if we do the blow-up by 
an element of the family. Lemma~5.11.2 of~\cite{MelroseBook} also treats the 
lift of a family under the blow-up.

\begin{proposition}\label{prop.bu.lattice}
Let $\emptyset \in \maS \subset 2^M$ be a clean semilattice. (Thus $\maS$
consists of closed \psbmanifolds{} of~$M$.) Let $P$ be a minimal element of $\maS \smallsetminus
\{\emptyset\}$. Then, for all $Q \in \maS$, $Q' := [Q: P\cap Q]$ is a closed \psbmanifold{} 
of $[M: P]$ and (after we remove the repetitions)
\begin{equation*}
  \maS' \ede \Bigl\{\, Q' = [Q : P\cap Q] \, \mid \ Q \in \maS \Bigr\}
\end{equation*}
is a clean semilattice of\/ $[M:P]$ with $\emptyset = \emptyset' = P' \in \maS'$
and, hence, with fewer elements than $\maS$.
\end{proposition}

\begin{proof}
The first part of tis result was proved in slightly less generality in
\cite[Theorem 2.8]{ACN} (assuming that the \psbmanifolds{} are {\em
not} contained in the boundary). The proof extends right away to
the current setting using Lemma \ref{prop.beta.m1}. See also \cite{Kottke-Lin}.

For the last part of the result, recall that the minimality of $P$ and the 
semilattice property of $\maS$ imply that, for any $Q\in \maS$, we have either 
$P \subset Q $ or $P\cap Q=\emptyset$. In the first case, we have $Q' :=
        [Q:P\cap Q]=[Q:P]$ and in the second case we have $Q' := [Q:
          P\cap Q]{=Q}$. Thus
\begin{equation*}
   \maS' \ede \{\, [Q:P] \, \vert \ P \subset Q \in \maS \} \, \cup
   \, \{ \, Q \, \vert \ Q \in \maS , Q \cap P = \emptyset \} \, .
\end{equation*}
Let us also notice that $P' := [P : P\cap P] = \emptyset =
[\emptyset: \emptyset \cap P] = \emptyset'$. Therefore, $\maS'$
has fewer elements than $\maS$. 
\end{proof}

We are ready now to describe what are the $k$-tuples $\maP$ with respect to which
we will consider the blow-up $[M: \maP]$: they are $k$-tuples coming from a total
order on a set of closed \psbmanifolds{} of~$M$ (denoted by abuse of notation still 
$\maP$) that contains $\emptyset$ and is stable by intersection (thus, a semilattice). 
In addition to this, our semilattice $\maP$ needs to be clean and the total order on 
$\maP$ needs to be ``admissible,'' a concept that we define next.

\begin{definition}\label{def.admissible}
Let $\maS \ni \emptyset$ be a {\em finite,} clean semilattice of closed \psbmanifolds\
of~$M$. We define by induction an {\em admissible order on $\maS$}
to be a total order $\maS = (P_0 = \emptyset, P_1, \ldots , P_n)$ on $\maS$
(thus different, in general, from the inclusion order) such that $P_1$ is minimal for
inclusion in $\maS \smallsetminus \{\emptyset\}$ and the resulting ordering on $\maS'$ 
is also admissible. (Recall that $\maS'$ has fewer elements than $\maS$). If $\maS$ has 
only one element, then its order is admissible. If $\maS$ has two elements, then the 
order is admissible if $\emptyset$ is minimal.
\end{definition}

We have the following important remark.

\begin{remark}\label{rem.size-order}
Proposition \ref{prop.bu.lattice} allows us to define 
the iterated blow-up $[M: \maS]$ with respect to an admissible order
on $\maS$ by induction by
\begin{equation*}
   [M: \maS] \ede [[M: P_1] : \maS']\,,
\end{equation*}
where we used the notation of that proposition.
This was the approach in \cite{ACN} for the definition of $[M: \maS]$.
A very similar method was used by Kottke in \cite{Kottke-Lin}, but
using size orders (recalled in Remark \ref{rem.repetition}). For our 
version of the iterated blow-up, the concept of size order, while very 
close to that of an admissible order, does not work, since $\maS'$ may 
not be size ordered anymore, even if $\maS$ was size ordered. 
Nevertheless, we can
always choose an admissible order on the semilattice $\maS$ 
of \psbmanifolds{} of~$M$ that is also a size order. (Of course, 
in view of Remark \ref{rem.repetition}, we
can assume that the order on any $k$-tuple with respect to which
we blow-up is a size-order, but that is a weaker statement.)
\end{remark}

\subsection{The pair blow-up lemma}
We now perform some essential calculations in local coordinates that
will be needed for our main result. Recall from Equation
\eqref{eq.not.sphere} that
\begin{equation*}
  \SS_k^{n} \ede \SS^{n} \cap \RR_k^{n+1},
\end{equation*}
where $\SS^{n}$ is the unit sphere in $\RR^{n+1}$, as always. For $\psi
\in \SS^{n'+1}_{k'+1} := \SS^{n'+1} \cap \RR^{n'+2}_{k'+1}$, we shall
write $\psi =: (\psi_1, \tilde \psi)$, with $\psi_1 \in [0, 1]$ and
$\tilde \psi \in \RR^{n'+1}_{k'}$, and we define the map
\begin{equation}\label{eq.def.G}
\begin{gathered}
  \Upsilon : \SS^{n-1}_k \times \SS^{n'+1}_{k'+1} \to \SS_{k,
      k'}^{n, n'} \ede  \SS^{n+n'} \cap \big ( \RR^{n}_k \times
  \RR^{n'+1}_{k'} \big )\\
  (\phi,\psi) \mapsto (\psi_1 \phi, \tilde{\psi})\,.
  \end{gathered}
\end{equation}
We embed the sphere orthant $\{0\} \times \SS^{n'}_{k'} = \{0_{\RR^n}\}
\times \SS^{n'}_{k'} \subset \RR^{n+n'+1}$ into $\RR^{n+n'+1}$ by
mapping the sphere orthant to the {\em last} components of
$\RR^{n+n'+1}$. Of course, we have an isomorphism
\begin{equation*}
   \SS_{k, k'}^{n, n'} = \SS^{n+n'} \cap \big ( \RR^{n}_k \times
   \RR^{n'+1}_{k'} \big ) \, \cong\, \SS^{n+n'}_{k + k'} =
   \SS^{n+n'} \cap \RR^{n+n'+1}_{k + k'}
\end{equation*}
given by the canonical permutation of coordinates diffeomorphism of
Equation \eqref{eq.def.can}.

We recall Proposition 5.8.1 of \cite{MelroseBook} and
we give the proof to fix the notation.

\begin{lemma}
\label{blow.sphere}
Let again $\SS_{k, k'}^{n, n'} := \SS^{n+n'} \cap \big ( \RR^{n}_k
  \times \RR^{n'+1}_{k'} \big ) \, \cong\, \SS^{n+n'}_{k+k'}$ and let the
  map $\Upsilon : \SS^{n-1}_k \times \SS^{n'+1}_{k'+1} \to \SS_{k,
    k'}^{n, n'}$ be as in Equation \eqref{eq.def.G}. If we define
\begin{equation*}
    \Psi : \SS_{k, k'}^{n, n'} \smallsetminus \big ( \{0\} \times
    \SS^{n'}_{k'} \big ) \ \to \ \SS^{n-1}_k \times \SS^{n'+1}_{k'+1}
    \,, \quad \Psi(\eta,\mu) \seq \Big (\, \frac{\eta}{|\eta|},\, \big
    (|\eta|,\mu \big ) \, \Big ) \,,
\end{equation*}
then $\Upsilon \circ \Psi$ is the inclusion $\SS_{k, k'}^{n, n'}
  \smallsetminus \big ( \{0\} \times \SS^{n'}_{k'} \big ) \subset
  \SS_{k, k'}^{n, n'}$ and $\Psi$ extends to a diffeomorphism
\begin{equation*}
  \widetilde \Psi : [\SS_{k, k'}^{n, n'} : \{0\} \times
    \SS^{n'}_{k'}]\ \stackrel{\sim}{\longrightarrow}\ \SS^{n-1}_{k}
  \times \SS^{n'+1}_{k'+1}
\end{equation*}
such that $\beta_{\SS_{k, k'}^{n, n'},\, \{0\} \times \SS^{n'}_{k'}} =
\Upsilon \circ \widetilde \Psi$.
\end{lemma}

If we write by abuse of notation $\SS^{n'}_{k'}$ for the image of
$\{0\} \times \SS^{n'}_{k'}$ in $\SS^{n+n'}_{k + k'}$ under the
permutation of coordinates described above, then we obtain a
diffeomorphism
\begin{equation*}
 \widetilde\Psi : [\SS^{n+n'}_{k+k'} :
   \SS^{n'}_{k'}]\ \stackrel{\sim}{\longrightarrow}\ \SS^{n-1}_{k}
 \times \SS^{n'+1}_{k'+1} \,.
\end{equation*}

\begin{proof}
Let
\begin{equation*}
  \beta \ede \beta_{\RR^{n}_{k} \times \RR^{n'+1}_{k'}, \{0\} \times
    \RR^{n'+1}_{k'}} : [\RR^{n}_{k} \times \RR^{n'+1}_{k'} : \{0\}
    \times \RR^{n'+1}_{k'}] \to \RR^{n}_{k} \times \RR^{n'+1}_{k'}\,.
\end{equation*}
denote the blow-down map. Also, recall that the lifting
$\beta^*(\SS_{k, k'}^{n, n'})$ is defined as the closure of
$\beta^{-1}(\SS_{k, k'}^{n, n'} \smallsetminus \{0\} \times
\RR^{n'+1}_{k'})$ in $[\RR^{n}_{k} \times \RR^{n'+1}_{k'} : \{0\}
  \times \RR^{n'+1}_{k'}]$. Since
\begin{equation*}
    \SS_{k, k'}^{n, n'} \cap \bigl(\{0\} \times \RR^{n'+1}_{k'}\bigr)
    \seq \{0\} \times \SS^{n'}_{k'} \,,
\end{equation*}
Proposition~\ref{prop.beta.m1} gives a diffeomorphism
\begin{equation*}
  \Phi : [\SS_{k, k'}^{n, n'} : \{0\} \times \SS^{n'}_{k'}]
  \ \stackrel{\sim}{\longrightarrow}\ \beta^*(\SS_{k, k'}^{n, n'}),
\end{equation*}
uniquely determined by the condition that it is the inclusion on $\SS_{k,
  k'}^{n, n'} \smallsetminus \{0\} \times \SS^{n'}_{k'}$. (That is,
the blow-up of $\SS_{k, k'}^{n, n'}$ along $ \{0\} \times
\SS^{n'}_{k'}$ is diffeomorphic to the lifting $\beta^*(\SS_{k,
  k'}^{n, n'})$ of $\SS_{k, k'}^{n, n'}$ to $[\RR^{n}_{k} \times
  \RR^{n'+1}_{k'} : \{0\} \times \RR^{n'+1}_{k'}]$ via the blow-down
map $\beta \ede \beta_{\RR^{n}_{k} \times \RR^{n'+1}_{k'}, \{0\}
  \times \RR^{n'+1}_{k'}}$.)

To identify more explicitly the space $\beta^*(\SS_{k, k'}^{n, n'})$,
it is convenient to use the diffeomorphism $\kappa :\SS^{n-1}_{k}
\times [0, + \infty) \times \RR^{n'+1}_{k'} \to [\RR^{n}_{k} \times
    \RR^{n'+1}_{k'} : \{ 0 \} \times \RR^{n'+1}_{k'}]$ of Equation
  \eqref{eq.kappa} with the order of its arguments reversed. To start
  with, the blow-down map $\beta := \beta_{\RR^{n}_{k} \times
    \RR^{n'+1}_{k'} , \{ 0 \} \times \RR^{n'+1}_{k'}}$ is such that
  $\beta_1 := \beta \circ \kappa$ satisfies
\begin{equation*}
  \begin{gathered}  
   \beta_1 := \beta \circ \kappa : \SS^{n-1}_{k} \times [0, + \infty)
     \times \RR^{n'+1}_{k'} \, \to \, \RR^{n}_{k} \times
     \RR^{n'+1}_{k'} \,, \\
     \beta_1(z,r, x) \seq (rz, x).
  \end{gathered}
\end{equation*}

We have that $(z,r, x) \in \beta_1^{-1} (\SS_{k, k'}^{n, n'}
\smallsetminus (\{0 \} \times \RR^{n'+1}_{k'}))$ if, and only if
$\|\beta_1(z,r, x)\| = 1$ and $r > 0$. Assume that $\|\beta_1(z,r,
x)\| = 1$ and $r > 0$. Then $ \|rz\|^2 + \|x\|^2 = 1$. Note that $z
\in \SS^{n-1}_k$, and hence $r^2 + \|x\|^2 = 1$. This leads to $(r,
x) \in \SS^{n'+1}_{k'+1} \subset \RR^{n'+2}_{k'+1} = [0,\infty)
  \times \RR^{n'+1}_{k'}$. We thus have
\begin{equation*}
   \beta_1^{-1} \big (\SS_{k, k'}^{n, n'} \smallsetminus (\{0 \} \times
   \RR^{n'+1}_{k'} \big )) \seq (\SS^{n-1}_{k} \times \SS^{n'+1}_{k'+1})
   \smallsetminus (\{0\} \times \RR^{n'+1}_{k'} )\,.
\end{equation*}
The closure of this set is $\SS^{n-1}_{k} \times \SS^{n'+1}_{k'+1}$,
and hence we obtain a diffeomorphism $\Phi_1 := \kappa^{-1} \circ \Phi
: [\SS_{k, k'}^{n, n'} : \{0\} \times \SS^{n'}_{k'}]
\ \stackrel{\sim}{\longrightarrow}\ \SS^{n-1}_{k} \times
\SS^{n'+1}_{k'+1}$. That $\Upsilon \circ \Psi$ is the inclusion
follows from the defining formulas. The relation $\beta_{\SS_{k,
    k'}^{n, n'},\, \{0\} \times \SS^{n'}_{k'}} = \Upsilon \circ
\widetilde \Psi$ follows from the fact that they are both continuous
and they coincide on the dense, open subset $\SS_{k, k'}^{n, n'}
\smallsetminus ( \{0\} \times \SS^{n'}_{k'})$. This shows that $\Phi_1
= \Psi$ on $\SS_{k, k'}^{n, n'} \smallsetminus (\{0\} \times
\SS^{n'}_{k'})$ and hence $\tilde \Psi := \Phi_1$ is the desired
extension.
\end{proof}

We now treat another basic case when the blow-up 
with respect to a clean semilattice consisting of three manifolds is defined, namely the
simplest case when we blow up with respect to $\emptyset$ and two closed \psbmanifolds{} 
$P$ and $Q$ with $\emptyset \neq Q\subset P$. This gives
 the ``second factorization lemma.'' 
(The other basic case is that two disjoint \psbmanifolds,
which was already treated in Lemma~\ref{lemma.factorization1}.)
Again, the existence of the map $\zeta_{M, Q, P}$ was proved
in \cite{Kottke-Lin}, but for a slightly different version of
the iterated blow-up.

\begin{lemma}[The second factorization lemma] 
\label{lemma.factorization2}
Let us assume that $Q$ is a closed \psbmanifold{} of $P$ and that $P$ 
is a closed \psbmanifold{} of~$M$ (and hence $Q$ is also a \psbmanifold{} 
of~$M$). Then $[P:Q]$ is canonically diffeomorphic to
a \psbmanifold{} of $[M: Q]$ and there exists
a unique, smooth, natural map
\begin{equation*}
  \zeta_{M, Q, P} : [M: Q, P ] \simeq 
  \bigl[[M: Q]: [P:Q]\bigr] \ \to \ [M : P]
\end{equation*}  
that restricts to the identity on $M \smallsetminus P$. Moreover, the
product map
\begin{equation*}
  \maB_{M, Q, P} \ede ( \zeta_{M, Q, P}, \beta_{[M:Q], [P:Q]}) : [M:
    Q, P ] \to [M : P] \times [M : Q]
\end{equation*}
is proper in each component. The image of $\maB_{M, Q, P}$ is a weak
submanifold in the sense of Definition~\ref{def.weak-submanifold}, and
$\maB_{M, Q, P}$ is a diffeomorphism onto its image.
\end{lemma}

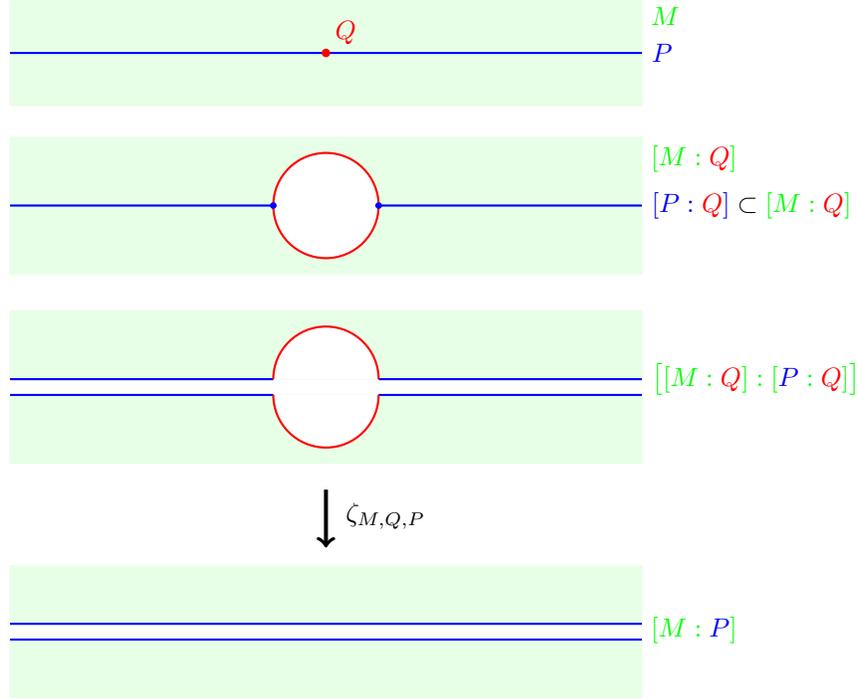
\begin{figure}\label{fig.MPQ}
\begin{center}
    \begin{tikzpicture}[scale=.7]    
       \begin{scope}[shift = {(0,10)}]
            \draw[lightgreen,fill=lightgreen] (-6,1) rectangle (6,-1);
            \draw[-,thick,blue] (-6,0) -- (6,0) node[right]
                 {\blue{$P$}}; \draw[red,fill=red,radius=2pt] (0,0)
                 circle node[above right] {$Q$}; \node[right] (M1) at
                 (6,.7) {$\green{M}$};
      \end{scope}
      \begin{scope}[shift = {(0,7.1)}]
            \draw[lightgreen,fill=lightgreen] (-6,1.3) rectangle
            (6,-1.3); \draw[-,thick,blue] (-6,0) -- (-1,0);
            \draw[-,thick,blue](1,0) -- (6,0) node[right] {$\blue{[P:
                  \red{Q}]}\color{black}\subset \green{[M:
                  \red{Q}]}$}; \draw[red,fill=white,thick,radius=1]
            (0,0) circle node[above right] {};
            \draw[blue,fill=blue,radius=1.5pt] (1,0) circle node[] {};
            \draw[blue,fill=blue,radius=1.5pt] (-1,0) circle node[]
                 {}; \node[right] (M2) at (6,.9) {$\green{[M:
                       \red{Q}]}$};
      \end{scope}
      \begin{scope}[shift = {(0,3.8)}]
            \draw[lightgreen,fill=lightgreen] (-6,1.3) rectangle
            (6,0); \draw[-,thick,blue] (-6,0) -- (-1,0);
            \draw[-,thick,blue](1,0) -- (6,0) node[right]
                 {\green{$\bigl[[M: \red{Q}]: [ \blue{P} : \red{Q}
                       ]\bigr]$}}; \draw[red,fill=white,thick] (0:1)
                 arc (0:180:1); \node[right] (M3) at (6,.9) {};
      \end{scope}
      \begin{scope}[shift = {(0,3.5)}]
            \draw[lightgreen,fill=lightgreen] (-6,-1.3) rectangle (6,0);
            \draw[-,thick,blue] (-6,0) -- (-1,0);
            \draw[-,thick,blue](1,0) -- (6,0) node[right] {};
            \draw[red,fill=white,thick] (-1,0) arc (180:360:1);
      \end{scope}
      \begin{scope}[shift = {(0,.4)}]
            \draw[->,line width=1.5pt,black]  (0,1.3) -- (0,.2); 
            \node[right] (A) at (.2,.8) {$\zeta_{M,Q,P}$};
      \end{scope}
      \begin{scope}[shift = {(0,-1)}]
         \begin{scope}[shift = {(0,.15)}]
            \draw[lightgreen,fill=lightgreen] (-6,1.1) rectangle (6,0);
            \draw[-,thick,blue] (-6,0) -- (6,0) node[] {};
         \end{scope}
         \begin{scope}[shift = {(0,-.15)}]
            \draw[lightgreen,fill=lightgreen] (-6,-1.1) rectangle
            (6,0); \node[right] (Z) at (6.0,.18) {\green{$[M:
                  \blue{P}]$}}; \draw[-,thick,blue] (-6,0) -- (6,0)
            node[] {}; \node[right] (M4) at (6,1.0) {};
         \end{scope}
      \end{scope}
   \end{tikzpicture}
\end{center}
\caption{The blow-ups $[M:Q]$, $\bigl[[M:Q]:[P:Q]\bigr]$, and $[M:P]$}
\end{figure}

See Figure~\ref{fig.MPQ} for a local picture of these blow-ups in the
example $M=\RR^2$, $P=\RR\times \{0\}$, $Q=\{0\}$.

\begin{proof}
The fact that $[P:Q]$ is (diffeomorphic to)
a closed \psbmanifold{} of $[M: Q]$ follows from \ref{prop.beta.m1}.
The uniqueness of the map $\zeta_{M, Q, P}$ follows from the fact that
it is the identity on the dense subset $M \smallsetminus (P \cup
Q)$. The statement is local, so, in view of Lemma~\ref{lemma.product}, we
can assume that $Q = \{0\}$. That is, we can assume that
\begin{equation}\label{local.MPQ}
\begin{cases}
  \ M\! & :=\ \ \RR^{m}_{k_m} \times \RR^{p}_{k_p}\,,\\
  \ P\! & :=\ \ \{0\} \times \RR^{p}_{k_p}\,, \\
  \ Q\! & :=\ \ \{0\}\,.
\end{cases}
\end{equation}
We have
\begin{align*}
  [M:P]\ & =\ [\RR^{m}_{k_m} \times \RR^{p}_{k_p} : \{0\} \times
    \RR^p_{k_p}] \\
  & =\ [\RR^m_{k_m} : \{0\}] \times \RR^{p}_{k_p} \\
  & \simeq\ \SS^{m-1}_{k_m} \times [0,\infty) \times \RR^{p}_{k_p} \\
  &  =\ \SS^{m-1}_{k_m} \times \RR^{p+1}_{k_p
      +1} \,.
\end{align*}
Its blow-down map is $\beta_{M, P}(x, t, y) = (tx, y)$.

On the other hand, we have (using the notation of Lemma~\ref{blow.sphere}):
\begin{equation*}
  [M:Q] \seq [\RR^{m+p}_{k_m+k_p}:\{0\}] \seq \SS_{k_m, \, k_p}^{m,
      p-1}  \times [0,\infty).
\end{equation*}
Its blow-down map is $\beta_{M, Q}(x, t) = tx$. Lemma~\ref{prop.beta.m1} 
gives that the lift of $P$ to $[M:Q]$ is $P' :=
    [P:Q] = \{0\} \times \SS^{p-1}_{k_p} \times [0,\infty)$. Lemmas~\ref{lemma.product} and~\ref{blow.sphere} (in this order) then give
      canonical diffeomorphisms
\begin{align*}
  \bigl[[M:Q]:P'\bigr]\ & \simeq\ [ \SS_{k_m, \, k_p}^{m, p-1} \times
    [0,\infty) : \SS^{p-1}_{k_p} \times [0,\infty) ] \\
  & =\ [ \SS_{k_m, \, k_p}^{m, p-1}  :\SS^{p-1}_{k_p} ] \times
      [0,\infty)\\
  & \simeq\ \SS^{m-1}_{k_m} \times \SS^{p}_{k_p+1} \times
        [0,\infty)\,.
\end{align*}
The blow-down map $\beta_{[ [M:Q]:P] } : \bigl[ [M:Q]:P' \bigr] \to
[M:Q]$ is given, up to canonical diffeomorphisms, by the map $\Upsilon
\times \id$, where $\Upsilon$ is as defined in Equation
\eqref{eq.def.G}. Hence $\Upsilon\times \id(\phi, \psi, t) = (\psi_1
\phi, \tilde \psi, t)$.

The desired map $\zeta_{M, Q, P}$ is then obtained from the blow-down
map $\SS^{p}_{k_p+1} \times [0,\infty) \to \RR^{p+1}_{k_p+1} = [0,
    \infty) \times \RR^{p}_{k_p}$, that is $\zeta_{M, Q, P}(x, y, t) =
    (x, ty)$. In particular, it is proper. It remains to check that
    this map is the identity on $M\setminus P$. Since we used, for $x\in
    M\setminus P$, the identifications
    $x =  \beta_{M,Q}(x) = \beta_{M,P}(x)= \beta_{[M:Q],[P:Q]}(x)$,
    it is enough to check
\begin{equation}\label{tocheck.noups}
  \beta_{M, P} \circ \zeta_{M, Q,P} = \beta_{M, Q} \circ \beta_{[M:Q],[P:Q]}
\end{equation}
on $M\setminus P$. As this calculation is local, we can again assume
\eqref{local.MPQ} and the concrete presentations of $\beta_{M,Q}$,
$\beta_{M,P}(x)$ and $\beta_{[M:Q],[P:Q]}$ described above,
\eqref{tocheck.noups} turns into
\begin{equation}\label{tocheck.lokal.ups}
   \beta_{M, P} \circ \zeta_{M, Q, P} = \beta_{M, Q} \circ (\Upsilon
   \times \id)
\end{equation} 
on $\SS^{m-1}_{k_m} \times \SS^{p}_{k_p+1} \times [0,\infty)$.
Indeed, for $x\in \SS^{m-1}_{k_m}$, $y=(y_1,\tilde y)\in
  \SS^p_{k_p+1}\subset \RR^{p+1}_{k_p+1}=\RR^1_1\times \RR^p_{k_p}$,
  $t \in[0,\infty)=\RR^1_1$, we have
\begin{equation*}
   \beta_{M, P} \circ \zeta_{M, Q, P}(x, y, t) \seq \beta_{M, P} (x,
   ty)
   \seq \beta_{M, P} (x, ty) \seq
   (ty_1 x, t \tilde y)\,.
\end{equation*}
Together with
\begin{equation*}
  \beta_{M, Q} \circ (\Upsilon \times \id) (x, y, t) \seq \beta_{M,
    Q}(y_1 x, \tilde y, t)
  \seq (ty_1 x, t \tilde y)\,,
\end{equation*}
this implies \eqref{tocheck.lokal.ups}.  

The map $\maB$ is given in local coordinates by $\maB(x, y, t) = \big
(x, ty, (y_1x, \tilde y), t)$ with differentiable left inverse $(x, z,
(w_1, w_2), t) \mapsto (x, (\|w_1\|, w_2) , t)$. Hence by
Corollary~\ref{cor.smfd.crit} the image of $\maB$ is a weak
submanifold and $\maB$ is a diffeomorphism onto its image.
\end{proof}

\begin{remark}\label{rem.Bernd}
Note that, in general, the image of the map $\maB_{M,Q,P}$ introduced
in the proof above is not a \psbmanifold{} of $[M:P] \times
[M:Q]$. Indeed, let us consider the case when~$M$ is the closed unit
disk in~$\RR^2$, and let $p$ and $q$ be two disjoint points in the
interior of~$M$. Let $P := \{p\}$ and $Q := \{q\}$. We claim that the
image $N$ of $\maB=\maB_{M,Q,P}$ is {\em not} a \psbmanifold{} of
$M_1:= [M:P] \times [M:Q]$. Suppose $N$ were a \psbmanifold{} of
$M_1$.  As~$N$ is connected, the function $\depth_{M_1}(x) -
\depth_N(x)$ is constant on $N$, see Remark~\ref{rem.depths}. However,
the map $\maB$ sends the interior points of $M \smallsetminus \{p,q\}$
to the interior of $M_1 = [M:P] \times [M:Q]$, thus $\depth_{M_1}(x) -
\depth_N(x) = 0 - 0 = 0$ for $x = \maB(y)$ with $y$ in the interior of
$M \smallsetminus \{p,q\}$. On the other hand $\maB$ maps the boundary
$\partial M = \pa ( M \smallsetminus \{p,q\}) $ to the corner
$\partial M \times \partial M$ of $[M:P] \times [M:Q]$, which has
boundary depth $2$ in $M_1 = [M:P] \times [M:Q]$. Thus, if $x =
\maB(y)$, with $y \in \pa M$, we obtain $\depth_{M_1}(x) - \depth_N(x)
= 2 - 1 = 1$. Therefore, the function $\depth_{M_1}(x) - \depth_N(x)$
is not constant on $N$, and hence $N$ is not a \psbmanifold{} of $M_1
= [M:P] \times [M:Q]$.

A careful investigation \cite{koenig.master} shows that the image  
$\maB_{M,Q,P}([M: Q, P]$ of $[M: Q, P]$ in $[M : P ] \times 
[M : Q]$ may fail to be a submanifold of $[M : P ] \times 
[M : Q]$ in the sense of manifolds with corners, see Definition~\ref{def.submanifold-gen}.
This fact justifies our introduction of the notion of a ``\wsbmanifold.'' 
In particular, a \wsbmanifold{} is neither a 
\bsbmanifold{} nor a \wibsbmanifold{}, see Appendix~\ref{ssec.appendix.submanifold.2}. 
\end{remark}

\subsection{The graph blow-up}
In this section we introduce the \emph{graph blow-up}, which is a version of the blow-up 
with respect to a family
of closed \psbmanifolds{} that obviously does not depend on any order on that family.
For our applications the most important case is the one of a compact manifold with 
corners $M$. In this case locally finiteness implies finiteness. Finiteness of the 
semilattice simplifies the presentation, thus we will assume this from now on; but 
let us mention that there are obvious extensions to locally finite clean semilattices 
in non-compact manifolds with corners.

Let $M$ be a manifold with corners and $\maP = (P_i)_{i \in I}$ be a 
finite family of closed \psbmanifolds{} of~$M$. Let $\delta : M \smallsetminus \unionP 
\to \prod_{i \in I} [M: P_i]$ be the diagonal map $\delta(x) = (x, x, \ldots,
x)$, as before. We write $\unionP \ede \bigcup_{i \in I} P_i$.
Then $M \smallsetminus \unionP$ is  an open subset of $[M: P_i]$, for each
$i \in I$.  Motivated by the results of \cite{GN, MNP}, we now
introduce the following definition.

\begin{definition}\label{def.unres.blowup}
Let $\maP = (P_i)_{i \in I}$ be a  finite family  of closed 
\psbmanifolds{} of the manifold with corners~$M$. 
Then the {\em graph blow-up} $\bl{M: \maP}$ of
$M$ along $\maP$ is defined by
\begin{equation*}
    \bl{M: \maP} \ede \overline{ 
    \delta(M\smallsetminus \unionP) }
    \seq       
      \overline{ \{ (x, x, \ldots, x) \, \vert \ x \in
      M\smallsetminus \unionP \} } \ \subset\ \prod_{i \in I} [M:
      P_i] \,.
\end{equation*}
\end{definition}



\begin{remark}\label{rem.repetition2}
The definition of the family $\maP = (P_i)_{i \in I}$ allows for repetitions.
That is, we may have $P_i = P_j$ for some $i, j \in I$, $i \neq j$.
If we {\em remove the repetitions,} will obtain a graph blow-up that is canonically
homeomorphic to the original one. Similarly, changing the index set $I$ will also yield a 
canonically homeomorphic graph blow-up. In particular, if $I$ is finite
(which is the case for most of this paper), we can introduce a total order
on $\maP$. The graph blow-up corresponding to different orders on the
finite set $\maP$ will be, however, canonically homeomorphic. See also
Remark~\ref{rem.repetition}.
\end{remark}

We now show that $\bl{M: \maP}$ is a \wsbmanifold{}
of a suitable manifold with corners provided that $\maP$ is an
admissible ordered clean semilattice.

\begin{theorem}\label{thm.main1}
Let $\maS \ni \emptyset$ be a finite, clean semilattice of closed
\psbmanifolds{} of~$M$ with an admissible order
(Definition~\ref{def.admissible}), so that $[M: \maS]$ is
well-defined (Remark~\ref{rem.size-order}).  
\begin{enumerate}[\kern3mm\rm (i)]
\item\label{thm.main1.i} For each $P
\in \maS$, there exists a unique smooth map $\phi_{\maS, P} : [M:
  \maS] \to [M: P]$ that is the identity on $M \smallsetminus
\bigcup_{P \in \maS} P$. These maps are such that the induced map
\begin{equation*}
   \maB_\maS \ede \left(\phi_{\maS, P_0}, \ldots, \phi_{\maS,
     P_k}\right) : [M: \maS] \to \prod_{j=0}^k [M: P_j]
\end{equation*}
is an injective immersion    and proper in each component. 

\item\label{thm.main1.ii} The image of $\maB_\maS$ is $\bl{M:\maS}$. Hence, $\bl{M:\maS}$
is a \wsbmanifold{} of the product $\prod_{j=0}^k [M: P_j]$ in the sense of
Definition~\ref{def.weak-submanifold} and $\maB_\maS$ is 
a diffeomorphism
 \begin{equation*}
   \maB_{\maS} : [M: \maS] \ \stackrel{\sim}{\longrightarrow}\ \bl{M:
     \maS}\,.
\end{equation*}  
\end{enumerate}
\end{theorem}

\begin{proof}
We shall prove~\eqref{thm.main1.i} and~\eqref{thm.main1.ii} together by induction on the number $k+1$ of elements of
$\maS$. Recall that the iterated blow-up $[M: \maS]$ is defined since $\maS$ is 
endowed with an admissible order. (See also Remark~\ref{rem.size-order}.)
In the case $k=0$, there is nothing to prove, since $\maS = \{\emptyset\}$
then. 
\smallskip

\noindent \textbf{Case $k=1$:} If $\maS$ has $1+1=2$ elements, we have
$\maS = (\emptyset , P)$ and $\mathcal{B_\maS} = (\beta_{M,
  P},\id_{[M:P]})$ so the claim is trivially satisfied, since the
blow-down map is proper (Corollary~\ref{cor.beta.proper}).
\smallskip

\noindent  \textbf{Case $k=2$:}
If $\maS$ has $2+1=3$ elements, we have $\maS = \{\emptyset , Q, P\}$
and either $Q \cap P = \emptyset$ or $Q \subset P$ (the case $P \subset Q$
would not yield an admissible order on $\maS$).  
\smallskip

\noindent 1)\ In the {\em first} subcase, that is, if $Q\cap P =
\emptyset$, the result was already proved in the first factorization lemma,
Lemma~\ref{lemma.factorization1}, with $\maB_\maS= (\beta_{M,P\cup
  Q},\beta_{[M:Q],P},\beta_{[M:P],Q})$, that is, all the components of
$\maB_\maS$ are given by blow-down maps. The diffeomorphism property
for $\maB_\maS$ comes from the fact that its restriction to
$[M\setminus Q:P]$ and $[M\setminus P:Q]$ has a component equal to the
identity, so it is a local diffeomorphism onto its image, which is at
the same time injective and proper, thus having a continuous inverse.
\smallskip

\noindent 2)\ Similarly, in the {\em second} subcase, that is, 
if $Q \subset P$, the result was already proved in the second factorization 
lemma, Lemma~\ref{lemma.factorization2}, with
\begin{equation*}
  \maB_\maS \ede (\beta_{M,Q}\circ \beta_{[M: Q], [P:Q]} , \beta_{[M:
      Q], [P:Q]},\zeta_{M, Q, P} ) \,,
\end{equation*}
that is, we have, $\phi_{\maS,\emptyset}= \beta_{M,Q}\circ\beta_{[M:
    Q], [P:Q]}$, $\phi_{\maS, Q} := \beta_{[M: Q], [P:Q]}$,
$\phi_{\maS, P} := \zeta_{M, Q, P}$. In particular, the fact that
$\maB_{M,Q,P} = (\beta_{[M: Q], [P:Q]},\zeta_{M, Q, P})$ is a diffeomorphism
onto its image implies the same statement for $\maB_\maS$.
\smallskip

\noindent \textbf{Case $k\geq 3$:} Let us now proceed with the
induction step from $k-1$ to $k$, that is, let us assume that $\maS$
has $k+1$ elements, $P_0=\emptyset, P_1,\ldots, P_k$, arranged in
the given, admissible order. Let $P' := [P: P \cap P_1]$. Thus we have $P'= [P:
  P_1]$, if $P_1 \subset P$, and $P'= P$, if $P_1 \cap P =
\emptyset$. We shall use the notation of Proposition~\ref{prop.bu.lattice} with $P:=P_1$, in particular, $Q' := [Q: Q \cap
  P_1]$. The semilattice $\maS' = \bigl(P_j' := [P_j: P_j\cap
  P_1])_{j=1,\ldots k}$ of Proposition~\ref{prop.bu.lattice} is then
clean and with an admissible order.  
As we have remarked already, $P_1' := [P_1: P_1\cap P_1] = \emptyset =
\emptyset'$, and hence $\maS'$ has at most $k$ elements. By the induction
hypothesis, the map $\maB_{\maS'}$ is a diffeomorphism onto its
image. The same property is shared by the maps
\begin{equation*}
  \maB_{M, P_1, P_j} :\bigl[[M: P_1]: [P_j:P_1]\bigr] \to [M: P_1]
  \times [M: P_j]
\end{equation*}
of the Lemmata~\ref{lemma.factorization1} and~\ref{lemma.factorization2},
since either $P_1 \cap P_j = \emptyset$ or $P_1 \subset P_j$,
since we have assumed that the order on $\maS$ is admissible.
Let $\Phi := {\id \times \prod_{j=2}^k} \maB_{M, P_1, P_j}$ and
consider the composition
\begin{multline}\label{eq.important2}
  [M: \maS] \ede \bigl[[M:P_1] : \maS'\bigr] \ \stackrel{
    \maB_{\maS'}}{\xrightarrow{\hspace*{.54cm}}}
    \prod_{j=1}^k \bigl[[M: P_1]: [P_j:P_1]\bigr] \\
  \stackrel{\Phi }{\longrightarrow} \ [M: P_1] \times
  \prod_{j=2}^k \Bigl([M: P_1] \times [M: P_j]\Bigr)\,. \\
\end{multline}

The two maps of the composition are both injective immersions, and
hence their composition is again an injective immersion. The desired
map $\phi_{\maS, P_j}$ is the projection onto the $P_j$-component.
The projection of the composite map onto any of the factors is the
identity on $M \smallsetminus \left(\bigcup_{j=1}^k P_j\right) $.  Note that all
components with factors of the form $[M: P_1]$ (which are repeated),
yield the same projection, again because this projection is the
identity map on $M \smallsetminus \left(\bigcup_{j=1}^k P_j\right) $. By
removing these repetitions, and by adding the iterated blow-down map
$[M: \maS]\to M$ we obtain the desired map $\maB_{\maS}$, which is
consequently also an injective immersion. The map $\maB_{\maS}$, is
proper in each component, and thus proper. It follows
from Corollary~\ref{cor.prop.homeo} that $\maB_{\maS}$ is a
homeomorphism to its image $N:=\maB_{\maS}([M:\maS])$. With
Proposition~\ref{prop.smfd.crit} we see that $N$ is a \wsbmanifold{}
of $\prod_{j=0}^k [M: P_j]$, and that $\maB_{\maS}$ is a
diffeomorphism onto $N$.

It remains to argue that $N$ coincides with 
\begin{equation*}
  \bl{M:\maS}
  := 
  \overline{\maB_{\maS}\Bigl(M\setminus
    \bigcup_{j=1}^k P_j \Bigr)}\,.
\end{equation*}

For any $x\in [M:\maS]$, there is a sequence $(x_i)$ in $M\setminus
\left(\bigcup_{j=1}^k P_j\right) $ converging to $x$ in $[M:\maS]$. Therefore
\begin{equation*}
  \maB_{\maS}\Bigl(M \smallsetminus (\bigcup_{j=1}^k P_j)  \Bigr)\ni
  \maB_{\maS}(x_i)\to \maB_{\maS}(x)\,,
\end{equation*}
and hence $\maB_{\maS}(x)\in \bl{M:\maS}$. It follows that $N\subset
\bl{M:\maS}$.

Conversely, for $y\in \bl{M:\maS}$, there is a sequence
$y_i=\maB_\maS(x_i)$ in $\maB_{\maS}\Bigl(M \smallsetminus \left(\bigcup_{j=1}^k P_j\right)
 \Bigr)$ converging to $y$ in $\prod_{j=0}^k [M: P_j]$. Thus
$\{y_i\mid i\in \mathbb{N}\}\cup \{y\}$ is compact, and by properness
of $\maB_\maS$ the set
\begin{equation*}
    \bigl(\maB_\maS\bigr)^{-1}\left(\{y_i\mid i\in \mathbb{N}\}\cup
    \{y\}\right)=\{x_i\mid i\in \mathbb{N}\}\cup
    \bigl(\maB_\maS\bigr)^{-1}\left(\{ y\}\right)
\end{equation*}
is compact as well. As a consequence a subsequence $x_{i_k}$ has to
converge to some ${z\in [M:\maS]}$. We conclude that
\begin{equation*}
N\ni \maB_\maS(z)=\lim_{l \to\infty }\maB_\maS(x_{i_l})=
\lim_{l\to\infty}y_{i_l}=y.
\end{equation*} This yields $\bl{M:\maS}\subset N$.
\end{proof}

Again, the image of the map $\maB_{\maS}$ is, in general, not a
\psbmanifold, see Remark~\ref{rem.Bernd}. We obtain as a first application 
of our results the following 
actions of Lie groups.

\begin{definition}\label{def.Gfam}
If $G$ is a Lie group acting smoothly on~$M$ and $\maP$ is a
 finite set of closed \psbmanifolds{} of~$M$ 
such that, for every $P \in \maP$ and $g \in
G$, we have $g(P) =Q$  for some $Q$  $\in \maP$, then 
we shall say that {\em $\maP$ is a $G$-family of closed \psbmanifolds{} of~$M$}.
\end{definition}

Proposition~\ref{prop.group.action} yields right away the following
corollary.



\begin{theorem}\label{thm.cor.main1}
Let $G$ be a Lie group acting smoothly on~$M$ and let $\maP$ be 
a $G$-family of closed \psbmanifolds{} of~$M$ (see Definition~\ref{def.Gfam}). 
Then $G$ acts continuously on $\bl{M: \maP}$.
If, moreover, $\maP$ is a clean semilattice of closed \psbmanifolds{} 
of~$M$, then $G$ commutes with the homeomorphism $\maB_\maP$ of
Theorem~\ref{thm.main1} and it acts smoothly on $\bl{M: \maP} \simeq [M: \maP]$.
\end{theorem}

\begin{proof}
Let $\delta(x) = (x, \ldots, x)$ be the diagonal embedding 
$\delta : M \smallsetminus \cup \maP \to \prod_{P \in \maP} [M: P]$
considered before. We have that each $G$ acts smoothly on
  $M \smallsetminus \cup \maP$ and on $\prod_{P \in \maP} [M: P]$, with
  the action sending $[M: P]$ to $[M: g(P)]$, by Proposition~\ref{prop.group.action}. The action
  by homeomorphisms of $G$ on $\bl{M: \maP}$ then 
  follows since $\delta$ commutes with the action of~$G$. 
  The map $\maB_\maP$ is also clearly $G$-equivariant. The 
smoothness of the action of $G$ on $\bl{M: \maP}$ then follows from Theorem~\ref{thm.main1}
and the smoothness of the action of $G$ on $\prod_{P \in \maP} [M: P]$.
\end{proof}


\section{Identification of the Georgescu-Vasy space} \label{sec6} 

We now apply the results of the previous sections
to identify the spaces introduced by Georgescu and Vasy with the
space $\XGV := \overline{\delta_\maF(X)}$ defined in the Introduction.
In what follows, the role played by $M$ in the previous sections will
be played by the spherical compactification $\overline{Z}$ of a vector
space~$Z$, which we recall next.

\subsection{Spherical compactifications}
\label{ssec.sph.comp}
For any finite-dimensional real vector space $Z$, recall that $\SS_Z$
denotes the set of vector directions in $Z$, that is, the set of
(non-constant) open half-lines $\RR_{+} v$, with $0 \neq v \in Z$ and
$\RR_{+}:=(0,\infty)$. The disjoint union
\begin{equation} \label{eq.oZ}
  \overline{Z} \ede Z \sqcup \SS_Z 
\end{equation}
is then called the {\em radial compactification} of $Z$. For example,
if $Z = \RR$, then $\overline{Z} := [-\infty, \infty]$ with the usual
topology. The action of the group $\GL(Z)$ of linear automorphisms of
$Z$ extends, by definition, to an action on~$\overline{Z}$. Similarly,
if $Y \subset Z$, then $\overline{Y} \subset \overline{Z}$. In
particular, $\overline{Z}$ is the union of all {\em closed lines}
$\overline{\RR v}$, $0 \neq v \in Z$, with closure taken
in~$\overline{Z}$.

As it is well known, $\overline{Z}$ carries a topology and a
smooth structure, and our next goal is to recall their
definitions, which will, in particular, turn $\overline{Z}$ into a smooth manifold with
boundary. For notational purposes it is convenient consider the case
$Z = \RR^n$ first. We start by noticing that there is a bijection between the set of vector
directions in $\RR^{n+1}$ and its unit sphere~$\SS^n$. This allows us
to regard $\SS^n_1 := \{(x_1, x') \in [0,\infty)\times \RR^n \mid x_1^2 + |x'|^2 = 1  \}$ as  
the set of vector directions in~$\RR^{n+1}_1$, where we used the usual
notation of Equation~\eqref{eq.def.Rnk}. Let
\begin{equation*}
  \< x\>^2 \ede 1 + \|x\|^2 \seq \|(1, x)\|^2 \,,
\end{equation*}
as usual. We then have the following simple observation.

\begin{remark}\label{rem.Theta_n}
Let $\Theta_n:\overline{\RR^n} = \RR^n \sqcup
\SS_{\RR^n} \to \SS^n_1$ be given by the formula:
\begin{equation*}
 \Theta_n(x) \ede
 \begin{cases} 
    \ \frac{1}{\langle x \rangle} (1,x) & \ \mbox{ if } x \in \RR^n\,, \\
    \ \frac1{\|v\|}(0,v) & \ \mbox{ if } x = \RR_{+}v\in \SS_{\RR^n}\,.
  \end{cases}
\end{equation*}
First, the map $\Theta_n$ is well defined because $\RR_{+}v=\RR_{+}w$ implies
$v=\lambda w$, for some $\lambda \in \RR_{+}$. Second, $\Theta_n$ is 
$\GL(n,\RR):= \GL(\RR^n)$-invariant for the action defined in the last paragraph. 
Finally, it is bijective and its inverse is given by 
\begin{equation}\label{theta.inv.formula}
    \mathbb{S}^n_1 \ni (y_0,y_1, \ldots, y_n) \stackrel{ \Theta^{-1}_n}{\longmapsto}
     \begin{cases}\frac{1}{y_0}(y_1,\ldots ,y_n) \in \mathbb{R}^n & \text{if } y_0\neq 0\\
       \RR_{+} (y_1,\ldots,y_n) \in \SS_{\RR^n} & \text{if } y_0=0 \,.
     \end{cases}
\end{equation}
We endow $\overline{\RR^n}$ with the structure of a smooth manifold
(with boundary) that makes $\Theta_n$ a
diffeomorphism. This manifold structure on $\overline{\RR^n}$ extends
the standard manifold structure of~$\RR^n$. See also \cite{MelroseBook, VasyReg}.
\end{remark}

We now extend the definition of the smooth structure on $\overline{\RR^n}$
of Remark~\ref{rem.Theta_n} to any $n$-dimensional real vector space $Z$ in the usual way.
First, choose a vector space isomorphism $Z\to \RR^n$, which
yields bijections
\begin{equation}\label{eq.Z.to.Sn1}
  \overline{Z}\ \stackrel{\sim}{\longrightarrow} \ \overline{\RR^n}
  \ \stackrel{\sim}{\longrightarrow} \ \SS^n_1\,.
\end{equation}
In turn, these bijections can be used to define a smooth structure on the radial
compactification~$\overline{Z}$ of~$Z$. The $\GL(n,\RR)$-invariance of~$\Theta_n$
implies that the resulting smooth structure
on~$\overline{Z}$ does not depend on the isomorphism $Z \to \RR^n$.
We note, in passing, that $\overline{Z} \simeq [Z^+: \{\infty\}]$.

It follows from the definition of the radial
compactification and of its topology that, if $Y \subset Z$ is a
(linear) subspace, then $\overline{Y} \subset \overline{Z}$ is a
closed \psbmanifold{} and $\SS_Y = \SS_Z \cap \overline{Y}$.

\subsection{Quotients and compactifications}
If $Y$ is a proper linear subspace of~$X$, then the natural 
projection map $\pi_{X/Y} : X \to X/Y \to \oXY$ extends to a
well-defined map $\oX \setminus  \SS_Y \to \oXY$, which,
at the boundary, is given by $\RR_{+}x \mapsto \RR_{+}(x+Y)$. This
map does not extend to a continuous map on $\overline{X}$, but, as we
will show next, it extends to the blow-up of $\oX$ with respect to $\SS_Y$.

\begin{proposition}\label{prop.X/Y}
The canonical surjection $\pi_{X/Y} : X \to X/Y$ extends to a smooth
map $\psi_Y :[\oX: \SS_Y] \to \overline{X/Y}$ such that the induced
map 
\begin{equation*}
 \theta_{Y}  \ede (\beta_{\oX, \SS_Y}, \psi_Y) : 
 [\oX : \SS_Y] \to \oX \times \oXY
\end{equation*} 
is a diffeomorphism onto its image, which is
a \wsbmanifold{} of the product $\oX \times
\oXY$. Let $G = \GL(X, Y) \subset \GL(X)$ be the group of
automorphisms of~$X$ that map $Y$ to itself. Then $\psi_Y$ is
$G$-equivariant.
\end{proposition}

Again, one can check, that the image $\theta_Y$ is not a submanifold
in the sense of Definition~\ref{def.submanifold-gen};
it is only a \wsbmanifold{} (that is, a submanifold in our weaker sense
of Definition~\ref{def.weak-submanifold}.).

\begin{proof}
In view of the equivariance of $\Theta_n$ and of the bijections in 
\eqref{eq.Z.to.Sn1}, we can assume $X = \RR^n$ and $Y
= \{0 \} \times \RR^q$. We will write $\SS^{q-1}$ and $ \RR^q$
instead of $\{0\} \times \SS^{q-1}$ and $\{0\} \times \RR^q$, for
simplicity. Recall that Lemma~\ref{blow.sphere} yields a
diffeomorphism $\widetilde \Psi : [\SS_{k, k'}^{r, r'} : \{0\} \times
  \SS^{r'}_{k'}]\ \stackrel{\sim}{\longrightarrow}\ \SS^{r-1}_{k}
\times \SS^{r'+1}_{k'+1}$. We shall use this result for $r = n-q+1$,
$r' = q-1$, $k = 1$, and $k' = 0$. Since $\SS^{q-1}_0 = \SS^{q-1}$ and
$\SS_{1, 0}^{n-q+1, q-1}=\SS_1^{n}$, we obtain the diffeomorphism
\begin{equation*}
  \widetilde \Psi : [\SS_1^n :
    \SS^{q-1}]\ \stackrel{\sim}{\longrightarrow}\ \SS^{n-q}_{1} \times
  \SS^{q}_{1}\,.
\end{equation*}
Let $p_1 : \SS^{n-q}_{1} \times \SS^q_{1} \to \SS^{n-q}_{1}$ be the
projection onto the first component.

By definition of the smooth structure on $\overline{X}$, the map
$\Theta_n:\overline{X}\to \SS^n_1= \SS_{1, 0}^{n-q+1, q-1}$ of
Remark~\ref{rem.Theta_n} is a diffeomorphism, and it maps
diffeomorphically $\SS_Y$ onto $\SS^{q-1}$. Then by
Lemma~\ref{lemma.ACN} we obtain a diffeomorphism $\Theta_n^\beta :
[\oX: \SS_Y] \to [ \SS_1^n :\SS^{q-1}]$.

We define $\psi_Y$ as the composition
\begin{equation*}
   [\oX: \SS_Y]\xrightarrow{\Theta_n^\beta} [ \SS_1^n
    :\SS^{q-1}]\xrightarrow{ \widetilde \Psi} \SS^{n-q}_{1} \times
  \SS^q_{1} \xrightarrow{ p_1} \SS^{n-q}_1 \xleftarrow{\Theta_{n-q}}
  \overline{X/Y},
\end{equation*}
in other words  
\begin{equation*}
  \psi_Y \ede \bigl(\Theta_{n-q}\bigr)^{-1} \circ p_1 \circ \widetilde
  \Psi \circ \Theta_n^\beta \, :\, [\oX: \SS_Y] \, \to\,
  \overline{X/Y},
\end{equation*}
and we claim that $\psi_Y$ is the desired extension.

To prove the claim, recall that we defined $\widetilde \Psi$ in
Lemma~\ref{blow.sphere} as the unique continuous extension of the map
\begin{equation*}
    \Psi : \SS_1^n \smallsetminus \SS^{q-1} \ \to \ \SS^{n-q}_1 \times
    \SS^q_1,\quad (\eta,\mu) \ \mapsto\ \Big (\,
    \frac{\eta}{|\eta|},\, \big (|\eta|,\mu \big ) \, \Big ),
\end{equation*}
where $\eta \in\RR^{n-q+1}_1$ and $\mu\in \RR^q$. We write $v\in X=
Y^{\perp} \oplus Y$ as $v=(v_\perp,v_Y)$, that is, $v_Y\in Y$ and
$v_\perp\perp Y$, which means $v_\perp \in Y^{\perp} =
\RR^{n-q}\times\{0\}$. Then, in the case $v_\perp\neq 0$, we have
$\Theta_n (v)= \frac{1}{\<v\>} (1,v)\in \SS_1^n \smallsetminus
\SS^{q-1}$, and in this case we then calculate
\begin{equation*}
  \widetilde \Psi \circ \Theta_n (v) \,=\, \Psi \Big (\,
  \frac{1}{\<v\>} (1,v)\, \Big) \,=\, \Big (\, \frac{1}{
    \<v_{Y^\perp}\>}\, (1,v_{Y\perp}),\ \frac{1}{\<v\>}\, (
  \<v_{Y^\perp}\> , v_Y)\, \Big )\,.
\end{equation*}
By continuity of the extension, this formula even holds for all 
$v\in [\SS_1^n : \SS^{q-1}]$. By formula \eqref{theta.inv.formula} we have
$\bigl(\Theta_{n-q}\bigr)^{-1}(y_0,y_1, \ldots, y_{n-q}) =
\frac{1}{y_0}(y_1,\ldots , y_{n-q})$, if $y_0 >0$.
This formula will be used in the following  straightforward calculation:
\begin{equation*}
  \Theta_{n-q}^{-1} \circ p_1 \circ \widetilde \Psi \circ \Theta_n (v)
  = \Theta^{-1}_{n-q} \Big ( \, \frac{1}{ \<v_{Y^\perp}\> }
  (1,v_{Y\perp}) \, \Big ) = v_{Y^\perp} = \pi_{X/Y}(v)\,.
\end{equation*}
So $\psi_Y$ is indeed the desired extension of $\pi_{X/Y}$.

In the remaining part of the proof, we will use  Proposition~\ref{prop.smfd.crit} 
to show that   $\theta_Y = (\beta_{\oX,\SS_Y}, \psi_Y) : [\oX : \SS_Y] 
\to \oX \times \overline{X/Y}$ is a diffeomorphism on its image, and that 
the image of this map is a \wsbmanifold{} of the product $\oX \times \oXY$. 
One of the conditions required by this proposition is that $\theta_Y$  
be  an injective immersion, which we will check now.

The restriction of the map $\beta_{\oX, \SS_Y} : [\oX: \SS_Y] \to
\oX$ to $\oX \smallsetminus
\SS_Y$ is a diffeomorphism onto its image, by the definition of the blow-up, 
and thus $\theta_Y|_{\oX \smallsetminus\SS_Y}$
is an injective immersion as well.
The complement of $\oX \smallsetminus\SS_Y$ in  $[\oX: \SS_Y]$
is $\beta_{\oX, \SS_Y}^{-1}(\SS_Y) := \SS N^{\oX}_+ \SS_Y
 \simeq \SS_Y \times \overline{X/Y}$. On this set the map $\theta_Y$
 becomes the inclusion map 
\begin{equation*}
   \SS_Y \times \overline{X/Y} \to \oX \times \overline{X/Y}\,, 
\end{equation*}
and obviously this is smooth as well.
As $\theta_Y$ maps $\oX \smallsetminus\SS_Y$ and
$\SS_Y \times \overline{X/Y}$ to disjoint sets,
the injectivity of $\theta_Y$ follows. Similarly,
the differential of the map $\theta_Y$
is also injective at the boundary points. Thus, $\theta_Y$ is an injective
immersion.
Furthermore, $\theta_Y$ is defined on a compact set,
and thus it is a homeomorphism onto
its image. Using Proposition~\ref{prop.smfd.crit}, we see that 
its image, $\theta_Y([\oX : \SS_Y])$, is a
\wsbmanifold{} of $\oX \times \oXY$, and the diffeomorphism 
property follows as well. This completes the proof.
\end{proof}

\begin{remark}
Let $\psi := p_1 \circ \widetilde \Psi$, using the notation of the
proof of Proposition~\ref{prop.X/Y}. We thus have a commutative diagram
\begin{equation}  \xymatrix{
  [\oX: \SS_Y] \ar[r]^{\ \ \ \psi_Y} \ar[d]_{\Theta_n^\beta} &
  \overline{X/Y} \\ [ \SS_1^{n-1} : \SS^{q-1}]
  \ \ar[r]^{\ \ \ \ \psi} & \ \SS^{n-q}_1 \ar[u]_{\Theta_{n-q}^{-1}} }
\end{equation}
\end{remark}

The map $\psi_Y$ was also considered in \cite{GN,
  Kottke-Lin}.

\subsection{Georgescu's constructions using $C^*$-algebras}
\label{ssec.sp}
As mentioned before, one of the main motivations of our work is to 
prove that Georgescu's and Vasy's compactifications of $\RR^{3N}$
described in the introduction  are canonically homeomorphic.
Recall from the Introduction that
Georgescu's construction is that of a spectrum of a commutative
$C^*$-algebra \cite{Georgescu2018,GeIf06, GN}, whereas Vasy used blow-ups
\cite{VasyReg, VasySurv}. Georgescu's construction provides a
topological space, whereas Vasy's construction defines a smooth
manifold with corners. Thus a homeomorphism of these spaces that extends 
the identity of $\RR^{3N}$ is the best that we can hope for. In turn, however,
this homeomorphism will then equip Georgescu's compactification with
the structure of a smooth manifold with corners. To compare the approaches
of these two authors, we need to recall a few facts about commutative
$C^*$-algebras. We refer to \cite{Dixmier,Pedersen} for more basic facts about 
$C^\ast$-algebras.

\begin{definition}\label{def.C*algebre}
A \emph{$C^\ast$-algebra} $A$ is an algebra over $\mathbb{C}$ with a
norm $\|.\|$ and with a map $\ast : A \to A$ such that $A$ is a Banach
algebra and for every $\lambda,\mu \in \mathbb{C}$ and $a,b\in A$, we
have
\begin{enumerate}
\item[(i)] $(a^\ast)^\ast=a$,
\item[(ii)] $(ab)^\ast=b^\ast a^\ast$,
\item[(iii)] $(\lambda a +\mu b)^\ast=\overline{\lambda} a^\ast
  +\overline{\mu}b^\ast$,
\item[(iv)] $\|aa^\ast\|=\|a\|^2$.
\end{enumerate}
The $C^\ast$-algebra is \emph{commutative} if $ab = ba$ for all $a, b
\in A$. 
\end{definition}

Every closed, self-adjoint subalgebra of bounded operators on a Hilbert space
is a $C^*$-algebra. In fact, this is a general example, as a basic result 
is that every $C^*$-algebra is isometrically isomorphic to a 
closed, self-adjoint subalgebra of bounded operators on a Hilbert space.
We shall mostly be interested in the following commutative
$C^*$-algebra.

\emph{Let us fix from now on a finite semilattice $\maF$ of linear subspaces 
of some finite dimensional, real vector space  $X$.} It will be convenient 
to assume that $X \notin \maF$, but that $\{0\} \in \maF$. Our approach
works also without these assumptions, but they do simplify the presentation.

\begin{example}\label{ex.Georgescu}
Give two vector spaces $X$ and $Y\subset X$, the composition 
\begin{equation*}
  X\xrightarrow{\pi_{X/Y}}X/Y\xrightarrow{\mathrm{incl}} \overline{X/Y}
\end{equation*}
induces by pullback an injective map
$\maC(\overline{X/Y})\xrightarrow{\pi_{X/Y}^*} \maC_b(X)$, 
where $\maC_b(X)$ is the $C^*$-algebra
of continuous and bounded complex-valued functions on $X$, 
again equipped with the supremum norm. Recall that $\maF$ be a finite semilattice of linear
subspaces of~$X$, $X \notin \maF$, $\{0\} \in \maF$. 
As in the Introduction, Equation \eqref{eq.def.maES}, let $\mEF X$ be the norm
closed subalgebra of $\maC_b(X)$ generated by the pullbacks of the spaces
$\maC(\overline{X/Y})$, where $Y$ runs over $\maF$. Then $\mEF X$ is
stable for complex conjugation, and hence it is a $C^*$-algebra that contains
$\mathcal{C}(\overline{X})$ because $\{0\} \in \mathcal{F}$.
\end{example}

Another general result is that all commutative $C^*$-algebra
are isometrically isomorphic to one of the form $\maC_0(Z)$, discussed in
the following example.

\begin{example}\label{ex.commutative}
For a locally compact and Hausdorff topological space $Z$, let 
$\maC_0(Z)$ be the algebra of complex-valued continuous function $f$ on $Z$
that vanish at infinity (in the sense that the set $|f(z)| \ge \epsilon > 0$
is compact for all $\epsilon > 0$). We endow $\maC_0(Z)$
with the involution $f^*=\overline{f}$ (the complex conjugation) and
with the norm $\|f\|_{\infty}= \sup\limits_{z \in Z} |f(z)|$. With this
structure, $\maC_0(Z)$ is a commutative $C^*$-algebra. It is unital if,
and only if, $Z$ is compact.
\end{example}

The space $Z$ in the last example can be recovered (up to a homeomorphism) 
from the algebra $\maC_0(Z)$ via \emph{characters}, so we now recall this 
concept.
A \emph{character} of a $C^*$-algebra~$\maA$ is a non-zero $*$-morphism
$\chi : \maA \to \CC$. A basic result is that such a character is
continuous of norm $1$. If~$\maA$ is commutative, we denote by
$\Spec(\maA) \subset \maA^*$ the set of characters of~$\maA$ and endow it with
the topology induced from the weak topology on $\maA^*$. It is a locally compact
space and $\maA \simeq \maC_0(\Spec(\maA))$.
If $\maA = \maC_0(Z)$, then~$Z$ and $\Spec(\maA)$ are homeomorphicm via 
evaluations: $Z \ni z \to \ev_z \in \Spec(\maA) \subset A^*$, $\ev_z(f) := f(z)$.
Although we shall not use this in this paper, let us mention that, in case~$\maB$ 
is non-commutative, it may have very few characters (maybe none!),
so the concept of $\Spec(\maB)$ is not very useful in this case. It is rather
the concept of a primitive ideal spectrum that is useful. Recall that an ideal of 
$\maB$ is \emph{primitive} if it is the kernel of a non-zero irreducible representation 
of~$\maB$. If~$\maA$ is commutative, then the primitive ideals of~$\maA$ are exactly 
the maximal ideal of $\maA$ and there is a one-to-one correspondence between the set of 
characters $\chi : \maA \to \CC$ of $A$ and the set of maximal ideals of $A$. This 
correspondence is given by $\chi \mapsto \ker(\chi)$.

In view of the discussion and the examples just introduced, we can now 
recall Georgescu's definition of the compactification of~$X$.

\begin{definition}\label{def.georgescus.compactif}
Let $\maF$ be a finite semilattice of linear subspaces of the finite
dimensional vector space $X$ with $\{0\} \in \mathcal{F}$ and $X \notin \maF$, as 
(as agreed in the lines before Example~\ref{ex.Georgescu}).
Then the spectrum $\Spec(\mEF X)$ of the algebra introduced
in Equation~\eqref{eq.def.maES} is called \emph{Georgescu's compactification of~$X$ with respect to~$\maF$}.
\end{definition}

This definition makes sense since $\maC_0(X) \subset \mEF X$, by definition
(since we have assumed that $\{0\} \in \maF$, and this assumption was introduced
exactly for this reason.) In \cite[Theorem 4.4]{MNP}, two of the authors of this paper 
(together with  Prudhon) have proved the following result. Recall the
space $\XGV := \overline{\delta_\maF(X)}$ defined in Equation \eqref{eq.def.XGV}.

\begin{proposition}\label{prop.Mougel}
The spectrum $\Spec(\mEF X)$ of $\mEF X$ is homeomorphic to the closure 
$\XGV := \overline{\delta_\maF(X)}$ of the image of~$X$ in the product 
$\prod_{Y \in \maF} \overline{X/Y}$ under the ``diagonal'' map 
$\delta_\maF : X \to \prod_{Y \in \maF} \overline{X/Y}$, $\delta_\maF(x) := (\pi_Y(x))_{Y \in
  \maF}$. More precisely, the homeomorphism $\Phi_\maF :
\overline{\delta_\maF(X)} \to \Spec(\mEF X)$ is given as follows. Let
$z=(z_Y)_{Y\in \maF }$ be in the closure of $\delta_\maF(X)$. Then the
homeomorphism $\Phi_\maF$ sends $z$ to the character $\chi_z$
defined by $\chi_z(f_Y) = f_Y(z_Y)$ whenever $f_Y \in \maC(\overline{X/Y})$.
\end{proposition}

Proposition~\ref{prop.Mougel} thus identifies the spectrum $\Spec(\mEF X)$ of 
the $C^*$-algebra $\mEF X$ (Georgescu's space) with the space $\XGV := 
\overline{\delta_\maF(X)}$ introduced in \cite{MNP} and recalled in 
Equation \eqref{eq.def.XGV}.

\subsection{Identification of the Georgescu and Vasy spaces}
Recall that beginning with Example~\ref{ex.Georgescu}, we have assumed that
$\maF$ denotes a finite semilattice of linear subspaces of~$X$ with $\{0\} \in \maF$,
$X \notin \maF$. Let $\SS_{\maF} := \{ \SS_Y = \SS_X \cap \overline{Y} \mid Y   \in \maF \}$ 
be the semilattice introduced in Equation~\eqref{eq.def.SSmaF}.
Then $\emptyset \in \SS_{\maF}$, as it corresponds to the subspace $\{0\}
\subset X$ that was assumed to be in $\maF$. Moreover, $\SS_{\maF}$ is a 
clean semilattice and we endow it with an admissible order.
Recall also the space $\XGV := \overline{\delta_{\maF}(X)}$ introduced
in Equation~\eqref{eq.def.XGV} in the Introduction and the fact that
the graph blow-up $\bl{\oX: \SS_{\maF}}$ has a natural structure
of manifold with corners, by Theorem~\ref{thm.main1}.

\begin{proposition} \label{prop.graph.Mougel}
The product map 
\begin{equation*}
   \Psi_{\maF} := \prod_{Y \in \maF} \psi_Y : \prod_{Y \in \maF} [\oX: \SS_{Y}] \to 
   \prod_{Y \in \maF} \oXY
\end{equation*}
of the maps $\psi_Y$ of Proposition~\ref{prop.X/Y} induces
a diffeomorphism of the graph blow-up 
$\bl{\oX: \SS_{\maF}} \subset \prod_{Y \in \maF} [\oX: \SS_{Y}]$ 
onto its image. Moreover, its image is  $\overline{\delta_{\maF}(X)} =: \XGV$, 
so the latter is a \wsbmanifold{} of the product, and hence a manifold
with corners on its own.
\end{proposition}

\begin{proof} 
Let 
\begin{equation*}
   \Theta_\maF :=  
   \prod_{Y \in \maF} \theta_Y : \prod_{Y \in \maF} [\oX: \SS_{Y}] \to 
   \prod_{Y \in \maF} (\oX \times \oXY)\,,
\end{equation*}
be the product of the maps $\theta_Y$ of Proposition~\ref{prop.X/Y}.
By that proposition, the map $\Theta_\maF$ is a product of injective immersions, 
and hence it is an injective imersion itself. 
Hence $\Theta_\maF$ maps $\bl{\oX: \SS_{\maF}}$
diffeomorphically onto its image. The restriction of the map $\Theta_\maF$ to $\bl{\oX: \SS_{\maF}}$
repeates the component corresponding to $\oX$, and hence,
it is obtained from $\Psi_{\maF} := \prod_{Y \in \maF} \psi_Y$ by repeating these components.
(Note that if $Y = \{0\}$, then $\psi_Y = \psi_{\{0\}} = \id : \oX \to \oX = \oXY$,
so all the $\oX$ component can be obtained from $\psi_{\{0\}}$, and hence
from $\Psi_{\maF} : \prod_{Y \in \maF} \psi_Y$ as well.
(This is yet another reason why we assume $\{0\} \in \maF$.)
It follows that $\Psi_{\maF}$ also maps $\bl{\oX: \SS_{\maF}}$
diffeomorphically onto its image in $\prod_{Y \in \maF} \overline{X/Y}$.
The result follows since $X$ is dense in both $\bl{\oX: \SS_{\maF}}$ and in 
$\XGV := \overline{\delta_{\maF}(X)}$.
\end{proof}

We now obtain the desired diffeomorphism between Vasy's space $[\oX : \SS_{\maF} ]$
with the Georgescu-Vasy space $\XGV := \overline{\delta_{\maF}(X)}$ introduced 
in \cite{MNP} (see Equation \eqref{eq.def.XGV}).

\begin{proposition}\label{prop.cor.Mougel}
The product map 
\begin{equation}
    \Xi_{\maF} : [\oX : \SS_{\maF} ] \ede \prod_{Y \in \maF} 
    \psi_Y \circ \phi_{\SS_{\maF}, \SS_Y} 
    \ \to \ \prod_{Y \in \maF} \overline{X/Y}
\end{equation}
of the composite maps 
\begin{equation*}
     [\oX: \SS_{\maF}]
  \stackrel{ \phi_{\SS_{\maF}, \SS_Y} }{\ - \! - \! \! \! \longrightarrow\ } 
  [\oX: \SS_Y] \stackrel{\psi_Y}\longrightarrow \overline{X/Y}
\end{equation*} 
is a diffeomorphism onto $\XGV := \overline{\delta_{\maF}(X)}.$
Let $G$ be a Lie group of linear
automorphisms of~$X$ that map elements of $\maF$ to elements of
$\maF$ (thus $g(\SS_{\maF}) = \SS_{\maF}$ for all $g \in G$). Then $G$ acts smoothly on
$[\oX: \SS_{\maF}]$, the map $\Xi_{\maF}$ is $G$-equivariant,
and $G$ acts smoothly on $\XGV$.
\end{proposition}

For $Y = \{0\}$, we have $\SS_Y = \emptyset$ (yet another reason for
requiring $\{0\} \in \maF$ and $\emptyset$ to belong to our semilattices),
and hence the map $\psi_Y \circ \phi_{\SS_{\maF}, \SS_Y} = \psi_Y
\circ \phi_{\SS_{\maF}, \emptyset}$ is simply the blow-down map $[\oX :
  \SS_{\maF}] \to \oX$.

\begin{proof}
This follows by combining Proposition~\ref{prop.graph.Mougel} 
with Theorem~\ref{thm.main1} applied to the 
semilattice $\SS_{\maF}$ of closed \psbmanifolds\ of $\oX$ (so $\maS$ of that
theorem is replaced by $\SS_{\maF}$). More precisely, in that theorem, 
the pairs $(\maS, P_j)$  are replaced with the pairs
$(\SS_{\maF}, \SS_Y)$, $Y \in \maF$, and the map $\maB_\maS$
is replaced with the map $\maB_{\SS_{\maF}} := 
\prod_{Y \in \maF} \phi_{\SS_{\maF}, \SS_Y}$. 
Thus Theorem~\ref{thm.main1} gives a diffeomorphism
$\maB_{\SS_{\maF}} : [\oX : \SS_{\maF}] \to \bl{\oX : \SS_{\maF}}$
(including a manifold with corners structure on the latter).
The fact that $\Xi_{\maF}$ is a diffeomorphism with
the stated properties follows from the diffeomorphism $\Psi_{\maF} : 
\bl{\oX : \SS_{\maF}} \to \XGV$ of Proposition~\ref{prop.graph.Mougel}
and the fact that $\Xi_{\maF} = \Psi_{\maF} \circ \maB_{\SS_{\maF}}$.

Finally, the action of $G$ and the fact that
$\Xi_{\maS}$ is $G$-equivariant follow from the fact that all
the maps used to definite $\Xi_{\maS}$ are $G$-equivariant and 
from Proposition~\ref{prop.group.action}.
\end{proof}

Combining Propositions~\ref{prop.Mougel} and~\ref{prop.cor.Mougel}, 
we obtain the following result.

\begin{theorem}\label{thm.main2}
Let $\maF$ be a finite semilattice of linear subspaces of~$X$ containing
$\{0\}$ and $\SS_{\maF} := \{\SS_Y \mid Y \in \maF \}$ be as in 
Equation \eqref{eq.def.SSmaF}.
There exists a unique homeomorphism
\begin{equation*}
   \Spec(\mEF X) \ \simeq\ [\oX : \SS_{\maF} ] 
\end{equation*}
that is the identity on $X$.
\end{theorem}

\begin{proof} 
Let $\delta_\maF : X \to \prod_{Y \in \maF} \overline{X/Y}$ be the diagonal
map. Proposition~\ref{prop.Mougel} states that we have a homeomorphism
$\Spec(\mEF X) \to \XGV := \overline{\delta_\maF(X)}$. 
The result follows from Proposition~\ref{prop.cor.Mougel}, which
states that the map $\Xi_{\maF}$ defined on $[\oX: \SS_{\maF}]$ is a
diffeomorphism onto $\XGV$.
\end{proof}

To conclude, the above results show that the following spaces:
\begin{itemize}
  \item the iterated blow-up $[\oX : \SS_{\maF}]$ (Vasy's space),
  \item the graph blow-up $\bl{\oX: \SS_{\maF}}$, 
  \item $\XGV := \overline{\delta_\maF(X)}$ of Equation \eqref{eq.def.XGV}, and
  \item $\Spec(\mEF X)$ (Georgescu's space)
\end{itemize}
are all homeomorphic.
This yields the sequence of homeomorphisms \eqref{eq.plan.proof} of the Introduction.
More precisely, we can complete that equation with the explicit diffeomorphisms
proved (in order) in Theorem~\ref{thm.main1} (for $\maB := \maB_{\SS_\maF}$), 
Proposition~\ref{prop.graph.Mougel} (for $\Psi := \Psi_{\maF}$), and,
finally, Proposition~\ref{prop.Mougel} for the last morphism $\Phi := \Phi_\maF$
\begin{equation}
  [\oX : \SS_{\maF}] \stackrel{ \maB }{\ \longrightarrow \ } \bl{\oX: \SS_{\maF}}
   \stackrel{ \Psi }{\ \longrightarrow \ } \XGV \ede \overline{\delta_\maF(X)}  
   \stackrel{ \Phi }{\ \longrightarrow \ } \Spec(\mEF X) \,. 
\end{equation}
Any of these spaces will be denoted $\XGV$ from now on and called the 
{\em Georgescu-Vasy space}.
We obtain as a corollary the following description for the 
space introduced in \cite{GeIf06, GN} (the ``small Georgescu space'').

\begin{remark}\label{rem.Georgescu}
In \cite {GeIf06, GN}, Georgescu and his collaborators have
considered the norm closed subalgebra of functions $\mfkA_\maF$ of
$L^\infty(X)$ generated by all the algebras $\maC_0(X/Y)$
with $Y \in \maF$. This
corresponds to potentials that have zero limit at infinity on
$X/Y$. The spectrum of this algebra (after adjoining a unit)
identifies with the closure of the image of the diagonal map of
$X$ to $\prod_{Y \in \maS} (X/Y)^+$, where $Z^+$ denotes the one point 
compactification of a locally compact space $Z$. 
(This is a result analogous to Proposition~\ref{prop.Mougel},
likewise proved in \cite{MNP}.)  Since $\mfkA_\maF \subset \mEF X$, we
obtain that $\Spec(\mfkA_\maF)$ is a quotient of $\Spec(\mEF X)$, and
hence also a quotient of $$\XGV := [\oX: \SS_{\maF}],$$ by Theorem~\ref{thm.main2}. Generally, the topology on $\Spec\left(\mfkA_\maF\right)$ is
rather complicated and singular, see also \cite[Section 5]{Mageira2}
for concrete examples when $\dim(X) = 2$.
\end{remark}

\section{Applications to the $N$-body problem}\label{sec7}

Our main motivation is that, by identifying the spaces appearing in Georgescu's and
Vasy's constructions (Theorem~\ref{thm.main2}), one will be able to combine the results 
and the techniques in their papers and in other related papers to obtain new results.
(Among the papers that we have in mind are Georgescu's works
\cite{GeorgescuBookNew, DaGe04, Georgescu2018,GeIf06, GN}
and Vasy's papers \cite{VasyReg, VasySurv} as well as in \cite{DerGer2, Kottke-Lin, 
KottkeMelrose, MPR1, MantoiuJOT}, and in the references therein). 
In this spirit, in this section, we  discuss some applications of our
results. We provide a brief, but complete account of these applications based on
a complete set of references.

\subsection{The $N$-body semilattice and Pauli exclusion principle}\label{subsec.N-body}
The setting considered in the previous sections of a semilattice $\maF$ 
of linear subspaces of a vector space $X$ is inspired from the $N$-body
problem. In this subsection, we explain the concrete choice
of $\maF = \maF_N$ and of~$X$
in the case of the $N$-body problem and notice that it is compatible
with symmetry and antisymmetry assumptions as, for instance, the Pauli exclusion 
principle. More precisely, the Georgescu-Vasy space associated
to the semilattice of the effective $N$-body problem
carries a natural, concrete action of the symmetric group $S_N$ 
(the permutation group on $N$ letters). This subsection, while relevant
on its own, also sets the stage for the applications in the following sections.

\subsubsection{The semilattice of the $N$-body problem}\label{subsubsec.semilat}
Here is what the choices of~$X$ and $\maF$ are for the 
effective Hamiltonian $H_{N-1}^{\eff}$ of the $N$-body
problem.

\begin{example}\label{ex.Nbody}
In the concrete case of the Hamiltonian $H_N'$ of Equation
\eqref{eq.BO}, we take $X := \RR^{3N}$ and consider the subspaces
\begin{equation*}
  \begin{gathered}
  Y_j \ede \{ x = (x_1, x_2, \ldots, x_N) \in \RR^{3N} \ \vert \ x_j =
  0 \} \quad \mbox{and} \\
  Y_{ij} \ede \{ x = (x_1, x_2, \ldots, x_N) \in \RR^{3N} \ \vert
  \ x_j = x_j \}\,, \ \ i \neq j \,.
  \end{gathered}
\end{equation*}
Thus each $x_i \in \RR^3$. We let $\maF := \maF_N$, be the semilattice 
generated by the subspaces $Y_i$ and $Y_{ij}$, $i, j \in \{1, 2, \ldots, N\}$ 
\cite{GeorgescuBookNew, DerGer2}. (This example is related to the Born-Oppenheimer
approximation for a system with a single nucleus \cite{JeckoBO}.)
The case of $H_N$ of Equation \eqref{eq.def.HN} 
is very similar: we only consider the subspaces $Y_{ij}$. We note, however, that, in this
case $\{0\}$ will not be in the semilattice generated, but the minimal element
is the subspace $\{(x, x, \ldots, x) \mid x \in \RR^{3} \}$. That is not
a real problem, however, since the condition $\{0\} \in \maF$ is imposed
only for convenience. Besides, one can always increase $\maF$ by including
also the zero subspace. This problem does not arise in the case of $H_{N}'$ or
  $H_{N-1}^{\eff}$. For $H_{N-1}^{\eff}$, the semilattice would be more difficult 
  to describe. See the last chapter of \cite{GeorgescuBookNew} for a complete 
  treatement of this class of examples.
\end{example}

In particular, our results give the following.

\begin{remark}\label{rem.Nbody} 
Let us consider the case of the Hamiltonian $H_N'$ of Equation \eqref{eq.BO}, 
the case of the usual $N$-body Hamiltonian $H_N$ being completely similar.
Let the vector space be $X := \RR^{3N}$ and let the semilattice 
$\maF:=\maF_N$ be as 
in the last example,  Example~\ref{ex.Nbody}. Let $\SS_{\maF_N} := \{\SS_Y \mid Y \in \maF_N \}$ 
be the finite semilattice of closed \psbmanifolds\ of
$\oX$ as in Equation~\eqref{eq.def.SSmaF}. 
Then our results, especially Theorem~\ref{thm.cor.main1} imply
that $M_N := [\oX: \SS_{\maF_{N}}] = \XGV$, the Georgescu-Vasy space associated
to the semilattice~$\maF_{N}$, will be endowed with natural, smooth
actions of the following groups:
\begin{itemize}
  \item $S_{N}$, the symmetric group
    on $\{1,2,\ldots,N\}$, acting on the variables  by permutation;
  \item $\GL(3, \RR)$ acting diagonally on
    the components of $X := \RR^{3N}$; and
  \item $X$, extending the action by translation on itself.
  (This is valid for all semilattices $\maF$ of linear
  subspaces of~$X$, not just for $\maF_{N}$, yielding a smooth action of~$X$ on $\XGV$).
\end{itemize}

(These actions can also be obtained from Theorem~\ref{thm.main2} and 
Proposition~\ref{prop.cor.Mougel}.) These actions are easy to obtain at the level of spectra of
$C^*$-algebras or for the graph-family blow-up, but more difficult to
obtain geometrically using iterated blow-ups. In particular, in \cite{MelroseSinger}, 
it was formulated the problem of constructing a compactification of $\RR^{3N}$
endowed with the action of the symmetric group as above. Answers to this problem 
were provided in \cite{Kottke-Lin, MougelPrudhon}. The smoothness of these actions
is based on Theorem~\ref{thm.main2}, since the groups act smoothly on each $\oXY$,
$Y \in \maF$.
\end{remark}

See also  \cite{BaerCMP15, BaerGinoux, Dappiaggi1, Chrusciel, Dappiaggi2, 
GGH15, GerardSurvey, GerardStoskopf, Schrohe, KRS} 
for physically relevant results that can point out to further extensions of
our work, including to Quantum Field Theory on a curved space-time.

\subsubsection{Symmetry, antisymmetry and the Pauli exclusion principle}\label{pauli}
As already remarked above, the action of the symmetric group $S_N$ on 
$M_N := [\oX : \SS_{\maF_N}]$ is important for applications. Recall that, in the motivational 
part of the Introduction, we allowed mixed systems of 
particles. Some of them will be bosons, in which case the wave function will be symmetric under permuation 
of two variables corresponding to bosons of the same kind. Other particles will be fermions (for instance,
electrons), in which case the wave function is antisymmetric under permutations of two variables corresponding 
to fermions of the same kind. This is commonly known as the \emph{Pauli exclusion principle}. In total, 
we consider the subgroup $\Gamma \subset S_N$ of permutations of particles of the same kind, and we obtain a 
map $\chi: \Gamma \to \{-1,+1\}$ such that only functions with 
  $$f:\RR^{3N}\to \CC,\quad f(gx)=\chi(g) f(x)$$
are allowed as wave functions for physical reason, where $\Gamma $ acts on $\RR^{3N}$ by 
permutation of components. 
It is thus helpful that we have proven, see Theorem~\ref{thm.cor.main1}, that this 
$\Gamma $-actions extends to the compactification $M_N$.

In fact the situation becomes slightly more complicated if some particles will have spin, 
which implies -- in mathematical terms -- that they are vector valued. As an example, which 
hopefully is representative of 
the general case, let us explain the case of $N$ electrons. As electrons have spin $1/2$ , we should 
enlarge the target of the wave function and discuss functions
\begin{equation*}
   \Psi:\RR^{3N}\to (\CC^2)^{\otimes N} \,.
\end{equation*}  
Here the tensor product is the tensor product over $\CC$ and the $k$-th factor models the spin of the 
$k$-th electron. Let $S_N$ act on the target $(\CC^2)^{\otimes N}$ by permutation of the factors, that is,
$g(v_1\otimes \cdots\otimes v_N)= (v_{g(1)}\otimes \cdots\otimes v_{g(N)})$
and on $\RR^{3N}$ be exchanging the component vectors, more precisely,
$g(x_1,\ldots,x_N)$ to $(x_{g(1)},\ldots,x_{g(N)})$.
The Pauli exclusion principle states 
that the physically allowed wave functions are described by functions satisfying
\begin{equation*}  
  \Psi(g(x))=\sgn(g) g(\Psi(x))\,.
\end{equation*}

\subsection{Vasy's pseudodifferential calculus and Georgescu's algebra}\label{ssec.psdo}
We will now return to the more general setting 
of a general semilattice $\maF$ of linear subspaces of a finite-dimensional vector space $X$, 
using the notation of \eqref{eq.def.SSmaF}. Recall that we are assuming,
for convenience, that $\{0\} \in \maF$ and $X \notin \maF$ (this is no loss of
generality, since our argument works in general but is just a little bit
more involved; moreover, the general case can be reduced to this one).

The action of~$X$ by translation on 
\begin{equation}\label{eq.def.Omega}
    \XGV = [\oX: \SS_{\maF}]
\end{equation}
(see Remark~\ref{rem.Nbody} and Theorem~\ref{thm.cor.main1})
can be used to define Georgescu's algebra and (possibly) Vasy's 
pseudodifferential calculus, along the lines of Georgescu's method \cite{Georgescu2018, GeIf06}
(see also \cite{aln2}). Let us outline this construction and derive
some consequences.

Let $\sS(X)$ denote the Schwartz space of smooth, rapidly decreasing functions on $X$. 
Any $f \in \sS(X)$ gives rise to a  convolution operator $f(T):L^2(X)\to L^2(X)$, $h\mapsto f*h$. 
In the notation $T:X\to \maL(L^2(X))$, $q\mapsto T_q$ stands for the translation operator 
$T_qf(x)=f(x+q)$ as, for instance, in \cite{GeIf06, GN}. In fact, much more general functions~$f$ 
can be allowed here, such as function whose Fourier transform is a classical symbol. 
Similarly, a function  $g\in\maC(\XGV)$ gives rise to a multiplication operator $M_g$ on 
$L^2(X)$. By results of Georgescu \cite{GeIf06} (using also Theorem~\ref{thm.main2}),
Georgescu's algebra $\maC(\XGV) \rtimes X$
is the norm closure of the algebra generated by operators of
the form $M_gf(T)$ acting on $L^2(X)$. So, if $\maL(\maH)$ denotes
the algebra of bounded operators on a Hilbert space~$\maH$, then we obtain
$\maC(\XGV) \rtimes X \subset \maL(L^2(X))$. The resulting subalgebra 
$\maC(\XGV) \rtimes X$, called crossed product, is norm closed and
closed under adjoints, hence is a $C^*$-algebra.
See \cite{GeIf06, Georgescu2018} for the details on the link between
the crossed product by $\mathbb{R}^n$ of such a commutative $C^*$-algebra
and operators of the form $M_gf(T)$. 

If $f$ is such that its Fourier transform is a classical symbol 
of order $m$ on $X$, then $P := M_g f(T)$ is a pseudodifferential 
operator. Classically then, its distribution kernel $k_P \in \maD^\prime$ is  a 
classical conormal distribution in $I^m(X \times X; X)$, with $X$ 
diagonally embedded in ${X \times X}$. (See \cite{Hormander3} for the definition
of (classical) conormal distributions.)
The map $(x_1, x_2) \to x_1 - x_2$ extends then to a smooth map of pairs
$(X \times X, X) \to (\XGV \times X; \XGV)$, with 
the embedding $\XGV \simeq \XGV \times \{0\} \subset \XGV \times X$.
This embedding sends~$k_P$ to $g \otimes f$, and hence $k_P$
can be identified with a classical conormal distribution 
in $I^m(\XGV \times X ; \XGV)$ (this is a particular case of the 
construction in \cite{aln2}).
Let $I_c^m(\XGV \times X ; \XGV)$ be the set of such distributions with
compact support. Then it follows that 
\begin{equation*}
   \Psi_c^\infty(\XGV)  \ede I_c^\infty(\XGV \times X; \XGV) 
   \ede  \bigcup_{m \in \ZZ} I_c^m(\XGV \times X ; \XGV)
\end{equation*}
is a filtered algebra acting by convolution on (suitable) functions $X \to \CC$ 
as an algebra of pseudodifferential operators \cite{aln2}. (Recall that we are
considering only classical conormal distributions and the index ``$c$''
comes from ``compact support.'' Vasy's $N$-body calculus $\Psi_{N}^\infty(X)$ is
certainly bigger and better than $I_c^\infty(\XGV \times X; \XGV)$
in the sense that it contains the
resolvents of its $L^2$-invertible operators. Let
\begin{equation*}
   I_c^{- \infty}(\XGV \times X; \XGV) \subset \sS(X) \otimes_{\pi} 
   \CI(\XGV) \subset I^{-\infty}(\XGV \times X; \XGV)
\end{equation*}
be the projective tensor product. We have good reasons to believe
and hence we conjecture that Vasy's $N$-body calculus can be identified with
\begin{equation}\label{eq.Vasy}
   \Psi_{NB}^\infty(\XGV) 
   \ede I_c^\infty(\XGV \times X; \XGV) + 
   \sS(X) \otimes_{\pi} \CI(\XGV) \,.
\end{equation}
We need to include $\sS(X) \otimes_{\pi} \CI(\XGV) := \sS(X ; \CI(\XGV))$ on the 
right hand side to accomodate operators of the form $M_g f(T)$ with $f \in \sS(X)$
{\em with non-compact support,} since $M_g f(T) \in I_c^m(\XGV \times X; \XGV)$
if, and only if, $f$ is compactly supported (recall that $\hat f$ is a classical
symbol of order $m$).

\begin{proposition} \label{prop.sp.inv}
We define
\begin{equation*}
     \Psi_{NB}^n(\XGV) 
      \ede I^n_c(\XGV \times X; \XGV) + 
     \sS(X) \otimes_{\pi} \CI(\XGV)\,.
\end{equation*}
The space $\Psi_{NB}^\infty(\XGV) \ede \bigcup_k \Psi_{NB}^k(\XGV)$
is a filtered algebra that is closed under holomorphic functional calculus.
Let $D$  be a strongly elliptic differential operator of order $m > 0$, with constant 
coefficients and $v_Y \in \CI(\oXY)$. Then
$H_N^\prime := D + \sum_{Y \in \maF} v_Y \in \Psi_{NB}^m(\XGV).$
Consequently, for all $\lambda \notin \Spec(H_N^\prime)$, we have 
\begin{equation*}
    (H_N^\prime - \lambda)^{-1} \in \Psi_{NB}^{-m}(\XGV) 
    \ede I_c^{-m}(\XGV \times X; \XGV) + 
    \sS(X) \otimes_{\pi} \CI(\XGV)\,.
\end{equation*}
In the case  $X := \RR^{3N}$ and $\maF$ ans in the $N$-body problem, the action of the symmetric group $S_N$ on $X$
induces an order-preserving automorphism of the algebra $\Psi_{NB}^\infty(\XGV)$.
\end{proposition}

\begin{proof}
(Sketch) There are two main things to prove here: first, that the convolution
product makes 
$\sS(X) \otimes_{\pi} \CI(\XGV)  := \sS(X ; \CI(\XGV))$ an algebra
and, second, that it is stable for holomorphic functional calculus
(equivalently in this case, that the algebra with adjoint unit contatins the resolvents of its 
$L^2$-invertible elements). The first question is answered by noticing that
the action of~$X$ on $\CI(\oXY)$ is with polynomial growth (this is quite
unusual for the action of~$X$ on a manifold!) and hence it is again with
polynomial growth on $\XGV$ in view of our Theorem~\ref{thm.main2}. The
second question is answered by using the results of \cite{LMN2} as follows.
We consider three families of operators on $L^2(\XGV \times X)$, 
possible unbounded (so not defined everywhere). Let $A_1$ be the set of 
differential operators on $\XGV$, let~$A_2$ the set of multiplication
operators with polynomials on $X$ and, finally, let~$A_3$ be the set of constant coefficients differential operators on $X$.
Then 
\begin{equation*}
    \sS(X) \otimes_{\pi} \CI(\XGV) \ede \left\{\, f \in \maC(\XGV) 
    \rtimes X \,\Big|\, \bigl[[f, P_1], P_2\bigr] P_3 \mbox{ is bounded}\, 
    P_j \in A_j 
     \right\} \,.
\end{equation*}
The results of \cite{LMN2}, especially Theorems 2 and 3, then give that
$\sS(X) \otimes_{\pi} \CI(\XGV)$ is spectrally invariant (\ie\ stable
under holomorphic functional calculus).
\end{proof}

We ignore if one can replace in the resolvent estimate of the last proposition
$H_N^\prime$ with $H_N$ or with $H_{N-1}^{\eff}$,
which are not in $\Psi_{NB}^\infty(\XGV)$,
since these operators allow for Coulomb singularities in the potential. This brings us to 
one of our main reasons for considering Georgescu's algebras $\mEF X \rtimes X$ 
instead of a pseudodifferential calculus (and one of the reasons why
we may need to take norm closures), namely, that Georgescu's algebra does 
not suffer from this deficiency, and, in fact, one has
\begin{equation}\label{eq.res.G}
   (H_N - \lambda)^{-1} \in  \mEF X \rtimes X  \,, \quad \lambda \notin \Spec(H_N)\,.
\end{equation}
\cite{GeorgescuBookNew, GeIf06, Georgescu2018} (this is a consequence of Hardy's 
inequality and is explained also in \cite{GN}). Of course, the above proposition 
provides a much more precise result, when applicable, but is also much more difficult 
to prove than the relation of Equation \eqref{eq.res.G}. Let us notice, moreover,
that $\mEF X \rtimes X$ is the norm closure of 
$\Psi_{NB}^{-1}(\XGV)$
in $\maL(L^2(X))$, the algebra of bounded operator on $L^2(X)$.

\subsection{Connections to the HVZ theorem} \label{ssec.HVZ}
The algebras considered in the previous
subsection were introduced, in part, in order to obtain conceptual proofs 
and extensions of the classical HVZ theorem, named after Hunziker, van Winter, and Zhislin,
describing the essential spectrum
of $N$-body Hamiltonians $H$ \cite{DerGer2, DerGer1, Georgescu2018, MantoiuJOT, Teschlbook}. 
It is well-known that the operators $H$ considered here are self-adjoint. We have that 
$\lambda$ {\em is not} in the essential spectrum of $H$ if,
and only if, $H - \lambda$ is Fredholm. 
Our next application is of a conceptual nature on how to relate the
HVZ theorem with other classical Fredholm results in PDE theory. 
There exist many refinements of the HVZ theorem,
in the simple setting of an atom with $N$-electrons, we refer to 
\cite[Section 11, Theorem 11.2]{Teschlbook}; for a more general
version see \cite[Theorem XIII.12]{ReedSimon4}.
The HVZ theorem determines the essential spectrum of the Hamiltonian $H_N$ in
terms of other, simpler Hamiltonians $H_{N_\alpha}$, where $\alpha$
ranges over a certain index set. The operators $H_{N_\alpha}$ are usually
called ``limit operators,'' and can indeed be obtained as
strong limits of translations of $H_N$.
Very powerful
generalizations of the HVZ theorem were obtained by Georgescu
(using $C^*$-algebras \cite{GeorgescuBookNew, Georgescu2018, GN})
and by many other authors -- more on this below.

Nowadays, there are many 
results telling us when (pseudo)differential operators on non-compact
or singular spaces are Fredholm, and typically they are also in terms of
certain ``generalized limit operators'', the terminologies ``indicial
operator'' or ``normal symbol'' are also used  by Melrose and Schulze independently. 
We refer to \cite{MelroseVSSchulze} for a overview and comparison between these two approach.
Results of this type go back at least to Kondratiev's 1967 celebrated 
paper \cite{Kondratiev67}. Some of the strongest current results
are based on groupoids, see \cite{CCQ, CNQ, MaNi} and
the references therein. We also refer to
\cite{ComeFredholm, MantoiuJOT} for the case when the groupoid is obtained from the action
of a group on a space, as it is our case in this paper. The results
are in terms of orbits, their isotropies, and the induced operators.
In fact, each of these induced operators, referred to as ``a generalized limit
operator'' above, acts on the product of the corresponding orbit 
with the corresponding isotropy group and is invariant with respect
to that group.

A natural question is to reconcile the classical results on $H_N$ using limit
operators with the classical PDEs results based on ``generalized limit
operators''. This is almost done by \cite{VasyAsympt}, 
except that it is not clear whether the resolvents of $H_N$
belong to Vasy's pseudodifferential calculus. (We do know, however, that
the resolvents of $H_N$ belong to
the norm closure of the pseudodifferential calculus introduced in the
previous subsection, as discussed in the previous subsection.
Note that here we are using our Theorem~\ref{thm.main2}.)

The Fredholm results just mentioned, do apply, however, also to the norm
closure of the corresponding pseudodifferential calculi, and hence, 
in principle, the HVZ theorem could then be obtained from the structure of the orbits
of the action of~$X$ on $\XGV := [\oX: \mathbb{S}_\mathcal{F}]$ and their isotropies (both the orbits
and the isotropies are linear subspaces of~$X$) and the explicit form of
the generalized limit operators. 
However, in order not to increase too much the length 
of this paper, we leave this for a future publication. 
Nevertheless, it is interesting to point out  
that this approach has the potential to provide Fredholm conditions for
the restrictions of $H_N$ and its variants to the isotypical components
of the action of $S_N$ or some subgroup of $S_N$. Results in this direction
(for operators on compact manifolds) were recently obtained in 
\cite{BCLN1, BCLN2, Baldare3}.
See also \cite{ComeFredholm, DLR, GN, MougelH, MNP, MaNi} for related results.

\subsection{A regularity result for bound states}\label{subsec.regres} 
\label{ssec.reg}
An application of our results to regularity for bound states for Schr\"odinger
operators with inverse square potentials is contained in our recent
preprint \cite{AMN2}. Here we just quickly explain the result. Let
\begin{equation}\label{eq.def.maS}
   \overline{\maF} \ede \left\{\, \overline{Y} \Bigm |
  Y \in \maF  \right\}\,,
\end{equation}
which is a clean semilattice that we endow with an admissible order. Let also
\begin{equation}
   X_{\maF} \ede \big[\XGV: \overline{\maF}\big] \seq \big[\oX :\mathbb{S}_\mathcal{F} 
   \cup \overline{\maF}\big] \,.
\end{equation}
For instance, if $\mathcal{F} = \big\{\{0\},Y_1,Y_2\big\}$ with $Y_1 \subset Y_2$, then
$X_\mathcal{F} = \big[\oX: \mathbb{S}_{Y_1},\mathbb{S}_{Y_2}, \{0\}, \overline{Y}_1, 
\overline{Y}_2\big]$. Our results then show that 
$\big[\oX:\mathbb{S}_{Y_1},\mathbb{S}_{Y_2}, \{0\}, \overline{Y}_1, \overline{Y}_2\big] 
\simeq \big[\oX: \{0\}, \mathbb{S}_{Y_1}, \overline{Y}_1,
\mathbb{S}_{Y_2},\overline{Y}_2\big]$. See also \cite{Kottke-Lin}.

For each $Y \in \maF$, let $a_Y, b_Y \in \CI(X_{\maF})$ and let $d_Y$ denote the
distance to $Y$ in some fixed euclidean metric on $X$. Let also $c \in \CI(X_{\maF})$.
A function of the form 
\begin{equation}  \label{eq.def.genV}
    V(x) \ede \sum_{Y \in \maF} ( a_Y(x) d_Y(x)^{-2} 
    + b_Y(x) d_Y(x)^{-1} ) + c(x)
\end{equation}
will be called an {\em inverse square potential (associated to the semilattice $\maF$)}. 
Let 
\begin{equation}\label{eq.def.rho}
   \rho(x) \ede \min\big\{\dist(x,\unionF), 1\big\}\,,
\end{equation}
where $\dist(x, \unionF)$ is the distance to $x$ to $\maF$ in
some euclidean metric on~$X$.
The following result (which combines techniques of this paper with those in \cite{ACN})
was proved in \cite{AMN2}.

\begin{theorem}\label{theorem.main.reg} 
Let $D$ be a second order strongly elliptic operator with constant coefficients.
Let $V$ be an inverse square potential associated to the semilattice
$\maF$ of linear subspaces of the euclidean space $X$
see Equation~\eqref{eq.def.genV}, $\rho(x) := 
\min\{\dist(x, \unionF), 1\}$, and assume $u \in L^2(X)$ is an 
eigenfunction of $D + V$, that is $(D + V)u = \lambda u$ on 
$X\smallsetminus \unionF$ for some $\lambda \in \CC$, 
then, for all multi-indices $\alpha$, we have
\begin{equation*}
    \rho^{|\alpha|} \pa^\alpha u \in L^2(X)\,.
\end{equation*}
\end{theorem}

Our theorem covers, of course, the case of the operators
$H_N$ and $H_{N-1}^{\eff}$ of Equations~\eqref{eq.def.HN} and~\eqref{eq.def.H1eff} (for $D = - \Delta$).
Also, we note that, since $V$ is not assumed to be real valued, we do
not necessarily have $\lambda \in \RR$.
The regularity of bound states and, in general, the geometry of the 
Georgescu--Vasy space $\XGV = [X: \SS_{\maF}]$
may be useful for approximation 
purposes. In fact, another, related motivation of our work is the
approximation of the isolated eigenfunctions of $N$-body Hamiltonians
using the Finite Element Method. The role of the Georgescu--Vasy space~$\XGV$ 
here is to provide a good underlying support for the construction of the
approximation spaces. This is very tentative yet, but see 
\cite{Flad, Griebel1, HLiN, Yserentant} for some results in this 
direction, including more references. 

A first motivation for inverse square potentials comes
from relativistic physics, where operators of the form ``Dirac plus Coulomb
potential'' are used. The square of these operators will be an operator with
inverse square potentials of the type covered by Theorem~\ref{theorem.main.reg}. 
See also \cite{HLiN} and, especially, the recent paper by Derezi\'{n}ski and Richard 
\cite{DerezinskiRichard} and the references therein for
further physical motivation for inverse square potentials.


\appendix

\section{Proper maps} 
We now provide a characterization of proper maps used in the
main body of the paper.
Let $f : X \to Y$ be a continuous map between two Hausdorff
spaces. Recall that $f$ is called {\em proper} if $f^{-1}(K)$ is
compact for every compact subset $K \subset Y$.

\begin{lemma}[Generalizes {\cite[Prop 4.32]{Lee-Top}}]
\label{lemma.prop.closed}
Let $f : X \to Y$ be a continuous map between two Hausdorff spaces
with $Y$ locally compact. If $f$ is proper, then $f$ is closed.
\end{lemma}

In \cite[Prop 4.32]{Lee-Top} the lemma is stated with the additional
requirement that $X$ be  locally compact. However in the proof the
locally compactness of~$X$ is not needed. We omit the proof
since we will apply the lemma only when $X$ is locally compact.

\begin{corollary}\label{cor.prop.homeo}
Let $f : X \to Y$ be a continuous injective map between two Hausdorff
spaces with $Y$ locally compact. If $f$ is proper, then $f$ is a
homeomorphism onto its image.
\end{corollary}

\begin{proof}
The map $f:X\to f(X)$ is bijective continuous and closed and thus a
homeomorphism.
\end{proof}

We shall say that $f$ is {\em locally proper} if, for every $y \in Y$,
there exists an open neighborhood $V_y$ of $y$ in $Y$ such that the
map $f^{-1}(V_y) \to V_y$ induced by $f$ is proper.

\begin{lemma}\label{lemma.loc.prop}
Let $f : X \to Y$ be a continuous map between two Hausdorff spaces
with $Y$ locally compact. Then $f$ is proper if, and only if, it is
locally proper.
\end{lemma}

\begin{proof}
Clearly, every proper map is locally proper, by definition. Let us
assume that $f$ is locally proper and let $K \subset Y$ be a compact
subset. For any $y\in K$ we choose the open neighborhood $V_y$ as
in the definition of a locally proper map. As $Y$ is locally
compact, there is an open neighborhood $W_y$ of $y$ in $V_y$ such that
its closure $\overline{W}_y$ in $Y$ is a compact subset of $V_y$.
The local properness of $f$ together with the choice of $V_y$ implies that
$f^{-1}(\overline{W}_y\cap K)$ is compact. By the compactness of $K$
we can choose $y_1,\ldots,y_N$ such that $K$ is covered by
$\left(W_{y_j}\right)_{1\leq j\leq N}$. Then $K=\bigcup_{j=1}^N
\left(\overline{W}_{y_j}\cap K\right)$. Hence
\begin{eqnarray*}
  f^{-1}(K)&=& \bigcup_{j=1}^N f^{-1}(\overline{W}_{y_j}\cap K)
\end{eqnarray*}
is also compact. This completes the proof.
\end{proof}

\section{More on submanifolds of manifolds with corners}
\label{sec.appendix.submanifold}

We discuss here a few other notions of submanifolds and the
relation to our concept of \wsbmanifold. While this is not needed for
the proof of the main result,  we hope the interested reader will 
find this material useful.

\subsection{Submanifolds in Melrose's sense}\label{ssec.appendix.submanifold.1}
We begin with Melrose's concept of a submanifold in a manifold with corners, 
  following \cite[Definition~1.7.3]{MelroseBook}.

  
\begin{definition}
\label{def.submanifold-gen}
A subset $S$ of a manifold with corners $M$ 
of dimension $n$ is a \emph{submanifold (in the sense of manifolds with corners)}
if, for every $p \in S$, there exists $0 \le k \le n$ and a
(corner) chart $\phi:U\to \Omega\subset 
\RR_k^n := [0,\infty)^k \times \RR^{n-k}
$, numbers $n'\leq n$
and $k'\leq n'$, and a matrix $G\in \GL(n,\RR)$ such that
\begin{enumerate}
\item $p \in U$ 
\item $G \left(\RR^{n'}_{k'}\times \{0\}\right) \subset \RR^n_k\,.$ 
\item $\phi(S \cap U)\seq G \left(\RR^{n'}_{k'}\times \{0\}\right)
  \cap \Omega \,.$
\end{enumerate}
\end{definition}

Obviously, every submanifold  in the sense of manifolds with corners is 
a \wsbmanifold{}, see Definition~\ref{def.weak-submanifold}. In \cite{MelroseBook} 
all submanifolds are submanifolds   in the sense of Definition~\ref{def.weak-submanifold}, 
see \eg Lemma~\ref{lemma.krit.weak.submanifold}.   In Remarks~\ref{rem.sbmdf-corner.mfd-corners} 
\eqref{rem.sbmdf-corner.mfd-corners.i} we explained that any \wsbmanifold{} 
of a manifold with corners inherits an atlas, and thus this also applies to 
submanifolds in the above sense. However, it can be shown \cite{koenig.master}
that many submanifolds in our article are not submanifolds in the sense of manifolds 
with corners, but only \wsbmanifolds{}, as defined in Definition~\ref{def.weak-submanifold}. 
In Example~\ref{ex.weak.submanifold} we provide an example of a \wsbmanifold{} of 
$\RR^2_1$ that is not one in the sense of Definition~\ref{def.submanifold-gen}.

\begin{example}[Diagonal]\label{example.diagonal}
Let $N$ be a manifold with corners. Then $M:=N\times N$ is also a
manifold with corners. Consider the diagonal $\Delta_N:=\{(p,p)\in
M\mid p\in N\}$. Then $\Delta_N$ is a submanifold of~$M$ 
in the sense of manifolds with corners.
\end{example}

The following provides examples of \wsbmanifolds{} that are 
not \Msbmanifolds.

\begin{examples}\label{ex.weak.submanifold}\
\begin{enumerate}
\item\label{ex.weak.submanifold.i} The function $f : \RR^2_1 := [0, \infty) \times \RR \to \RR^2_1$,
$f(x,y):=(x+y^2,y)$, is an injective immersion. It is a homeomorphism onto 
its image $S:=f\bigl(\RR^2_1\bigr)$. However, it can be easily seen that $S$ 
is not a submanifold of $\RR^2_1$ in the sense of manifolds with corners. 
On the other hand $S$ is a submanifold of $\RR^2$ in the sense of manifolds with corners.
\item\label{ex.weak.submanifold.ii}The function $f(x):=(x,x^2)$ defines a injective immersion
$\RR^1_1\to \RR^2_2$. It is a homeomorphism onto its image $S:=f\bigl(\RR^1_1\bigr)$. However, 
$S$ is not a submanifold $\RR^2_2$ in the sense of manifolds with corners. On the other hand 
$S$ is a submanifold of $\RR^2_1$ and of $\RR^2$ in the sense of manifolds with corners.
\end{enumerate}
\end{examples}

In our article injective immersion which are a homeomorphism to its image, 
play an important role. Recall the following classical fact for manifolds 
$N$ and $M$ \emph{without boundary and without corners}:
\begin{equation*}
  (*) \begin{cases}
    \;\text{If $f:N\to M$ is an injective immersions, then $f(N)$ is a submanifold}\hfill\\
    
    \;\text{if, and only if, $f$ maps $N$ homeomorphically to $f(N)$.}\hfill
  \end{cases}
\end{equation*}
Examples~\ref{ex.weak.submanifold} show that $(*)$ does no longer hold if $M$ and 
$N$ are manifolds with corners and if we understand the word ``submanifold'' in the sense of Definition~\ref{def.submanifold-gen}.
On the other, we proved in Proposition~\ref{prop.smfd.crit} that $(*)$ 
holds for manifolds with corners, if we replace ``a submanifold'' by ``a weak submanifold.''

\subsection{Other classes of submanifolds}\label{ssec.appendix.submanifold.2}
For comparison and completness, we recall now the definitions of some further classes
of submanifolds. The reason the reader might be interested in these concepts is
that the concept of a \Msbmanifold{} 
seems to be too unspecific and the concept of a \psbmanifold\ (Definition~\ref{def.psubmanifold})
seems to be sometimes too restrictive. A first alternative is the
concept of a ``\wibsbmanifold{}'', where ``wib'' stands for a submanifold 
{\em without an interior boundary.} 

\begin{definition}
A submanifold $S\subset M$ is called a \emph{\wibsbmanifold{}} or a
\emph{submanifold without interior boundary} if it can be defined
locally in suitable charts as the kernel of a linear function. More
precisely: $S\subset M$ is a \emph{\wibsbmanifold{}} if, for every $x
\in S$, there exists a (corner) chart $\phi:U\to \Omega\subset
\RR^n_k$, and a linear subspace $L$ of $\RR^n$, such that
\begin{enumerate}
\item $x \in U$ and
\item $\phi(S \cap U)\seq L \cap \Omega \,.$
\end{enumerate}
\end{definition}

If $G\in \GL(n,\RR)$ is as in Definition~\ref{def.submanifold-gen}, 
then we necessarily have $L = G \left(\RR^{n'}\times \{0\}\right)$. 
If $x\in S\cap U$, then $n':=\dim (L)$ is the dimension of $S$ in $x$ defined above.
Obviously all \psbmanifolds{} are wib-manifolds, which can be easily
seen by defining the $L$ in the definition above as the linear
extension of $L_I$ in Definition~\ref{def.psubmanifold}.

\begin{remark}
In the above definition, we explicitly required $S$ to be a
submanifold. To justify this requirement, we will give an example of a
closed subset $S\subset M$ that is not a submanifold, but fulfills
all other requirements of the definition of a \wibsbmanifold{}.
Indeed, let
\begin{equation*}
   K \ede \{(x_1,x_2,x_3)\in \RR^3\mid x_1\geq 0,\; x_2\geq
   0,\;x_1\leq x_3,\;x_2\leq x_3\}\,,
\end{equation*}
which is a cone over a square. The map $f:\RR^3\to
\RR^4$,\ $f(x_1,x_2,x_3)= (x_1,x_2,x_3-x_1,x_3-x_2)$ has the property
$f^{-1}(\RR^4_4)=K$. Then for $\phi=\id$, $x=0$, and $L:=f(\RR^3)$ all
requirements of the definition are satisfied, but $S:=f(K)$ is not a
submanifold of $\RR^4_4$. It it were a submanifold, then its dimension
would have to be $3$, and then any boundary point of $S$ is in at most
$3$ closed boundary hyperfaces. But $0\in S$ is in $4$ closed boundary
hyperfaces of $S$.
\end{remark}

\begin{remark}\label{rem.further.submfd}
Note that Melrose also introduces the notions d-submanifold
\cite[Def.~1.7.4]{MelroseBook} and \bsbmanifold{}
\cite[Def.~1.12.9]{MelroseBook}, whose definitions will not be
recalled here. They satisfy
\begin{eqnarray*}
S \text{ is a \psbmanifold} \;\Longrightarrow\;S \text{ is a
  d-submanifold} &\Longrightarrow& S \text{ is a \bsbmanifold{}}
\end{eqnarray*}
\begin{eqnarray*}
 &\Longrightarrow& S \text{ is a submanifold}\;\Longrightarrow\;S
  \text{ is a \wsbmanifold}\,.
\end{eqnarray*}
However there are wib-manifolds that are not \bsbmanifolds{},
such as Melrose's example of the submanifold $\{x_3=x_1+x_2\}\in
\RR^3_3$. There are d-manifolds that are no wib-manifolds,
for instance, $\RR^1_1=[0,\infty)\subset \RR$ or any surface with boundary in
  $\RR^3$. However all \psbmanifolds{} introduced below are
  \dsbmanifolds{} and \wibsbmanifolds{}.
Melrose shows that the diagonal $\Delta_N$ is a \bsbmanifold{} of
$N\times N$, but in general not a \dsbmanifold{}. It follows that
$\Delta_N$ is not a \psbmanifold{}.
\end{remark}

\begin{remark}\label{remark.tame}
Let us remark that the concept of a \emph{tame} submanifold
considered in \cite[Sec.~2.3]{ammann.ionescu.nistor:06} is a concept of a
submanifold in an essentially different sense, it is actually a more
restrictive notion of submanifold than the ones encountered in this
paper. All notions of submanifolds discussed so far
involve properties that may or may not hold for a subset $N$ of a
manifold with corners $M$. In contrast to this, tame submanifolds in
\cite[Sec.~2.3]{ammann.ionescu.nistor:06} are submanifolds of a \emph{Lie manifold}
$(M,A)$, where $M$ is a manifold with corners and $A$ is a Lie
algebroid on~$M$ with some compatibility conditions. Whether
a subset~$N$ of~$M$ is a tame submanifold of $(M,A)$ or not depends
also on the Lie algebroid~$A$. In any case, a tame submanifold
  will have a tubular neighborhood in the strongest sense.
Similar remarks apply to the $A(\mathcal{G})$-tame submanifolds
considered in \cite{nistorDesing}.
\end{remark}


\def\cprime{$'$}

\end{document}